\documentclass[11pt]{article} 

\usepackage{amssymb,xspace,graphicx,relsize,bm,xcolor,amsmath,breqn,algorithm,algpseudocode,multirow}
\usepackage[pagebackref]{hyperref}
\usepackage[usestackEOL]{stackengine}
\usepackage[amsmath,thmmarks]{ntheorem}

\usepackage{cleveref}
\usepackage{thm-restate}
\usepackage{array}
\usepackage{parskip}

\newcommand{\BInf}{\operatorname{BInf}}
\newcommand{\sm}{\mathrm{sm}}
\newcommand{\cost}{\mathrm{cost}}

\usepackage{url}
\def\01{\{0,1\}}
\newcommand{\ceil}[1]{\lceil{#1}\rceil}
\newcommand{\floor}[1]{\lfloor{#1}\rfloor}
\newcommand{\eps}{\varepsilon}
\newcommand{\ket}[1]{|#1\rangle}
\newcommand{\bra}[1]{\langle#1|}

\newcommand{\occ}{\operatorname{occ}}
\newcommand{\Prb}{\mathbf{Pr}}
\newcommand{\cR}{\mathsf{R}}
\newcommand{\cQ}{\mathsf{Q}}
\newcommand{\one}{\mathbf{1}}

\newcommand{\E}{\mathbb{E}}
\newcommand{\Exp}{\mathbb{E}}

\newcommand{\F}{\ensuremath{\mathbb{F}}}

\newcommand{\A}{\ensuremath{\mathcal{A}}}

\newcommand{\R}{\ensuremath{\mathbb{R}}}

\usepackage{ntheorem}

\newcommand{\dist}{\operatorname{dist}}
\newcommand{\Bin}{\operatorname{Bin}}

\newcommand{\ed}{\ensuremath{\mathsf{ed}}}

\DeclareMathOperator{\poly}{poly}

\DeclareMathOperator{\polylog}{polylog}

\newtheorem{theorem}{Theorem}[section]
\newtheorem{definition}[theorem]{Definition}

\newtheorem{fact}[theorem]{Fact}
\newtheorem{remark}[theorem]{Remark}
\newtheorem{lemma}[theorem]{Lemma}

\newtheorem{corollary}[theorem]{Corollary}

\newtheorem{claim}[theorem]{Claim}

\usepackage{ifdraft}

\newcommand{\alg}{\mathsf{Alg}}

\newenvironment{proof}{\par\noindent{\bf Proof.}\quad}{\hfill  $\Box$}
\newenvironment{proofof}[1]{\par\noindent{\bf Proof of #1.}\quad}{\hfill  $\Box$}

\usepackage[margin=1in]{geometry}
\hypersetup{
	colorlinks,
	linkcolor={magenta!100!black},
	citecolor={blue!100!black},
}
\usepackage{tcolorbox}

\renewcommand\thmcontinues[1]{Formal Statement}

\usepackage{footnote}

\newcommand{\Simon}{\ensuremath{\mathsf{Simon}}}

\makeatletter
\def\thm@space@setup{\thm@preskip=\parskip \thm@postskip=0pt}
\usepackage{algpseudocode}

  \title{Quantum Query Complexity Beyond the Worst Case}

\author{
Srinivasan Arunachalam\\[2mm]
IBM Research\\
\small Silicon Valley Lab\\
\small \texttt{Srinivasan.Arunachalam@ibm.com}
\and
Yanlin Chen\\[2mm]
 QuICS\\ \small University of Maryland  \\
\small \texttt{yanlin@umd.edu}
\and
Amin Shiraz Gilani\\[2mm]
 QuICS\\ \small University of Maryland  \\
\small \texttt{asgilani@umd.edu}
}

\begin{document}
\date{}
\maketitle

\begin{abstract}
Smoothed analysis is a central framework in classical algorithms for explaining performance of algorithms beyond the worst case, often explaining why algorithms perform well in practice. We initiate a systematic study of its quantum counterpart and show the following:
\vspace{2mm}
\begin{enumerate}
    \item We show that there is a total function whose smoothed quantum query complexity is exponentially smaller than its classical query complexity.
    \vspace{1mm}
    \item We give near-tight characterizations of smoothed randomized and quantum query complexities for symmetric Boolean functions. This unifies the worst-case complexity results of [Beals et al, FOCS'98] and average-case complexity results of [Ambainis and de Wolf, STACS'00] for such functions. 
    \vspace{1mm}
    \item We study two variants of the pattern matching problem, one where the pattern can be accessed through queries and one where it is given as part of the input. {In both settings we show that the smoothed quantum query complexity can be \emph{superpolynomially} better than the worst-case quantum query complexity, if the pattern and text are of comparable length.} Furthermore, in the former setting, we show that its smoothed quantum complexity can be superpolynomially smaller than the corresponding classical complexity. In the latter setting, we prove \emph{optimal} bounds on the randomized and quantum smoothed complexity, proving a quadratic quantum speed-up.
    \vspace{1mm}
   \item We give new quantum algorithms for worst-case Ulam distance, yielding polynomially faster quantum algorithms for smoothed edit distance. Our main technical ingredients include a near-tight quantum algorithm for $\varepsilon$-approximating the number of collisions between two non-repetitive strings, improving the result of Le Gall and Ng [QIC'22].\end{enumerate}
\vspace{2mm}
Together, our results show that smoothing can reveal larger quantum speedups than worst-case analysis suggests, opening a path towards quantum advantage on more realistic~inputs. 
\end{abstract}


\newpage 
{\small \tableofcontents}
\newpage 
\section{Introduction}
\paragraph{Classical motivation.} Spielman and Teng~\cite{spielman2004smoothed} introduced the \emph{smoothed complexity model} to understand a longstanding puzzle surrounding the well-known simplex algorithm: despite its exponential worst-case complexity, it is remarkably efficient in practice. To understand this discrepancy, they observed that the hard instances underlying the exponential lower bound formed a highly fragile set, and that even small random perturbations of these instances caused the simplex algorithm to become efficient. To understand this,~\cite{spielman2004smoothed} introduced the
\emph{smoothed complexity} framework: Let $P$ be a problem over
$x\in \R^d$, let $\alg$ be an algorithm for~$P$. Denote by $C(\alg,x)$ the
complexity of $\alg$ on input $x$. In the worst-case setting, the complexity
of solving $P$ using $\alg$ is
\[
    \max_{x\in \R^d} [C(\alg,x)].
\]
Smoothed complexity instead measures
\[
    \max_{x\in \R^d}
    \mathop{\Exp}_{g\sim \mathcal{G}(0,\sigma)^d}
    \bigl[C(\alg,x+g)\bigr],
\]
where $\mathcal{G}(0,\sigma)$ denotes a mean-zero Gaussian distribution with
variance $\sigma$. Thus, an adversary is still free to choose the worst possible
input $x$, but the algorithm is evaluated only after $x$ has been subjected to a
small random perturbation. When $\sigma=0$, this recovers worst-case complexity;
as the magnitude of the perturbation increases, the model moves toward the
average-case regime. Hence, smoothed complexity interpolates between
worst-case and average-case analysis. 
Spielman and Teng called an algorithm
\emph{smoothed efficient} if its smoothed complexity is bounded by
$\poly(n,1/\sigma)$. Their seminal result showed that the simplex algorithm,
despite its exponential worst-case complexity, has polynomial smoothed
complexity. Since then there have been a sequence of works that improved the dependencies of the algorithm~\cite{spielman2004smoothed,dadush2018friendly}. For this result, they won the Godel Prize in 2008. Since their introduction, several works have been analyzed in the smoothed model and shown to have polynomial-time smoothed complexity much more efficient than the worst case complexity of the problem~\cite{chen2020polynomial,chen2022near,sankar2006smoothed,roglin2007integer,kalai2009learning,roglin2009multiobjective,arthur2011kmeans,etscheid2014maxcut,bhaskara2014tensor,ge2015gaussian}.

\textbf{Quantum motivation.} This classical motivation suggests a different perspective on the \emph{search for quantum advantage}. Rather than focusing exclusively on worst-case complexity, one can ask whether quantum speedups emerge under more realistic input models. Smoothed complexity provides precisely such a lens, comparing quantum and classical algorithms on arbitrary instances subjected to small random perturbations. There is, however, an important distinction between the classical and quantum settings. Classical smoothed analysis was developed only after decades of empirical evidence showed that worst-case complexity could poorly predict the behavior of algorithms such as simplex. In quantum computing, we need not wait for large-scale fault-tolerant devices---and decades of experience with them---before asking the analogous question. With such devices potentially arriving in the coming years, and with new settings for quantum speedups increasingly difficult to identify, this motivates the systematic study of \emph{quantum algorithms in the smoothed complexity model}.

\subsection{Main results} As far as we are aware, we initiate the first systematic study of this topic. Our results provide several first steps towards understanding quantum query algorithms in the smoothed model. We obtain both general query characterizations of symmetric functions, total function quantum-classical separations  and concrete quantum speedups for string problems in the smoothed model. In the next three subsections we summarize our contributions and give high-level proof ideas.

\paragraph{\emph{1.} Total function separation.}
One of the canonical exponential separations between quantum and classical query complexity is Simon's problem~\cite{Simon} on $N=2^n$ bits, where the input oracle is \emph{promised} to either be one-to-one or to hide a nonzero period. One can show that the quantum query complexity for this problem is $\poly(n)$ and the randomized query complexity is $\Omega(\sqrt{2^n})$. The seminal work of Beals et al.~\cite{beals2001quantum} showed that this partial promise on the function is \emph{necessary}, i.e., for total Boolean functions, quantum and classical query complexities are polynomially related. Thus, the exponential advantage of Simon's algorithm cannot directly survive after removing the promise  in the  worst-case query model. For our second application, we consider a natural \emph{total version} of Simon's problem where there is no promise on the function and the goal is to decide if there is a period. The quantum worst case complexity of this problem is $\Omega(2^{n/2})$ and surprisingly we show that it admits a polynomial-query smoothed quantum algorithm, witnessing  an exponential separation between quantum smoothed and classical smoothed complexity.

\emph{Fourier sampling for other groups?} The proof of our quantum smoothed algorithm follows~\cite{ambainis2001average} in using  the weak Fourier sampling algorithm for the usual Simon's problem, but showing why it works in the smoothed model is non-trivial and deferred to Section~\ref{sec:total}. We find this example conceptually suggestive. The quantum part of the algorithm is not a new problem-specific primitive: it is precisely the Fourier-sampling procedure at the heart of hidden-subgroup algorithms.   Conceptually, this result closely mirrors the philosophy of classical smoothed analysis: we do not design a new quantum algorithm specifically for the perturbed problem, rather what changes is the input model. Smoothing creates enough additional structure for this standard primitive to become efficient, even though it can be exponentially inefficient on arbitrary total instances. This raises the broader possibility that other standard quantum primitives for hidden-structure problems---including Fourier-sampling techniques beyond the abelian setting---may reveal new speedups when viewed through the lens of smoothed complexity.

\paragraph{\emph{2.} Characterizing symmetric functions.}
Symmetric Boolean functions have long served as a simple but rich class of functions for understanding query complexity. In a seminal work, Beals et al.~\cite{beals2001quantum} gave tight characterizations of the quantum query complexity of \emph{every} symmetric Boolean function. More generally, symmetric functions provide a particularly clean setting in which query complexity reduces to distinguishing inputs of different Hamming weights. Our first result gives an analogous characterization in the smoothed model. We introduce two combinatorial parameters, the \emph{transition cost} and \emph{weighted influence}, which characterize, up to polylogarithmic factors, the randomized and quantum query complexities of every symmetric Boolean function. We defer their precise definitions and the proof of the characterization to Section~\ref{sec:symmetric}.

Our proof first characterizes the expected number of queries needed on inputs of each Hamming weight. Symmetrization allows us to obtain lower bounds from known results on distinguishing two Hamming weights. For the upper bounds, we estimate the weight with increasing precision, using random sampling classically and amplitude estimation quantumly, and stop once all weights consistent with the estimates give the same function value. The main technical difficulty is averaging these costs over the perturbed weight distribution: inputs near a transition, where the function changes value between adjacent weights, can be much more expensive than typical inputs. For the randomized bound, concentration of the perturbed weight allows us to control the average using the transition cost. For the quantum bound, contributions from several transitions can accumulate, giving the additional weighted influence term. To control these contributions, we partition the Hamming weights into intervals on which the function is constant and use the unimodality of the perturbed weight distribution.

\paragraph{\emph{3.} Speedups for pattern matching.} 
Our third application turns to concrete, natural problems and asks whether quantum algorithms in the smoothed model can yield polynomial or even supe-polynomial speedups. To that end, we focus on two fundamental string problems: \emph{pattern matching} and \emph{edit distance}.  Our initial motivation to look at string problems was the  average-case quantum pattern-matching algorithm by Montanaro~\cite{montanaro2017quantum}. Recall that in pattern matching, we are given a pattern and a string (via an oracle) and the goal is to decide if the pattern occurs in the string or not. He showed that when both the text and pattern are \emph{random} say of length $n, n/100$, pattern matching admits a quantum algorithm running in time
$2^{O(\sqrt{\log n})}$, yielding a superpolynomial improvement over the best classical complexity whose complexity scales $\Omega(\sqrt{n})$. In the worst case, note that the quantum-classical separation is quadratic, i.e., $O(\sqrt{n})$ quantum upper bound versus classical $\Omega(n)$ lower bound. This provides a natural example where moving away from worst-case inputs exposes a much larger quantum advantage. Motivated by this result, we ask whether the same phenomenon persists in the smoothed model, where the underlying instance remains adversarial and only a small random perturbation is applied. 

Our first observation is that, one can redo the analysis in Montanaro and show that the corresponding quantum algorithm continues to give a super-polynomial improvement even in the \emph{smoothed} setting, i.e., slightly going away from the worst case already gives a super-polynomial quantum-classical separation.\footnote{We remark that the definition of smoothed model is also well-defined for discrete problems in classical literature, which we define in the prelimnaries.} We then consider the standard \emph{exact pattern-matching} problem, in which an arbitrary pattern $P\in\Sigma^m$ is given explicitly to the algorithm (instead of being given as an oracle) and the goal is to determine whether it occurs in a text of length $n$. Here the pattern is completely unrestricted and only the text is perturbed. This is both the usual formulation of string matching and a more adversarial setting than the random-pattern model: the algorithm cannot rely on any generic structure of the pattern $P$, and furthermore remain correct on every realized perturbed text. Restricting for simplicity to the binary alphabet $\Sigma=\mathbb{F}_2$, each coordinate of the text is independently rerandomized with probability $\sigma$. We give optimal bounds for this~task.

\paragraph{\emph{4.} Speedups for Edit distance.}Our second string application concerns \emph{edit distance}, one of the canonical problems in fine-grained complexity. The standard dynamic-programming algorithm runs in $O(n^2)$ time, and Backurs and Indyk~\cite{BackursIndyk15} showed that even an $O(n^{1.99})$-time algorithm for exact edit distance would refute SETH. Quantumly, the QSETH framework of Buhrman, Patro, and Speelman~\cite{BPS21} gives an analogous $n^{1.5-o(1)}$ conditional barrier for exact edit distance. 
These barriers concern the running-time complexity of the exact worst-case problem.
This makes approximation and smoothed input models natural avenues for going beyond the fine-grained complexity of the exact worst-case problem. Indeed, smoothed edit distance was studied classically in~\cite{andoni2012smoothed,AN10UlamTester}, where substantially faster approximation algorithms were obtained. 
Here, instead, we focus on the quantum query complexity of approximating edit distance in the smoothed model and show that there exists a quantum algorithm that solves the $O(1/\sigma)$-approximate edit distance problem in the $\sigma$-smoothed model using 
$$\widetilde O\left(
        \frac{1}{\sigma}\left(\frac{n}{R}+\frac{n^{2/3}}{R^{1/3}}\right)
    \right)$$
queries, where $R$ denotes the edit distance between the two smoothed strings.

\emph{Comparison with prior work.}
We first compare our result with the classical bound for smoothed edit distance. Andoni and Krauthgamer~\cite{andoni2012smoothed} gave a reduction from smoothed edit distance to worst-case Ulam distance, namely, edit distance between two non-repetitive strings. Combining this reduction with the worst-case Ulam distance algorithm of Andoni and Nguyen~\cite{AN10UlamTester} gives
\[
    \widetilde O \!\left(
        \frac{1}{\sigma}
        \left(
            \frac{n}{R}+\sqrt n
        \right)
    \right)
\]
queries to the smoothed strings. Our quantum algorithm replaces the $\sqrt n$ term by $n^{2/3}/R^{1/3}$. Hence, when $R=\omega(\sqrt n)$, our query complexity is polynomially smaller; in particular, for $R=\Theta(n)$ it decreases from $\widetilde O(\sqrt n/\sigma)$ to $\widetilde O(n^{1/3}/\sigma)$.

On the quantum side, Le Gall and Seddighin~\cite{le2022quantum} gave a $\widetilde O(\sqrt n)$-query algorithm for approximating worst-case Ulam distance, which together with the smoothed reduction we mentioned above gives a $\widetilde O(\sqrt{n}/\sigma)$-query algorithm for smoothed edit distance. In contrast, our algorithm above gives a better algorithm whenever $R=\omega(\sqrt{n})$, in particular for $R=\Theta(n)$, it improves the $\sqrt{n}/\sigma$ complexity to $O(n^{1/3}/\sigma)$.

\emph{Improved quantum collision estimation.} We do not sketch our quantum algorithm but state a technical contribution that might be of independent interest. 
Collision estimation is a basic primitive underlying sublinear algorithms for Ulam distance. Recently, Le Gall and Ng~\cite{Gall2022approximatecollisioncounting} gave a quantum algorithm for estimating the number $m$ of collisions between two non-repetitive strings to relative error $\varepsilon$ using
$\widetilde O\!\left(\varepsilon^{-25/24}(n/\sqrt m)^{2/3}\right)$ queries. We improve the dependence on $\varepsilon$ and obtain a query complexity of
$\widetilde O\!\left(\varepsilon^{-2/3}(n/\sqrt m)^{2/3}\right)$.
Building on this improvement, we further obtain a sharper, gap-sensitive collision-counting primitive: distinguishing at most $\ell_1$ collisions from at least $\ell_2$ collisions requires only
\[
\widetilde O\!\left(
\left(
\frac{n\sqrt{\ell_2}}{\ell_2-\ell_1}
\right)^{2/3}
\right)
\]
queries.
Our algorithm combines quantum walks on the Johnson graph for approximate counting along with a sharper analysis tailored to collision thresholds, yielding improved gap, multiplicative, and additive collision estimators. In complementary steps of the Ulam distance algorithm, we need to recover the collisions themselves rather than merely estimate their number. For this we adapt the recent multiple collision algorithm of~\cite{Bonnetain2025multiplecollision} to the fixed matching induced by two non-repetitive strings. The results of \cite{Bonnetain2025multiplecollision} concern random functions but their arguments can be modified to get a (tight) quantum algorithm for enumerating all $m$ collisions between two non-repetitive strings in $\widetilde O(n^{2/3}(m+1)^{1/3})$ queries. Our argument is much simpler than the argument of \cite{Bonnetain2025multiplecollision} since the collision statistics for a matching are much easier to establish than that of a random function, and we are guaranteed to only have singletons and collision pairs (and no larger collision tuples). Appendix~\ref{app:fixed-matching-collisions} gives a self-contained proof of this result. Thus, our new collision estimators together with this enumeration algorithm provide the two complementary primitives needed to quantize the Ulam distance framework of~\cite{AN10UlamTester}. 
Along the way, we also extend the framework of variable-time quantum search~\cite{Ambainis10VariableTime} to \emph{variable-time approximate counting}, which may be of independent interest.

\subsection{Open questions}
A central goal at the moment in the field  is to discover new problems for which quantum computers offer substantial computational advantages. Smoothed complexity provides a different lens through which to search for such algorithms: rather than restricting attention to worst-case instances, one can ask whether small random perturbations expose quantum speedups that are otherwise hidden. It is natural to ask whether this perspective leads to new examples of superquadratic, or even superpolynomial, quantum speedups. We view our work as initiating the study of quantum algorithms in the smoothed complexity model, and it leaves open several basic questions for this~direction.
\begin{enumerate}
    \item \textbf{Fourier sampling} for non-abelian groups: We show that a minor modification of Fourier sampling can be useful for solving total Simon's problem exponentially faster than classical algorithms. Perhaps can we use weak/strong Fourier sampling and show it solves Graph isomorphism on smoothed instances? Although this result is known classically, seeing if Fourier sampling solves it is an interesting question. 
    \item \textbf{More practical problems} wherein smoothed complexity can help: In this work, we restricted our attention only to string problems, but it is natural to consider other problems and see when algorithms work well on realistic instances, rather than worst case instances.
    \item \textbf{State learning} transitioning between Frobenius and operator norm. There have been many works in learning quantum objects in Frobenius norm or operator norm, which can be viewed as the average case and worst case metric respectively. Is there a smoothed notion that interpolates these norms for which one can obtain learning algorithms.  There have been a couple of works~\cite{baumer2025approximate,kahanamoku2025log,zhao2021smooth,yi2026certifying,johnston2023computing} that have considered some version of average case for their tasks, motivating the question of what happens for these tasks in the smoothed model.

    \item \textbf{Quantum heuristics.}  Some prominent quantum algorithms such as QAOA and quantum annealing are heuristic with limited provable guarantees. It would be interesting to analyze such algorithms in the smoothed model and ask whether small random perturbations of the input can lead to rigorous performance guarantees, providing a theoretical explanation for when and why these heuristics perform well in practice?
\end{enumerate}

\paragraph{AI disclosure.} 
All ideas in this paper were generated by the authors.
We relied on GPT 5.5's help in proving Lemma~\ref{lem:localization} (a  technical lemma for obtaining tight symmetric function characterization) and 
during the preparation of this manuscript, the authors used  GPT-5.5 and GPT-5.6 Sol for literature searches, mathematical discussions revolving~\cite{andoni2012smoothed,andoni2013shift}, and editorial revisions.  The authors independently verified all mathematical arguments and references and take full responsibility for the content of the paper.

\section{Preliminaries}

\subsection{Useful lemmas}
\begin{lemma}[Chernoff bound]\label{lem:chernoff2}
Let $X_1,\ldots,X_k$ be independent Bernoulli random variables,
and put $X=\sum_{i=1}^k X_i$ and $\mu=\E[X]$.
Then, for every $0<\delta<1$, we have
$\Pr[X\le(1-\delta)\mu]\le\exp(-\delta^2\mu/2)$.
\end{lemma}

\begin{lemma}[Bernstein inequality]\label{thm: Bernstein}
    Let $X_1, \dots, X_N$ be independent, zero-mean random variables satisfying $|X_i| \leq K$ all $i$. Then, for every $t \geq 0$, we have
    \[
        \mathbb{P}\left\{ \left| \sum_{i=1}^N X_i \right| \geq t \right\} \leq 2 \exp \left( - \frac{t^2/2}{\sigma^2 + Kt/3} \right)
    \]
    where $\sigma^2 = \sum_{i=1}^N \mathbb{E} X_i^2$ is the variance of the sum.
\end{lemma}

\begin{lemma}[McDiarmid's inequality]\label{lem:mcdiarmid}
Let $X_1,\ldots,X_N$ be independent random variables, and let
$Z=h(X_1,\ldots,X_N)$ be real-valued. Suppose that changing only
the $i$th argument of $h$ changes its value by at most $c_i$,
where $c_i>0$. Then, for every $t\ge0$, we have
$\Pr[Z-\E[Z]\le-t]\le
\exp(-2t^2/\sum_{i=1}^N c_i^2)$.
\end{lemma}

\begin{lemma}[\cite{AuldNeammanee24}]
\label{lem:local-limit}
Let $X=\sum_{i=1}^N X_i$, where the $X_i$ are independent Bernoulli random variables, and let
\[
    \mu:=\E[X],
    \qquad
    v:=\operatorname{Var}(X).
\]
For every constant $C>0$, if $v$ is sufficiently large and $z\in\mathbb Z$ satisfies
\[
    |z-\mu|\leq C\sqrt v,
\]
then 
$    \Pr[X=z]
    =
    \Theta_C\left({1}/{\sqrt v}\right).$
\end{lemma}

\begin{lemma}\label{lem:poisson-binomial-facts}
Let $Z$ be a sum of independent Bernoulli random variables, with mean $\mu$,
variance $v>0$, and probability mass function $p$. Then $p$ is log-concave
and has a mode $m$ satisfying $|m-\mu|<1$. Moreover, for absolute constants
$c,C>0$ and every $u\ge0$,
\[
\Prb[|Z-\mu|\ge u]
\le
2\exp\left(-c\min\left\{\frac{u^2}{v},u\right\}\right),
\qquad
\max_kp(k)\le\frac{C}{1+\sqrt v}.
\]
\end{lemma}

\begin{proof}
The probability generating polynomial of $Z$ is a product of linear
polynomials with nonpositive real roots. Newton's inequalities therefore show
that its coefficient sequence is log-concave. Darroch's theorem
\cite{darroch1964distribution} gives a mode $m$ with $|m-\mu|<1$. Apply Bernstein's inequality to the centered Bernoulli variables to obtain the
tail bound. If $v$ is sufficiently large, the bound on the largest atom follows
from Lemma~\ref{lem:local-limit}, applied at $m$. If $v$ is bounded, it follows
from the trivial bound $p(k)\le1$.
\end{proof}

\begin{lemma}\label{lem:monotone-sequence}
Let $a_0,\ldots,a_T\geq 0$ be nondecreasing, and let
$A=\sum_{t=0}^T a_t$. Then
$\sum_{t=0}^T a_t/(t+1)\le CA\log T/(T+1)$
for an absolute constant $C>0$.
\end{lemma}

\begin{proof}
Since the sequence is nondecreasing, $a_t\le A/(T-t+1)$.
The claim follows by summing $1/((t+1)(T-t+1))$.
\end{proof}

\begin{lemma}\label{lem:expected-query-truncation}
Let $X_0$ and $X_1$ be the two parts of a promise problem, and let $A$ be a
bounded-error randomized or quantum query algorithm for this problem. Suppose
that the expected number of queries made by $A$ is at most $t$ on every input
in $X_0$. Then there is a bounded-error algorithm for the same promise problem
that makes $O(t)$ queries in the worst case.
\end{lemma}

\begin{proof}
The probability that $A$ makes at least one query is independent of the input
until the first query is made. The output distributions on an input in $X_0$
and an input in $X_1$ differ in total variation distance by at least $1/3$,
whereas their no-query branches are identical. Thus the probability of making
a query is at least $1/3$, and hence $t=\Omega(1)$.

Truncate $A$ after $\lceil12t\rceil$ queries and, upon timeout, output the value
corresponding to $X_1$. On $X_1$ this does not increase the error. On $X_0$,
Markov's inequality shows that the timeout probability is at most $1/12$, so
the error is at most $1/3+1/12<1/2$. A constant number of repetitions gives
error at most $1/3$ and uses $O(t)$ queries.
\end{proof}

\subsection{Quantum subroutines}\label{sec:q_subroutines}
\begin{theorem}[Amplitude Estimation~\cite{brassard2002AmpAndEst}]\label{thm:amplitude_estimation}
Let $\delta \in(0,1)$. Given a natural number $M$ and access to an $(n + 1)$-qubit unitary $U$ satisfying
\[
U\ket{0^n}\ket{0}= \sqrt{a}\ket{\phi_0}\ket{0} +\sqrt{1-a}\ket{\phi_1}\ket{1},
\]
where $\ket{\phi_0}$ and $\ket{\phi_1}$ are arbitrary $n$-qubit states and $ a \in [0,1]$,
there exists a quantum algorithm that uses ${O}(M\log(1/\delta))$ applications of $U$ and $U^\dagger$ and $\tilde{{O}}(M\log(1/\delta))$ elementary gates, and outputs an estimator $\lambda$ such that, with probability $\geq 1-\delta$,
\[
|{a}-\lambda| \leq \frac{\sqrt{a(1-a)}}{M}+\frac{1}{M^2}.
\]
\end{theorem}

\begin{theorem}[Fixed-point amplitude amplification~\cite{gilyen2018QSingValTransf,yoder2014FixedPointSearch}]\label{thm:Fixed_AA}
Let $a,\delta>0$, let $U$ be a unitary that maps $\ket{0}\rightarrow\ket{\psi}$, and let $R_{\mathcal{A}},R_{\ket{0}}$ be quantum circuits that reflect through $\mathcal{A}$ and (the span of) $\ket{0}$, respectively. Suppose $\|\Pi_{\mathcal{A}}\ket{\psi}\|\geq a$. There is a quantum algorithm that prepares $\ket{\psi'}$ satisfying $\|\ket{\psi'}-\frac{\Pi_{\mathcal{A}}\ket{\psi}}{\|\Pi_{\mathcal{A}}\ket{\psi}\|}\|\leq\delta$ using $\mathcal{O}(\log(1/\delta)/a)$ applications of $U,U^{-1}$, controlled $R_{\mathcal{A}},R_{\ket{0}}$, and additional single-qubit gates. Moreover, if $q:=\|\Pi_{\mathcal A}\ket{\psi'}\|^2$ denotes the acceptance probability of the algorithm, then for constant $\delta$ and arbitrary $\ket{\psi}$, $q=O\!\left(\min\{1,\|\Pi_{\mathcal A}\ket{\psi}\|^2/a^2\}\right)$.
\end{theorem}

We will also apply the quantum-walk framework to Johnson graphs.
Following~\cite{MNRS11,BCSS23}, let $J([N],r)$ denote the graph whose vertices are the $r$-subsets of $[N]$, with $S,S'$ adjacent whenever $|S\cap S'|=r-1$. The corresponding random walk is reversible with respect to the uniform distribution and, for $r\leq N/2$, has spectral gap $\Theta(1/r)$. In its quantum implementation, we maintain a database of the queried input values indexed by the current subset; moving to a neighbor therefore changes only one stored index.

\begin{theorem}[Quantum walk cost~\cite{Gall2022approximatecollisioncounting}] \label{thm:approx_marked_fraction}
    Let $P$ be a reversible ergodic Markov chain, $\tau,\delta > 0$ be the spectral gap of $P$ and $\lambda > 0$ be a known lower bound on the fraction $p_M$ of marked states of $P$ with respect to its stationary distribution. Let $S$, $U$ and $C$ be the number of queries needed to implement, respectively, the setup operation, the update operation and the checking operation. For any $\eta \in (0,1)$, there exists a quantum algorithm which outputs with success probability $\geq 1-\tau$, an estimate $\widetilde{p}_M$ such that $|p_M - \widetilde{p}_M| \leq \eta p_M$ with query complexity
    \begin{align*}
        \widetilde{O}\left(S + \frac{1}{\eta} \frac{1}{\sqrt{\lambda}}\left(\frac{1}{\sqrt{\delta}} U + C\right) \right).
    \end{align*}
\end{theorem}

\subsection{Smoothed model and query complexity}

In this paper we will be concerned mostly with the query model of computation (leaving time model of computation for future work). 
In particular, consider the following query model: for every $x\in \01^n$, let $N_\sigma(x)$ be the distribution of all bit strings $y\in \01^n$ that can be obtained by letting\footnote{Although many works in smoothed function literature deal with real-valued problems, this smoothed definition for discrete objects has been considered in the well-known works  in~\cite{chen2020polynomial,spielman2006smoothed,andoni2012smoothed}.}  
\[ y_i=\begin{cases} 
      x_i & \text{with probability }1-\sigma \\
      1 & \text{with probability }\sigma/2 \\
      0 & \text{with probability }\sigma/2.
   \end{cases}
\]
We say an algorithm $\A$ computes $f$, if on every input $x$, the output of $\A$ on input $x$, equals $f(x)$ with probability $\geq 2/3$. This algorithm could be quantum or classical. 
For an algorithm $\A$, define $T(\A,y)$ as the expected query complexity of $\A$ on input $y$. In this work, we define the smoothed query complexities as
$$
Q_{\mathrm{sm},\sigma}(f)=\min_\A \max_x \mathop{\mathbb{E}}_{\textbf{y}\sim N_\sigma(x)} [T(\A,\textbf{y})],
$$
where the minimization is over all quantum algorithms $\A$ that compute $f$, and the expectation is over $y\sim N_\sigma(x)$. Similarly define $R_{\mathrm{sm},\sigma}(f)$, where the minimization is over all \emph{randomized} algorithms computing $f$.
We also write $Q_\sigma(f)$ and $R_\sigma(f)$ for short when it is clear from the context that we refer to smoothed query complexity. Naturally when discussing the complexity of the smoothed model, the complexity of the algorithm would depend on $1/\sigma$ and it is desirable to have query and gate complexity that scales polynomially with $1/\sigma$.

When $\sigma=0$, then these definitions give the usual worst-case expected query complexities. When $\sigma=1$, the perturbed input is uniform, so the cost is the uniform average expected cost, while correctness is still required on every input. A seminal result of Beals et al.~\cite{beals2001quantum} showed that for the worst case complexity $Q(f)$ and $R(f)$ are polynomially related to one another, ruling out exponential separations between quantum and classical query complexity. Subsequently, Ambainis and de Wolf~\cite{ambainis2001average} considered the average case query complexity showed that $Q(f)$ and $R(f)$ can be exponentially separated. Like we mentioned in the motivation for the smoothed model, the worst-case model seems to capture instances of function computations that could be hard and also rare and the average-case model captures the complexity of the problem that doesn't mimic ``real-life instances", motivating what is the relation between quantum and classical smoothed query complexities.  

\section{Total function separation}
\label{sec:total}

In this section, we will consider the smoothed quantum query complexity of the total version of the Simon's problem. In particular, we will consider the problem defined by Ambainis and de Wolf \cite{ambainis2001average}, who showed that in the average case, this problem has a quantum algorithm with query complexity $O(n)$ but any classical algorithm must make $\Omega(2^{n/2}$ queries.  We consider their same algorithm (which is essentially Fourier sampling approach used for the hidden subgroup problem) and show that its smoothed quantum query complexity is $\poly(n)$ and by smoothed here, we mean smoothing the bits of the truth table as we defined in the previous section.

\begin{definition}[Total Simon's problem ($\Simon$)]
   Let $f:\{0,1\}^n \to \{0,1\}^n$. Given oracle access to $f$, decide if there is an $s \in \{0,1\}^n\backslash \{0^n\}$ such that $f(x) = f(x \oplus s)$ for all $x \in \{0,1\}^n$.
\end{definition}
Note that this is a total function since, we are not promised that $f$ satisfy any promise. 
It is not hard to see that $Q(\Simon)=\Omega(2^{n/2})$ and $R(\Simon)=\Omega(2^n)$. To see this, one can consider the problem of distinguishing between \begin{enumerate}
    \item there is a unique $s \in \{0,1\}^n$ such that $f(x) = f(x \oplus s)$ for all $x \in \{0,1\}^n$
    \item there is a unique $s \in \{0,1\}^n$ and indices $y,z$ such that $y \oplus z = s$, $f(y) \neq f(z)$ and $f(x) = f(x \oplus s)$ for all $x \in \{0,1\}^n \setminus \{y,z\}$.
\end{enumerate} Clearly distinguishing the two cases above is as hard as the search problem on $2^n$ bits, so the quantum query complexity is $\Omega(2^{n/2})$ and the classical query complexity will be $\Omega(2^n)$. 
The main result we show here is the following. 
\begin{theorem}
\label{thm:smoothed-simon}
For every $\sigma\in(0,1]$ satisfying $\sigma\geq \Omega(n/N)$, the smoothed quantum query complexity of
Total Simon satisfies
   $Q_{\mathrm{sm},\sigma}(\mathrm{Simon})
    = O(n/\sigma).$
\end{theorem}
We make an important remark about the theorem above. The algorithm that we use to prove the theorem above is well-known standard method/Fourier-sampling used in solving the hidden subgroup problem. To this end, our Simon result gives another interpretation of smoothed quantum complexity: the standard Fourier-sampling primitive underlying the hidden subgroup algorithm, which has exponential worst-case complexity for the \emph{total} problem, becomes efficient in smoothed expected complexity. Random perturbations create further structure, and these generate uniform Fourier samples that rapidly certify the absence of a hidden period. Potentially this idea could also be further explored to understand if the strong Fourier sampling approach for non-Abelian Hidden subgroup problem could yield faster algorithms for non-Abelian groups (say dihedral group or symmetric group). Finally our algorithm is directly inspired by the one in~\cite{ambainis2001average}, we just adapt their analysis to go from average-case analysis to smoothed-case analysis.

\begin{proof}
  We use the following algorithm on an arbitrary oracle
$g:\F_2^n\to \F_2^n$.

\begin{tcolorbox}
\begin{enumerate}

    \item Prepare the uniform superposition
over $\F_2^n$. Query the oracle to get  
$        \frac{1}{\sqrt N}\sum_{x\in \F_2^n}\ket{x,g(x)}.$

    \item Measure the second register, suppose that the outcome is
    $y\in \F_2^n$.  Let $F_y:=g^{-1}(y)=\{x\in \F_2^n:g(x)=y\}$.  The first register collapses to
    \[
        \frac{1}{\sqrt{|F_y|}}
        \sum_{x\in F_y}\ket{x}.
    \]

    \item Apply $H^{\otimes n}$ to the first register, obtaining
    \[
        \frac{1}{\sqrt{N|F_y|}}
        \sum_{z\in \F_2^n}
        \left(\sum_{x\in F_y}(-1)^{z\cdot x}\right)\ket{z},
    \]
    and measure to obtain $z\in \F_2^n$.

    \item Repeat Steps 1--3, $t=512 n/\sigma$  times, obtaining
    $z_1,\ldots,z_t$. Let
    $W:=\operatorname{span}_{\mathbb F_2}\{z_1,\ldots,z_t\}$.

    \item If $W=\F_2^n$, output \textsc{No}. Else, query the entire
    truth table of $g$ to determine the output.
\end{enumerate}
\end{tcolorbox}

We first verify that this algorithm is correct on every input and then
bound its expected number of queries under smoothing.

We first observe that {YES instances are never rejected.}
Suppose that $g:\F_2^n\to \F_2^n$ has a nonzero period $s\in \F_2^n$, so that
$g(x)=g(x+s)$ for every $x\in \F_2^n$. Fix any measurement outcome $y$.
The fiber $F_y=g^{-1}(y)$ is invariant under translation by $s$, and,
since $s\neq0$, its elements can be partitioned into disjoint pairs
$\{x,x+s\}$. After applying the Hadamard transform, the amplitude of $\ket z$ is
proportional to $\sum_{x\in F_y}(-1)^{z\cdot x}$. If $z\cdot s=1$,
the contribution of each pair $\{x,x+s\}$ is
\[
    (-1)^{z\cdot x}+(-1)^{z\cdot(x+s)}
    =
    (-1)^{z\cdot x}\bigl(1+(-1)^{z\cdot s}\bigr)
    =
    0.
\]
Consequently every Fourier sample satisfies $z\cdot s=0$. Thus
$z_1,\ldots,z_t\in s^\perp$, where
$s^\perp:=\{z\in \F_2^n:z\cdot s=0\}$. Since $s\neq0$, the space
$s^\perp$ is a proper subspace of $\F_2^n$, and hence $W\neq \F_2^n$. Therefore the algorithm never outputs \textsc{No} in Step~$(5)$ on a
\textsc{Yes} instance. It instead reaches step $(5)$,
which returns the correct answer. On a \textsc{No} instance, either
$W=\F_2^n$, in which case the algorithm correctly outputs \textsc{No}, or
step $(5)$ returns the correct answer. Hence the
algorithm is correct on every $g$.

We now analyze the expected cost of this fixed algorithm when its input
is a $\sigma$-perturbation of an arbitrary function $f:\F_2^n\to \F_2^n$.
To this end, let $f^\sigma$ be a $\sigma$-perturbation of $f$, obtained by
independently resampling each truth-table entry with probability $\sigma$;
a resampled value is chosen uniformly from $\mathbb F_2^n$. In particular, for $y\neq f(x)$ we have
$\Pr[f^\sigma(x)=y]=\sigma/N$, whereas
$\Pr[f^\sigma(x)=f(x)]=1-\sigma+\sigma/N$. 
For a function $g:\F_2^n\to \F_2^n$, let
$m(g):=|\{y\in \F_2^n:|g^{-1}(y)|=1\}|$ denote the number of singleton
fibers of $g$. The following lemma is the key consequence of~smoothing.

\begin{lemma}
\label{lem:many-singletons}
Let $f:\F_2^n\to \F_2^n$ be arbitrary. Then
\[
    \Pr\left[
        m(f^\sigma)\ge \frac{N\sigma}{64}
    \right]
    \ge
    1-e^{-N\sigma/8}-e^{-N\sigma/4096}.
\]
In particular, 
$m(f^\sigma)\ge N\sigma/64$ with probability
$1-2^{-\Omega(n)}$.
\end{lemma}

\begin{proof}
Let $\mathcal P\subseteq \F_2^n$ be the set of points at which $f$ is
resampled, and let $\eta:=|\mathcal P|$. Since
$\eta\sim\Bin(N,\sigma)$, we have that $\E[\eta]=N\sigma$ and, by Chernoff bound (Lemma~\ref{lem:chernoff2}), we get
\[
    \Pr\!\left[\eta<{N\sigma}/2\right]
    \le e^{-N\sigma/8}.
\]
Condition now on a fixed resampled set $\mathcal P$ of size $\eta$.
On $\mathcal P$, the values of $f^\sigma$ are independent and uniform
in $\F_2^n$; outside $\mathcal P$, they are unchanged. 
Let $\mathcal N:=\F_2^n\setminus\mathcal P$, so
$|\mathcal N|=M:=N-\eta$, and let $R:=f(\mathcal N)$ with
$r:=|R|$. Let $q$ be the number of values in $R$ that occur exactly
once among the points of $\mathcal N$. Since the remaining $r-q$
values occur at least twice,
$    M\ge q+2(r-q)=2r-q,$
and therefore $q\ge 2r-M$. We now consider two cases.
\paragraph{Case 1.}
Suppose $r\ge M/2+\eta/8$. Then $q\ge\eta/4$. Let $X$ be the number
of these $q$ singleton values that are not hit by any of the $\eta$
fresh uniform outputs. Every value counted by $X$ remains a singleton
fiber of $f^\sigma$, and
\[
    \mathbb E[X]
    =
    q\left(1-\frac1N\right)^\eta
    \ge
    \frac q4
    \ge
    \frac{\eta}{16},
\]
where we used $\eta\le N$ and
$(1-1/N)^N\ge1/4$ for $N\ge2$. Changing one of the $\eta$ fresh outputs can change $X$ by at most
$2$. Therefore McDiarmid's inequality
(Lemma~\ref{lem:mcdiarmid}) gives us
\[
    \Pr\left[X<\frac{\eta}{32}\right]
    \le
    \exp\left(
        -\frac{2(\eta/32)^2}{4\eta}
    \right)
    =
    e^{-\eta/2048}.
\]
Thus, except with probability $e^{-\eta/2048}$,
$m(f^\sigma)\ge X\ge\eta/32$.

\paragraph{Case 2.}
Suppose instead that $r<M/2+\eta/8$. Since $M=N-\eta$, this implies
$r<N/2-3\eta/8\le N/2$, and hence $|\F_2^n\setminus R|\ge N/2$. Let $Y$ be the number of values in $\F_2^n\setminus R$ that are hit
exactly once by the $\eta$ fresh uniform outputs. Every value counted
by $Y$ is a singleton fiber of $f^\sigma$. For each
$y\in \F_2^n\setminus R$, the probability that exactly one fresh output
equals $y$ is
$\frac{\eta}{N}(1-\frac1N)^{\eta-1}$. Consequently,
\[
    \mathbb E[Y]
    =
    |\F_2^n\setminus R|\cdot 
    \frac{\eta}{N}
    \left(1-\frac1N\right)^{\eta-1}
    \ge
    \frac{\eta}{2}
    \left(1-\frac1N\right)^{N-1}
    \ge
    \frac{\eta}{2e}
    >
    \frac{\eta}{6}.
\]
Again, changing one fresh output changes $Y$ by at most $2$.
McDiarmid's inequality
(Lemma~\ref{lem:mcdiarmid}) therefore gives
\[
    \Pr\left[Y<\frac{\eta}{12}\right]
    \le
    \exp\left(
        -\frac{2(\eta/12)^2}{4\eta}
    \right)
    =
    e^{-\eta/288}.
\]
Thus, except with probability $e^{-\eta/288}$,
$m(f^\sigma)\ge Y\ge\eta/12\ge\eta/32$.

Combining the two cases, for every \emph{fixed perturbation} set
$\mathcal P$ of size $\eta$,
\[
    \Pr\left[
        m(f^\sigma)<\frac{\eta}{32}
        \,\middle|\,
        \mathcal P
    \right]
    \le
    e^{-\eta/2048}.
\]
On the event $\eta\ge\mu/2$, this implies
$m(f^\sigma)\ge\eta/32\ge\mu/64=N\sigma/64$ except with conditional
probability at most $e^{-\mu/4096}$. Combining this with
$\Pr[\eta<\mu/2]\le e^{-\mu/8}$ proves
\[
    \Pr\left[
        m(f^\sigma)<\frac{N\sigma}{64}
    \right]
    \le
    e^{-N\sigma/8}+e^{-N\sigma/4096}.
\]
For a sufficiently large enough constant in $\sigma\geq \Omega(n/N)$, we have
$e^{-N\sigma/4096}\le e^{-2n}$, so the failure probability is
$2^{-\Omega(n)}$.
\end{proof}

We have now so far argued about the correctness of the algorithm and also the complexity of it being small, except when the algorithm reaches step $5$ in which case it makes the exponentially many $N=2^n$ many queries to the truth table. To prove our main theorem statement, we  now bound the probability that  Step~5 is invoked. Conditioned on the event
$m(f^\sigma)\ge N\sigma/64$ from
Lemma~\ref{lem:many-singletons}, consider one execution of the Fourier-sampling procedure. Measuring
the second register yields a value $y$ with probability
$|(f^\sigma)^{-1}(y)|/N$. Consequently, the probability of observing
a singleton fiber is exactly $m(f^\sigma)/N\ge\sigma/64$.
 Conditioned on observing a singleton fiber $\{x\}$, the first register
is exactly $\ket x$, and applying the Hadamard transform gives
\[
    H^{\otimes n}\ket x
    =
    \frac{1}{\sqrt N}
    \sum_{z\in \F_2^n}(-1)^{x\cdot z}\ket z.
\]
Thus the resulting Fourier sample $z$ is uniformly distributed over
$\F_2^n$. Fix any nonzero $a\in \F_2^n$. A uniform $z\in \F_2^n$ satisfies
$a\cdot z=1$ with probability $1/2$. Hence in each repetition,
$ \Pr[a\cdot z_i=1]
    \ge
    {\sigma}/{64}\cdot 1/2
    =
    {\sigma}/{128}.$
Conditioned on the fixed oracle $f^\sigma$, the $t$ repetitions are
independent, so
\[
    \Pr[
        a\cdot z_i=0
        \text{ for every }i\in[t]
    ]
    \le
    \left(1-\frac{\sigma}{128}\right)^t
    \le
    e^{-\sigma t/128}.
\]
If $W=\operatorname{span}\{z_1,\ldots,z_t\}$ is a proper subspace of
$\F_2^n$, then $W^\perp$ contains some nonzero vector $a$. A union bound
over the $2^n-1$ possible nonzero $a$ therefore gives
\[
    \Pr[W\neq \F_2^n]
    \le
    (2^n-1)e^{-\sigma t/128}
    <
    2^n e^{-\sigma t/128}\leq  2^ne^{-4n}
    =
    e^{-(4-\ln2)n}
    =
    2^{-\Omega(n)},
\]
where we used  $t=\Theta(n/\sigma)$.  
Combining this with Lemma~\ref{lem:many-singletons}, the unconditional
probability, over both the smoothing and the measurements, that the
algorithm reaches step $(5)$ satisfies
\[
    \Pr[\text{step $(5)$}]
    \le
    e^{-N\sigma/8}
    +
    e^{-N\sigma/4096}
    +
    e^{-(4-\ln2)n}\leq 
        e^{-1024n}
        +
        e^{-2n}
        +
        e^{-(4-\ln2)n}=2^{-Cn},
\]
for some $C>1$.
The initial Fourier-sampling stage uses
$t=O(n/\sigma)$ queries, while the exhaustive fallback costs at most
$N$ additional queries. Hence for every center $f:\F_2^n\to \F_2^n$,
\[
    \mathbb E_{f^\sigma}
    [T_A(f^\sigma)]
    \le
    O(n/\sigma)
    +
    N\Pr[\text{fallback}]
    =
    O(n/\sigma).
\]
Taking the maximum over all centers $f$ gives
$Q_{\mathrm{sm},\sigma}(\mathrm{Simon})=O(n/\sigma)$, as claimed.
\end{proof}

\begin{lemma}[Classical smoothed lower bound]
\label{lem:smoothed-simon-classical-lower-bound}
For every $\sigma\in[0,1]$,
\[
R_{\mathrm{sm},\sigma}(\mathrm{Simon})
\geq R_{\mathrm{unif}}(\mathrm{Simon})
=\Omega(2^{n/2}).
\]
\end{lemma}
\begin{proof}
Let $\mathcal U$ denote the uniform distribution over functions from $\F_2^n$ to $\F_2^n$. For every randomized algorithm $A$ that is correct on every oracle,
\[
\max_f \mathbb E_{f^\sigma}[T_A(f^\sigma)]
\geq
\mathbb E_{f\sim\mathcal U}\mathbb E_{f^\sigma}[T_A(f^\sigma)]
=
\mathbb E_{g\sim\mathcal U}[T_A(g)].
\]
The equality holds because, when $f\sim\mathcal U$, independently resampling any truth table entry with a uniform value leaves the resulting function $f^\sigma$ uniformly distributed. Taking the minimum over all such algorithms $A$ and applying the average-case lower bound of~\cite{ambainis2001average} proves the claim.
\end{proof}

Together with Theorem~\ref{thm:smoothed-simon}, Lemma~\ref{lem:smoothed-simon-classical-lower-bound} gives an exponential quantum-classical separation for every $\sigma\geq 1/\operatorname{poly}(n)$.

\section{Characterizing query complexity of symmetric functions}
\label{sec:symmetric}
\newcommand{\wInf}{\operatorname{wInf}}
\newcommand{\transcost}{\operatorname{transcost}}
In this section, we characterize the smoothed query complexity of symmetric Boolean functions. Beals et al.~\cite{beals2001quantum} gave a tight characterization of their worst-case quantum query complexity, while Ambainis and de Wolf~\cite{ambainis2001average} studied query complexity averaged over an input distribution, with correctness on every input. We ask how complexity changes when the input is obtained by perturbing an arbitrary instance, interpolating between worst-case and average-case complexity.

For a symmetric function, the difficulty of evaluating an input depends on its Hamming weight and the nearby weights on which the function takes a different value. The \emph{Hamming layer of weight $k$} is the set $\{x\in\{0,1\}^n:|x|=k\}$ of inputs with exactly $k$ ones. In the worst case, the hardest pair of adjacent weights determines the quantum query complexity. Under smoothing, however, the weight is random, so we must also account for how often the perturbed input reaches the difficult layers. A hard layer may have little probability, while contributions from several different layers may add up.

Our characterization uses two quantities, defined below in terms of the transitions of $f$, namely the adjacent Hamming weights on which its value changes. We write $\transcost_\sigma(f)$ and $\wInf_\sigma(f)$, abbreviating \emph{transition cost} and \emph{weighted influence}, respectively. The weighted influence $\wInf_\sigma(f)$ in~\eqref{eq:symmetric-binf} sums the quantum query bounds at these transitions, weighted by the probabilities of the corresponding layers. The transition cost $\transcost_\sigma(f)$ in~\eqref{eq:symmetric-g} measures the largest contribution of a single transition after accounting for the mean and spread of the perturbed weight. We show the following.

\begin{theorem}
\label{thm:symmetric}
For every total symmetric Boolean function $f$ and every $\sigma\in[0,1]$,
the randomized complexity is
$R_{\sm,\sigma}(f)=\widetilde\Theta(\min\{n,\transcost_\sigma(f)^2\})$,
and the quantum complexity is
$Q_{\sm,\sigma}(f)=\widetilde\Theta(\wInf_\sigma(f)+\transcost_\sigma(f))$.
\end{theorem}

To derive the above bounds, we combine two steps. We first bound the expected number of queries needed on inputs of each Hamming weight, using the complexity of distinguishing inputs of two different Hamming weights and algorithms that estimate the weight only until the value of $f$ is determined. We then average these bounds over the perturbed Hamming-weight distribution. In Section~\ref{sec:symmetric-layers}, we give the symmetrization argument and prove matching lower and upper bounds for each Hamming weight in both query models. In Section~\ref{sec:symmetric-averaging}, we prove Lemma~\ref{lem:localization}, which expresses the resulting averages in terms of $\wInf_\sigma(f)$ and $\transcost_\sigma(f)$ and gives the claimed smoothed query complexities. Finally, Section~\ref{sec:symmetric-examples} illustrates the characterization at $\sigma=0$ and $\sigma=1$ and compares the two query models for \textsc{Or}, thresholds at $\sqrt n$ and $n/4$, \textsc{Majority}, and \textsc{Parity}.

Throughout this section, let $f:\{0,1\}^n\to\{0,1\}$ be a total
symmetric Boolean function, and write $f(x)=g(|x|)$ for some
$g:\{0,\ldots,n\}\to\{0,1\}$. We assume $n\ge2$.\footnote{The case
$n=1$ is immediate.} For an original input $x$ of Hamming weight $r$,
let $Y\sim N_\sigma(x)$, write $W_{r,\sigma}=|Y|$, and let
$p_{r,\sigma}(k)=\Pr[W_{r,\sigma}=k]$. The distribution of the perturbed
weight depends on $x$ only through $r$. Each of its $r$ coordinates equal to one
remains one with probability $1-\sigma/2$, while each coordinate equal to zero
becomes one with probability $\sigma/2$. Thus $W_{r,\sigma}$ is the sum
of independent variables with distributions $\Bin(r,1-\sigma/2)$ and
$\Bin(n-r,\sigma/2)$. Its mean is
$\mu_{r,\sigma}=\sigma n/2+(1-\sigma)r$, and we write
$s_\sigma=1+\sqrt{n\sigma(2-\sigma)}/2$ for one plus its standard deviation.
The transition set is $\mathcal T(f)=\{j\in\{0,\ldots,n-1\}:g(j)\ne g(j+1)\}$.
For a transition $j$ and $z\in[0,n]$, define
\begin{equation}
t_j=\sqrt{(n-j)(j+1)},
\qquad
q_{j,\sigma}(z)=
\frac{\sqrt{(n-\min\{j,z\})\max\{j+1,z\}}}
{s_\sigma+\dist(z,[j,j+1])}.
\label{eq:symmetric-transition-cost}
\end{equation}
By Lemma~\ref{lem:two-layers} below, $t_j$ gives, up to constant
factors, the quantum query complexity of distinguishing inputs of
Hamming weights $j$ and $j+1$. More generally, for $k\le j$ the
quantity $q_{j,0}(k)$ is the cost of distinguishing weights $k$ and
$j+1$, and for $k\ge j+1$ it is the cost of distinguishing weights
$j$ and $k$. It therefore measures the cost of distinguishing an
input from the nearest weight across this transition. The denominator
increases with the distance to the transition, so this cost is largest
at the adjacent weights. In particular,
$q_{j,\sigma}(z)\le q_{j,0}(z)\le t_j$.

Under smoothing, the weight is typically within $O(s_\sigma)$ of its
mean. We account for this spread by replacing the denominator
$1+\dist(z,[j,j+1])$ with $s_\sigma+\dist(z,[j,j+1])$. The resulting
quantity $q_{j,\sigma}(\mu_{r,\sigma})$ measures the cost of this
transition after accounting for the spread of the perturbed weight. We now define
the two parameters appearing in Theorem~\ref{thm:symmetric}.

The \emph{weighted influence} sums the costs of transitions weighted by the
probability of obtaining either adjacent Hamming weight:
\begin{equation}
\wInf_\sigma(f)
=\max_{0\le r\le n}\sum_{j\in\mathcal T(f)}
\bigl(p_{r,\sigma}(j)+p_{r,\sigma}(j+1)\bigr)\sqrt{(n-j)(j+1)}.
\label{eq:symmetric-binf}
\end{equation}
Here the factor $p_{r,\sigma}(j)+p_{r,\sigma}(j+1)$ is the probability
of obtaining an input at one of the two weights adjacent to transition
$j$, and $t_j$ is the corresponding query cost. Summing over transitions
allows contributions from several different parts of the distribution
to accumulate. The maximum over $r$ selects the original input weight
for which their total is largest. The term weighted influence refers
to this particular weighted sum.

The \emph{transition cost} is
\begin{equation}
\transcost_\sigma(f)
=\max_{\substack{0\le r\le n\\j\in\mathcal T(f)}}
\frac{\sqrt{(n-\min\{j,\mu_{r,\sigma}\})\max\{j+1,\mu_{r,\sigma}\}}}
{s_\sigma+\dist(\mu_{r,\sigma},[j,j+1])}.
\label{eq:symmetric-g}
\end{equation}
Equivalently, it is the maximum of $q_{j,\sigma}(\mu_{r,\sigma})$ over
the original input weight $r$ and the transition $j$. This parameter
can remain significant even when the probability of landing exactly
beside a transition is small: distinguishing a typical perturbed weight
from a weight across that transition may still require many queries.
Thus $\transcost_\sigma(f)$ measures the largest contribution of one
transition, whereas $\wInf_\sigma(f)$ sums the contributions at all
transitions. We set both parameters to zero when $f$ is constant.

\subsection{Query complexity at a fixed Hamming weight}
\label{sec:symmetric-layers}
We first characterize the expected query cost on inputs of each Hamming
weight. Symmetry lets us assume this cost is the same for all inputs
with a given weight. We then prove lower bounds by distinguishing two
weights on which $f$ differs, and matching upper bounds by estimating
the weight until the value of $f$ is determined.

\begin{lemma}[Symmetrization]
\label{lem:symmetrization}
Let $\A$ be a randomized or quantum query algorithm computing $f$.
There is an algorithm computing $f$ with the same error bound, no larger
smoothed query cost, and expected query cost $T_{|y|}$ on every input
$y$, for some $T_0,\ldots,T_n$. Its smoothed query cost is
$\max_{0\le r\le n}\sum_{k=0}^n p_{r,\sigma}(k)T_k$.
\end{lemma}
\begin{proof}
Choose a uniformly random permutation of the coordinates and run $\A$
on the permuted input, answering each query with the corresponding
coordinate of the actual input. This uses the same number of queries
as $\A$. Since $f$ is symmetric, permuting coordinates preserves its
value and hence the error bound. Any two inputs of the same Hamming
weight have the same distribution after this random permutation, so
the new expected cost depends only on the weight.

It remains to compare the smoothed costs. Coordinate permutations
commute with the perturbation, because every coordinate is perturbed
independently according to the same rule. Consequently, for each
original input $x$, the new smoothed cost is the average of the old
smoothed costs over permutations of $x$. This average is at most the
maximum smoothed cost of $\A$. Finally, if $|x|=r$, the perturbed
weight equals $k$ with probability $p_{r,\sigma}(k)$, giving the formula
in the statement.
\end{proof}

We use the following characterization of the cost of distinguishing
inputs of two different Hamming weights.
\begin{lemma}[\cite{nayak1999quantum,PodderYaoYe2025}]
\label{lem:two-layers}
For $0\le a<b\le n$, distinguishing inputs of Hamming weight $a$ from inputs of
Hamming weight $b$ has quantum query complexity
$\Theta(\sqrt{(n-a)b}/(b-a))$ and randomized query complexity
$\Theta(\min\{n,(n-a)b/(b-a)^2\})$.
\end{lemma}

For a fixed weight $k$, any algorithm computing $f$ must distinguish
weight $k$ from every weight $\ell$ with $g(\ell)\ne g(k)$. The hardest
such distinction motivates the following quantities:
\begin{equation}
q_f(k)=\max_{\ell:\,g(\ell)\ne g(k)}
\frac{\sqrt{(n-\min\{k,\ell\})\max\{k,\ell\}}}{|k-\ell|},
\qquad
r_f(k)=\min\{n,q_f(k)^2\}.
\label{eq:symmetric-layer-costs}
\end{equation}
We set both quantities to zero for constant $f$. For nonconstant $f$,
$q_f(k)$ and $r_f(k)$ are the quantum and randomized query complexities,
up to constant factors, of the hardest distinction involving weight
$k$. The next lemma shows that they also characterize the expected
cost of computing $f$ on inputs of weight $k$, up to logarithmic factors.
The upper bounds are achieved by a single algorithm in each model;
the algorithm is not given the input weight.

\begin{lemma}
\label{lem:symmetric-layer-costs}
Every symmetrized bounded-error algorithm computing $f$ has expected
query cost $\Omega(r_f(k))$ in the randomized model and $\Omega(q_f(k))$
in the quantum model on inputs of Hamming weight $k$. In each model, one algorithm attains
the corresponding upper bounds $\widetilde O(r_f(k))$ and
$\widetilde O(q_f(k))$ simultaneously for all $k$. Consequently,
$R_{\sm,\sigma}(f)=\widetilde\Theta(\max_r\sum_k p_{r,\sigma}(k)r_f(k))$
and
$Q_{\sm,\sigma}(f)=\widetilde\Theta(\max_r\sum_k p_{r,\sigma}(k)q_f(k))$.
\end{lemma}
\begin{proof}
For the lower bounds, fix $k$ and $\ell$ with $g(k)\ne g(\ell)$.
An algorithm computing $f$ solves the promise problem of distinguishing
these two weights: its output identifies which of $g(k)$ and $g(\ell)$
is the function value. By symmetrization, its expected cost on every
input of weight $k$ is the same value $T_k$.
Lemma~\ref{lem:expected-query-truncation} converts an algorithm with
expected cost $T_k$ on inputs of Hamming weight $k$ into an
$O(T_k)$-query algorithm distinguishing weights $k$ and $\ell$. Lemma~\ref{lem:two-layers}, followed by
maximization over $\ell$, gives both lower bounds. Notice that this
argument uses the expected cost on weight $k$; it does not require
the costs on weights $k$ and $\ell$ to be comparable.

For the upper bounds, assume $f$ is nonconstant. Both algorithms estimate
the Hamming weight with increasing precision, retain all weights
consistent with their estimates, and stop when $g$ is constant on the
retained set, outputting this common value. Initially every weight
from $0$ to $n$ is retained. A weight is removed as soon as it is
inconsistent with one of the estimates. If the true weight remains
in the retained set, stopping in this way gives the correct answer.
If that set is empty, the algorithms query the entire input. They
also query the entire input if the next stage would make the total cost
exceed $n$. Thus every run uses $O(n)$ queries.

\paragraph{Randomized algorithm.}
At stages $t=1,2,4,\ldots$, sample $t$ uniformly random coordinates with
replacement and let $\widehat k$ be $n$ times the fraction of ones observed.
Bernstein's inequality (Lemma~\ref{thm: Bernstein}) gives
$\lvert{\widehat k-k}\rvert
\le C(\sqrt{k(n-k)\log n/t}+n\log n/t)$
except with probability at most $(n+2)^{-12}$, for a suitable constant
$C$. A candidate weight is retained only if it satisfies this inequality
in place of $k$ for every estimate so far. If weights $a<b$ both satisfy
the inequality for one estimate, then
$b-a\le 2C(\sqrt{(n-a)b\log n/t}+n\log n/t)$.
Indeed, the distance between two consistent weights is at most the
sum of their estimation errors, and both $a(n-a)$ and $b(n-b)$ are
at most $(n-a)b$. Since $(n-a)b\ge n(b-a)$, the second error term
is also controlled at the number of samples needed for the first.
Thus the displayed inequality is impossible once
$t\ge C'(n-a)b\log n/(b-a)^2$, for a sufficiently large
$C'$. To apply this to an input of weight $k$, set
$a=\min\{k,\ell\}$ and $b=\max\{k,\ell\}$ for each weight $\ell$
with $g(\ell)\ne g(k)$. Maximizing the required number of samples
over these $\ell$ gives $O(q_f(k)^2\log n)$. If this exceeds the
cost of reading the input, the algorithm uses the exact procedure.
Hence all weights on which $g$ takes the opposite value are eliminated within
$\widetilde O(r_f(k))$ queries, or querying the entire input already gives this
bound.

\paragraph{Quantum algorithm.}
At each stage, use amplitude estimation
(Theorem~\ref{thm:amplitude_estimation}) on the uniform superposition
over input coordinates, with parameter $t=1,2,4,\ldots$. Amplifying
its success probability uses
$O(t\log n)$ queries and produces an estimate satisfying
$\lvert{\widehat k-k}\rvert
\le C(\sqrt{k(n-k)}/t+n/t^2)$
except with probability at most $(n+2)^{-12}$. Retain weights consistent
with every estimate so far and use the same stopping rule as in the
randomized algorithm.
If $a<b$ both remain consistent, then
$b-a\le 2C(\sqrt{(n-a)b}/t+n/t^2)$.
For $t\ge C'\sqrt{(n-a)b}/(b-a)$, the first term inside the parentheses
is at most $(b-a)/C'$. The second is at most $(b-a)/(C')^2$, since
$n(b-a)\le(n-a)b$. Choosing $C'$ sufficiently large therefore makes
the right-hand side strictly smaller than $b-a$, a contradiction.
As in the randomized case, maximizing over weights $\ell$ with
$g(\ell)\ne g(k)$ shows that all weights on which $g$ takes the opposite value are
eliminated within $\widetilde O(q_f(k))$ queries, or querying the entire input
gives this bound.

There are $O(\log n)$ stages in either algorithm, so a union bound
shows that all estimates are valid with
probability at least $1-(n+2)^{-2}$. On this event, the true weight is
always retained. Once every weight on which $g$ takes the opposite value has been removed,
the algorithm stops and outputs $g(k)$. Since the stage sizes double,
the total cost through any stage is at most a constant multiple of
the cost of that stage, apart from the stated amplification factor.
If reading the entire input becomes cheaper, the exact procedure has
the same upper bound. The exceptional event adds at most
$O(n)(n+2)^{-2}=O(1)$ to the expected cost, which is absorbed since
$q_f(k),r_f(k)\ge1$. Finally, averaging the lower bounds over the
perturbed weight distribution gives a lower bound for every algorithm.
The single algorithms just constructed achieve the corresponding upper
bounds for all weights simultaneously. Maximizing these averages over
$r$ and applying Lemma~\ref{lem:symmetrization} gives the two
characterizations in the statement.
\end{proof}

\subsection{Averaging over the perturbed weight}
\label{sec:symmetric-averaging}
We now average the costs for the different Hamming weights and express
the result in terms of the weighted influence and transition cost.
Lemma~\ref{lem:symmetric-layer-costs} reduces the query problem to
these averages. The remaining issue is that a small set of weights
near transitions can have much larger cost than weights near the mean.
The following lemma shows that $\transcost_\sigma(f)$ suffices for the
randomized bound, while the quantum bound also requires
$\wInf_\sigma(f)$ to account for contributions from several transitions.

\begin{lemma}
\label{lem:localization}
For every total symmetric Boolean function $f$ and every $\sigma\in[0,1]$,
\begin{align}
\max_r\sum_k p_{r,\sigma}(k)r_f(k)
&=\widetilde\Theta\!\left(\min\{n,\transcost_\sigma(f)^2\}\right),
\label{eq:randomized-localization}\\
\max_r\sum_k p_{r,\sigma}(k)q_f(k)
&=\widetilde\Theta\!\left(\wInf_\sigma(f)+\transcost_\sigma(f)\right).
\label{eq:quantum-localization}
\end{align}
\end{lemma}
\begin{proof}
Assume $f$ is nonconstant. Fix $r$ and abbreviate
$p(k)=p_{r,\sigma}(k)$, $\mu=\mu_{r,\sigma}$, and $s=s_\sigma$.
For this proof, write
$\wInf_r=\sum_{j\in\mathcal T(f)}(p(j)+p(j+1))t_j$ and
$\transcost_r=\max_{j\in\mathcal T(f)}q_{j,\sigma}(\mu)$ for the
corresponding costs before maximizing over the original input's Hamming weight.
We first prove the claimed comparisons for this fixed $r$, and take
the maximum over $r$ at the end.

For a weight $k$, let $I(k)$ be the maximal interval containing $k$ on
which $g$ is constant, and let $\partial I(k)$ be its existing boundary
transitions. If $I(k)=[a,b]$, these are $a-1$ when $a>0$ and $b$ when
$b<n$. The nearest weight to the left of $k$ with a different function
value, if one exists, is $a-1$; the nearest such weight to the right
is $b+1$. For fixed $k$, the cost of distinguishing it from another
weight decreases as that weight moves farther away on either side.
For these two boundary transitions, $q_{j,0}(k)$ is exactly the cost
of the corresponding distinction. Transitions farther from $k$ give
no larger value, whether or not their nearest endpoint has the
opposite function value. Consequently,
$q_f(k)=\max_{j\in\partial I(k)}q_{j,0}(k)
=\max_{j\in\mathcal T(f)}q_{j,0}(k)$.

If $\sigma=0$, then $p$ is a point mass at $r$ and
$\transcost_r=q_f(r)$. Also $\wInf_r\le2q_f(r)$, so both claims
follow immediately. Henceforth assume $\sigma>0$.

\paragraph{Lower bounds.}
We first show that the cost of a typical perturbed weight is at least
a constant multiple of $\transcost_r$. We have
$\min\{\mu,n-\mu\}\ge (s-1)^2$. If $|k-\mu|\le2s$, the
two factors in the numerator of $q_{j,0}(k)$ are within constant factors
of those at $\mu$, uniformly in $j$: each changes by at most $2s$, each
is at least one, and at $\mu$ each is at least $(s-1)^2$.
To see the constant-factor comparison, if $(s-1)^2\ge64$, changing
either factor by at most $2s$ changes it by at most half its value.
If $(s-1)^2<64$, the change is bounded by an absolute constant, and
both factors remain at least one. Moreover,
$1+\dist(k,[j,j+1])\le3(s+\dist(\mu,[j,j+1]))$. Thus
$q_{j,0}(k)\ge c\,q_{j,\sigma}(\mu)$ whenever $|k-\mu|\le2s$.
Chebyshev's inequality puts at least $3/4$ of the probability in this
interval. Choosing a transition attaining $\transcost_r$ gives
$\sum_k p(k)q_f(k)=\Omega(\transcost_r)$ and
$\sum_k p(k)r_f(k)=\Omega\!\left(\min\{n,\transcost_r^2\}\right)$.
The weighted influence gives another quantum lower bound. For each
transition $j$, an algorithm must distinguish its two adjacent weights,
so $q_f(j),q_f(j+1)\ge t_j$. When these inequalities are weighted by
$p(j)$ and $p(j+1)$ and summed over transitions, any weight is counted
at most twice. Consequently
$\wInf_r\le2\sum_k p(k)q_f(k)$.

\paragraph{Randomized upper bound.}
Let $d=\min_{j\in\mathcal T(f)}\dist(\mu,[j,j+1])$ be the distance
to the nearest transition, and fix a sufficiently large constant $C_0$.
We distinguish whether this transition lies within $C_0s\log n$ of
the mean. If it does, the transition cost is already large enough
that the trivial $n$-query upper bound suffices. If it does not, the
perturbed weight stays far from every transition with high probability,
and its cost is comparable to the cost at the mean.
For every $j$,
\[
(n-\min\{j,\mu\})\max\{j+1,\mu\}
\ge\max\{n,\mu(n-\mu)\}
\ge cns^2.
\]
Here the last inequality follows from
$\mu(n-\mu)\ge n(s-1)^2/2$ and $s^2\le2(1+(s-1)^2)$.
If $d\le C_0s\log n$, a nearest transition therefore gives
$\transcost_r\ge c\sqrt n/\log n$. The trivial bound $r_f(k)\le n$
then proves the desired upper bound, since
$\min\{n,\transcost_r^2\}\ge cn/(\log n)^2$.

If $d>C_0s\log n$, Lemma~\ref{lem:poisson-binomial-facts} gives
$\Pr[|W_{r,\sigma}-\mu|>d/2]\le(n+2)^{-10}$. On the complementary
event, every transition remains on the same side of the weight, and its
distance and numerator in~\eqref{eq:symmetric-transition-cost} change
by at most constant factors. Indeed, for a transition to the left of
the mean, only the factor $\mu$ in the numerator varies, and for a
transition to the right, only the factor $n-\mu$ varies. In either
case that factor is at least the distance to the transition, so a
change of at most $d/2$ changes it by at most a constant factor.
Since $\dist(\mu,[j,j+1])>C_0s\log n$, we obtain
$q_{j,0}(W_{r,\sigma})=O(q_{j,\sigma}(\mu))$ for every $j$.
The transition identity now gives
$q_f(W_{r,\sigma})=O(\transcost_r)$. The exceptional event contributes
at most $n(n+2)^{-10}$. This is negligible: the numerator of
$q_{j,\sigma}(\mu)$ is at least $c\sqrt n\,s$ and at least
$1+\dist(\mu,[j,j+1])$, so $\transcost_r=\Omega(1)$. We have proved
$\sum_k p(k)r_f(k)
=\widetilde\Theta\!\left(\min\{n,\transcost_r^2\}\right)$.

\paragraph{Quantum upper bound.}
The quantum bound requires a separate estimate of the contributions
near transitions. We first split the cost into a term that decreases
inversely with the distance to a transition and a term controlled by
the randomized bound. For a left boundary $j$ of $I(k)$,
\[
q_{j,0}(k)=\frac{\sqrt{(n-j)k}}{k-j}
\le\frac{t_j}{k-j}+\sqrt{\frac{n-j}{k-j}}
\le\frac{t_j}{k-j}+\sqrt{r_f(k)}.
\]
The last inequality holds because the square-root term is at most both
$\sqrt n$ and $q_f(k)$. The right boundary is symmetric, so
$q_f(k)\le\max_{j\in\partial I(k)}\frac{t_j}{1+\dist(k,[j,j+1])}
+\sqrt{r_f(k)}$.
By Cauchy--Schwarz and the randomized bound, the
average of the last term is $\widetilde O(\transcost_r)$. It remains
to prove
\begin{equation}
\sum_k p(k)\max_{j\in\partial I(k)}\frac{t_j}{1+\dist(k,[j,j+1])}
\le C\log n\bigl(\wInf_r+\transcost_r\bigr).
\label{eq:profile-bound}
\end{equation}

To prove this inequality, partition the weights into the maximal
intervals on which $g$ is constant. On either side of the mode of
$p$, the probabilities are monotone. This lets us bound the contribution
of an interval using its total probability and the probability at
the endpoint closer to the mode. The former is controlled by
$\transcost_r$, and the latter contributes to $\wInf_r$.

Let $m$ be a mode of $p$ with $|m-\mu|<1$, as supplied by
Lemma~\ref{lem:poisson-binomial-facts}. Partition the weights into the
maximal intervals on which $g$ is constant, and write
$p(I)=\sum_{k\in I}p(k)$. We bound the two boundary contributions
separately, using that their maximum is at most their sum. A missing
boundary is omitted from every expression below.

Consider first an interval $I=[a,b]$ strictly to the left of $m$.
For its right boundary, $p(k)\le p(b)$, and summing
$1/(b-k+1)$ gives a factor $O(\log n)$. For its left boundary,
$p(k)$ is nondecreasing while $1/(k-a+1)$ is nonincreasing.
Lemma~\ref{lem:monotone-sequence} therefore bounds their weighted
sum by $O(p(I)\log n/|I|)$. Together these give
\[
\sum_{k\in I}p(k)
\left(\frac{t_{a-1}}{k-a+1}+\frac{t_b}{b-k+1}\right)
\le C\log n
\left(\frac{t_{a-1}}{|I|}p(I)+p(b)t_b\right).
\]
If $|I|\ge s+m-b+2$, then $s+\dist(\mu,[a-1,a])=O(|I|)$, so
$t_{a-1}/|I|=O(\transcost_r)$. Here we used that the numerator of
$q_{a-1,\sigma}(\mu)$ is at least $t_{a-1}$. In this case the left
boundary is controlled by the total probability of the interval.
Otherwise, $p(I)\le |I|p(b)$ and
$t_{a-1}=O(t_b)$, since
\[
\frac{t_{a-1}^2}{t_b^2}\le1+\frac{|I|}{n-b}=O(1),
\qquad
n-b\ge\max\{1,(s-1)^2+m-b-1\}.
\]
Thus this interval contributes at most
$C\log n(\transcost_r\,p(I)+p(b)t_b)$.
Reflection gives the corresponding bound
$C\log n(\transcost_r\,p(I)+p(a)t_{a-1})$ for intervals to the
right of $m$.

There is one remaining interval $I=[a,b]$, namely the interval containing
$m$. Its probabilities need not be monotone over the entire interval,
so we split the sum at $m$. Consider its left boundary.
If $m-a+1\ge s$, apply Lemma~\ref{lem:monotone-sequence} on $[a,m]$;
on $[m+1,b]$, the denominator is at least $m-a+1$, so that part of
the sum is at most $t_{a-1}p(I)/(m-a+1)$. Since
$t_{a-1}/(m-a+1)=O(\transcost_r)$, their total contribution is
$O(\transcost_r\,p(I)\log n)$. If $m-a+1<s$, use
$p(k)\le C/s$ and $t_{a-1}/s=O(\transcost_r)$ to obtain
$O(\transcost_r\log n)$ by summing the harmonic series.
The right boundary is handled in the same way after reflection.
When we sum over all intervals, the terms proportional to $p(I)$
sum to at most $O(\transcost_r\log n)$, because the intervals partition
the weights. Each endpoint term is one of the summands
$p(j)t_j$ or $p(j+1)t_j$ in $\wInf_r$. The interval containing the
mode contributes at most another $O(\transcost_r\log n)$. This proves
\eqref{eq:profile-bound}.
Together with the comparison above and the lower bounds, this proves
$\sum_k p(k)q_f(k)
=\widetilde\Theta\!\left(\wInf_r+\transcost_r\right)$.
Maximize both bounds over $r$. For the quantum
bound, use that the maximum of a sum of two nonnegative quantities is
within a factor of two of the sum of their maxima. This proves the lemma
and, with Lemma~\ref{lem:symmetric-layer-costs},
Theorem~\ref{thm:symmetric}.
\end{proof}

\subsection{Examples and comparisons}
\label{sec:symmetric-examples}
We now apply Theorem~\ref{thm:symmetric} to compare the two query models
at $\sigma=0$ and $\sigma=1$. We consider \textsc{Or}, the thresholds
at $\sqrt n$ and $n/4$, \textsc{Majority}, and \textsc{Parity}.
A threshold at $t$ outputs one exactly when $|x|\ge\lceil t\rceil$.
These examples
illustrate both the location of the transitions and the effect of
having many transitions. In Tables~\ref{tab:symmetric-sigma-zero}
and~\ref{tab:symmetric-sigma-one}, each entry in a parameter column
denotes a $\Theta$ bound, and each entry in a query complexity column
denotes a $\widetilde\Theta$ bound.

\paragraph{No perturbation: $\sigma=0$.}
The perturbed weight equals the original Hamming weight deterministically:
$W_{r,0}=r$, $\mu_{r,0}=r$, and $s_0=1$. Choosing an original input
adjacent to a hardest transition gives
$\transcost_0(f)=\max_{j\in\mathcal T(f)}t_j$. A fixed weight is adjacent
to at most two transitions, so $\wInf_0(f)$ lies between this maximum
and twice its value. Hence Theorem~\ref{thm:symmetric} recovers
$Q_{\sm,0}(f)=\widetilde\Theta(\max_j t_j)$ and
$R_{\sm,0}(f)=\widetilde\Theta(n)$ for nonconstant $f$, since
$t_j^2=(n-j)(j+1)\ge n$.

For \textsc{Or}, the only transition is $j=0$, with $t_0=\sqrt n$.
The threshold at $\sqrt n$ has its transition at $j=\lceil\sqrt n\rceil-1$,
so $t_j=\sqrt{(n-\lceil\sqrt n\rceil+1)\lceil\sqrt n\rceil}
=\Theta(n^{3/4})$. The threshold at $n/4$ has its transition at
$j=\lceil n/4\rceil-1$, with $t_j=\Theta(n)$. The transition of
\textsc{Majority} is near $n/2$, where $t_j=\Theta(n)$ as well.
\textsc{Parity} has a transition at every weight and therefore also
has transitions of cost $\Theta(n)$.
These observations give Table~\ref{tab:symmetric-sigma-zero}.

\begin{table}[htbp]
\centering
\renewcommand{\arraystretch}{1.25}
\begin{tabular}{lcccc}
\hline
$f$ & $\transcost_0(f)$ & $\wInf_0(f)$ & $R_{\sm,0}(f)$ & $Q_{\sm,0}(f)$\\
\hline
\textsc{Or} & $\sqrt n$ & $\sqrt n$ & $n$ & $\sqrt n$\\
Threshold at $\sqrt n$ & $n^{3/4}$ & $n^{3/4}$ & $n$ & $n^{3/4}$\\
Threshold at $n/4$ & $n$ & $n$ & $n$ & $n$\\
\textsc{Majority} & $n$ & $n$ & $n$ & $n$\\
\textsc{Parity} & $n$ & $n$ & $n$ & $n$\\
\hline
\end{tabular}
\caption{The two parameters and query complexities when $\sigma=0$.}
\label{tab:symmetric-sigma-zero}
\end{table}

\paragraph{Uniform input: $\sigma=1$.}
Every coordinate is rerandomized, so the perturbed input is uniform
on $\{0,1\}^n$, independently of the original input. Consequently,
$W_{r,1}\sim\Bin(n,1/2)$, $\mu_{r,1}=n/2$, and
$s_1=1+\sqrt n/2$. In particular,
$p_{r,1}(k)=2^{-n}\binom nk$, and the maximization over $r$ has no
effect. The location of a transition relative to $n/2$ now matters:
weights within $O(\sqrt n)$ of the mean carry most of the probability,
while transitions a linear distance from the mean have exponentially
small adjacent probabilities.

For \textsc{Or}, the transition is at $j=0$, a distance $\Theta(n)$
from the mean. The numerator and denominator in
\eqref{eq:symmetric-g} are both $\Theta(n)$, giving
$\transcost_1(\textsc{Or})=\Theta(1)$.
Theorem~\ref{thm:symmetric} therefore gives
$R_{\sm,1}(\textsc{Or})=\widetilde\Theta(1)$ and
$Q_{\sm,1}(\textsc{Or})=\widetilde\Theta(1)$. In fact, the randomized
algorithm that reads coordinates until it finds a one, or has read
the entire input, has constant expected cost on a uniform input.

For the threshold at $\sqrt n$, the transition lies a distance
$\Theta(n)$ from the mean, so its transition cost is $\Theta(1)$.
Its weighted influence is exponentially small in $n$, and both query
complexities are $\widetilde\Theta(1)$.

For the threshold at $n/4$, the transition lies a distance $\Theta(n)$
from the mean. As for \textsc{Or}, this gives transition cost
$\Theta(1)$. Its weighted influence is exponentially small in $n$.
Both query complexities are
$\widetilde\Theta(1)$. Thus a threshold away from the mean can become
easy in both models, even though its worst-case complexities are
both linear up to logarithmic factors.

For \textsc{Majority}, the single transition is within a constant
distance of the mean. The denominator in \eqref{eq:symmetric-g}
is $\Theta(\sqrt n)$ and its numerator is $\Theta(n)$, so
$\transcost_1(\textsc{Majority})=\Theta(\sqrt n)$.
The quantum complexity becomes $\widetilde\Theta(\sqrt n)$,
while the randomized complexity remains $\widetilde\Theta(n)$.
Compared with $\sigma=0$, full smoothing therefore gives a polynomial
improvement only in the quantum model.

\textsc{Parity} has the same transition cost
$\transcost_1(\textsc{Parity})=\Theta(\sqrt n)$, because a transition
near the mean attains this value and no transition gives a larger
order of growth. Its weighted influence is
$\wInf_1(\textsc{Parity})=\Theta(n)$, so both query complexities
remain $\widetilde\Theta(n)$.

\begin{table}[htbp]
\centering
\renewcommand{\arraystretch}{1.3}
\begin{tabular}{lcccc}
\hline
$f$ & $\transcost_1(f)$ & $\wInf_1(f)$ & $R_{\sm,1}(f)$ & $Q_{\sm,1}(f)$\\
\hline
\textsc{Or} & $1$ & $n^{3/2}2^{-n}$ & $1$ & $1$\\
Threshold at $\sqrt n$ & $1$ & $n^{3/4}2^{-n}\binom{n}{\lceil\sqrt n\rceil}$ & $1$ & $1$\\
Threshold at $n/4$ & $1$ & $n2^{-n}\binom{n}{\lceil n/4\rceil}$ & $1$ & $1$\\
\textsc{Majority} & $\sqrt n$ & $\sqrt n$ & $n$ & $\sqrt n$\\
\textsc{Parity} & $\sqrt n$ & $n$ & $n$ & $n$\\
\hline
\end{tabular}
\caption{The two parameters and query complexities when $\sigma=1$.}
\label{tab:symmetric-sigma-one}
\end{table}

The comparison between \textsc{Majority} and \textsc{Parity} shows why
the transition cost alone does not characterize quantum smoothed
query complexity. Both functions have transition cost $\Theta(\sqrt n)$
under the uniform distribution, but the weighted influence differs
by a factor of $\sqrt n$. For \textsc{Or} and the thresholds at
$\sqrt n$ and $n/4$, the weighted influence is exponentially small
and the transition cost
determines the bound. Throughout these comparisons, correctness is
required on every input; only the expected query cost is averaged
over the perturbation.

\section{Pattern matching in the smoothed model}\label{sec:smoothedpattern}
In this section, we consider the quantum query complexity of the smoothed pattern matching problem. We consider two variants of the problem. In the first setting, both the pattern and the text are accessed through queries, whereas in the second setting, the pattern is given explicitly and only the text is accessed through queries. 

Throughout this section, for a string $S:[\ell]\to\Sigma$, let $N_\sigma(S)$ denote the distribution obtained by independently rerandomizing each character of $S$ with probability $\sigma$, where a rerandomized character is chosen uniformly from $\Sigma$. We also let $\rho=\sigma(1-1/q)$.

\begin{definition}[Pattern matching]\label{def:patternmatching}
    Let $n, m, q \geq 2$ be positive integers with $m \leq n$. 
    Let $\Sigma$ be a finite alphabet of size $q$.     Let  $T:[n]\rightarrow\Sigma$ 
    denote a text and $P:[m]\rightarrow\Sigma$     denote a pattern.
    Let $\occ(T,P)=\left\{s\in\{0,\ldots,n-m\}:T_{s+j}=P_j\text{ for every }j\in[m]    \right\}$ denote the set of occurrences of $P$ in $T$. 
    Define the pattern matching function $\mathsf{match}(T,P)=\one\{\occ(T,P)\ne\textsf{Var}nothing\}$. 
    We will also write $\mathsf{match}(T,P)$ as $\mathsf{match}_P(T)$ for applications where the pattern is fully known.
\end{definition}

We will also consider the corresponding search version ofthe  pattern matching problem defined above: whenever $\operatorname{match}(T,P)=1$, in addition to deciding that a match exists, the goal is to output any
$s^*\in\operatorname{occ}(T,P)$.

We use the following simple observation about the perturbation $N_\sigma$ in both settings.

\begin{lemma}\label{lem:coordinate} 
    Let $\Sigma$ be the alphabet of size $q$, $\sigma>0$, and $\rho = \sigma(1-1/q)$.
    Let $a, a'\in \Sigma$. If $b \sim N_\sigma(a)$, then $\Prb[b\ne a']\ge\rho$ and $\Prb[b=a']\le1-\rho$. 
\end{lemma}

\begin{proof}
    If \(a=a'\), then $b \neq a'$ exactly when the $a$ is rerandomised and the new symbol is not \(a'\), whose probability is \(\sigma(1-1/q)=\rho\). If \(a\ne a'\), the only way to have \(b=a'\) is to rerandomise and draw \(a'\), which has probability $\sigma/q\le 1-\sigma+\sigma/q=1-\rho$.
\end{proof}

\subsection{Smoothed pattern matching}
\begin{definition}[Smoothed Pattern matching]
\label{defn:smoothedpatmatch}
    Let $T:[n]\rightarrow \Sigma, P:[m]\rightarrow \Sigma$. Consider $\widetilde{T}\sim N_\sigma(T)$. Define $\widetilde{P}$ as follows:
    \begin{enumerate}
        \item If $\operatorname{occ}(T,P)\neq\emptyset$, fix any $s^*\in\operatorname{occ}(T,P)$ and let $\widetilde{P}(x)=\widetilde{T}(x+s^*)$ for all $x\in [m]$.
        \item If $\operatorname{occ}(T,P)=\emptyset$, let $\widetilde{P}\sim N_\sigma(P)$. 
    \end{enumerate}
    The goal is to distinguish which is the case given access to $\widetilde{P},\widetilde{T}$ with high probability (where the probability is taken over $\widetilde{P},\widetilde{T}$). In case $(1)$, one needs to output $s^*$.
\end{definition} 

Recall that if $\sigma=0$, then $\widetilde{T}=T$ and $\widetilde{P}=P$, so the problem above reduces exactly to the corresponding search version of pattern matching in Definition~\ref{def:patternmatching}. Below, we show that in the smoothed setting defined above, one can obtain a superpolynomial quantum speedup (inspired by Montanaro's algorithm~\cite{montanaro2017quantum}).

\begin{theorem}
\label{thm:smoothedpattern}
Let $\Sigma$ be the alphabet of size $q$, $\sigma>0$, $\rho = \sigma(1-1/q)$ Let $T:[n]\rightarrow \Sigma$ and $P:[m]\rightarrow \Sigma$ be text and pattern respectively.  Suppose that $\rho m= \Omega (\log n)$.  There is a $O(\sqrt{n/m}\cdot 2^{O(\sqrt{\log m})}\cdot 1/\sigma)$-query quantum algorithm for smoothed pattern matching in Definition~\ref{defn:smoothedpatmatch}.
\end{theorem}
Proving Theorem~\ref{thm:smoothedpattern} is fairly simple. To prove this, we will use the following general theorem by Montanaro~\cite{montanaro2017quantum} (which applies to all texts and patterns) which we analyze for smoothed~inputs. 
We need one notation before that. For a string $S:[n]\rightarrow \Omega$ and $k\leq n$, define $S^{\rhd k}:[n-k]\rightarrow \Omega^k$~as:  
$$
S^{\rhd k}(s)=\Big(S(s+1),\ldots,S(s+k)\Big) \text{ for every } s\in [n-k].
$$
Define the \emph{injectivity length} of a string $S$ as\footnote{Below, by distinct, here we mean that $S^{\rhd k}(s)\neq S^{\rhd k}(s')$ for all $s\neq s'$ where $S^{\rhd k}(s)$ is viewed as a string in $\Omega^k$.}
$$
\nu(S)=\min\{k: S^{\rhd k} \text{ is non-repetitive, i.e., all values in }S^{\rhd k} \text{ are distinct}\}.
$$
We are now ready to state the theorem that we will use. \begin{theorem}{\cite{montanaro2017quantum}}
\label{thm:montanarothm2}
Fix $d=O(1)$. Let $T:[n]\rightarrow \Sigma$, $P:[m]\rightarrow \Sigma$. Suppose $\nu$ is such that $\nu(T),\nu(P)\leq \nu\leq m/2$. Furthermore, suppose that for every offset $s$ such that $P$ does not match $T$, the fraction of $i\in [m]$ such that $P(i)\neq T(i+s)$ is at least $\gamma$. Then, there is a bounded-error quantum algorithm that outputs $s\in [m]$ such that $P$ matches $T$ at offset $s$ (if such an $s$ exists) and otherwise requires not found.  The algorithm uses in total
$$
\widetilde{O}\Big(\sqrt{n/m}\cdot 2^{O(\sqrt{\log m})}\cdot (\nu+1/\sqrt{\gamma})\Big)
$$
queries to $T,P$. \end{theorem}
At this point it suffices to bound $\nu$ in order to understand the complexity of the algorithm above. Montanaro shows that, for a \emph{random} $P,T$, we have that $\nu\leq O(\log_{|\Omega|}n)$ with  probability $\geq 1-1/n$. Below we show this is true even for smoothed inputs.
\begin{lemma}
\label{lem:injectiveonsmoothing}
    Let $\sigma>0$ and $|\Omega|\geq 2$. Let $S:[n]\rightarrow \Omega$. Let $\widetilde{S}=N_\sigma(S)$. Then
    $$
    \Pr_{\widetilde{S}\sim N_\sigma(S)}[\nu(\widetilde{S})\geq  \lceil\frac{12\log n)}{\sigma}\rceil]\leq \frac1n.
    $$
\end{lemma}
    
\begin{proof}
    For notational simplicity, let $|\Omega|=q$. Fix $k\in[n]$ and consider $S^{\rhd k}$.     Also fix $s,t\in[n]$ such that $1\leq s<s+t\leq n-k+1$.

    Among the $k$ pairs $\{i,i+t\}$, for $i\in s+[k]$, we can select at least $k/2$ pairs such that no two selected pairs share an index. Let $I\subseteq s+[k]$ denote the corresponding set of indices. Therefore, if the two length-$k$ substrings are equal, then $\widetilde S(i)=\widetilde S(i+t)$ for every $i\in I$.
    Hence
    \begin{align*}
        &\Pr_{\widetilde{S}\sim N_\sigma(S)}
        [\widetilde{S}^{\rhd k}(s)=\widetilde{S}^{\rhd k}(s+t)]\leq
        \Pr_{\widetilde{S}\sim N_\sigma(S)}
        \left[
        \bigcap_{i\in I}
        \{\widetilde{S}(i)=\widetilde{S}(i+t)\}
        \right]=
        \prod_{i\in I}
        \Pr_{\widetilde{S}\sim N_\sigma(S)}
        [\widetilde{S}(i)=\widetilde{S}(i+t)].
    \end{align*}
   The last equality follows because no two selected pairs share an index, so the corresponding events are independent under the perturbation $N_\sigma$.

    For any two characters $a,b\in\Omega$, if    $\widetilde a\sim N_\sigma(a)$ and $\widetilde b\sim N_\sigma(b)$ independently, a direct calculation gives
    \[
        \Pr[\widetilde a=\widetilde b]
        \leq
        1-\sigma(1-1/q)(2-\sigma)
        \leq
        1-\frac{\sigma}{2},
    \]
    where the first bound is attained when $a=b$, and the last inequality uses $q\geq2$.
    Therefore,
    \[
        \Pr_{\widetilde{S}\sim N_\sigma(S)}
        [\widetilde{S}^{\rhd k}(s)=\widetilde{S}^{\rhd k}(s+t)]
        \leq
        \left(1-\frac{\sigma}{2}\right)^{k/2}
        \leq
        e^{-\sigma k/4}.
    \]

    Now, taking a union bound over all possible choices of the two starting positions, we obtain an upper bound
    $     n^2 e^{-\sigma k/4}$ on the probability that $\widetilde{S}^{\rhd k}$ fails to be non-repetitive.
    By choosing
    $k=\left\lceil\frac{12\log n}{\sigma}\right\rceil$,
    this probability is at most $1/n$, proving the claim.
\end{proof}

\begin{corollary}
\label{cor:bothinstancessmooth}
In the context of Definition~\ref{defn:smoothedpatmatch}, in both instances of the problem we have $\nu(\widetilde{P})\leq O((\log n)/\sigma)$ and $ \nu(\widetilde{T})\leq O((\log n)/\sigma)$ with high probability.  
\end{corollary}
\begin{proof}
    Note that in Definition~\ref{defn:smoothedpatmatch}, $\widetilde{T}=N_\sigma(T)$. By Lemma~\ref{lem:injectiveonsmoothing},  we have that $\nu(\widetilde{T})\leq O((\log n)/\sigma)$ with probability $\geq 1-1/n$.    In case $(1)$, we have that $\widetilde{P}$ is a substring of $\widetilde{T}$, so we have that $\nu(\widetilde{P})\leq \nu(\widetilde{T})\leq O((\log n)/\sigma):=k$. Indeed, this follows because if $\widetilde{T}^{\rhd k}(s)\neq \widetilde{T}^{\rhd k}(s')$ (for $s\neq s'$), i.e., it is non-repetitive, then for \emph{every} substring of $\widetilde{T}$, we have that the similar inequality is true (i.e., substrings $P$ of $\widetilde{T}$ will also be non-repetitive or  $P^{\rhd k}(s)\neq P^{\rhd k}(s')$).   
    In case $(2)$, observe that $\widetilde{P}=N_\sigma(P)$, so by Lemma~\ref{lem:injectiveonsmoothing} we have that $\nu(\widetilde{T})\leq O((\log m)/\sigma)$ with probability $\geq 1-1/m$.
\end{proof}

We also will use the following fact.

\begin{fact}
Let $T:[n]\to\Sigma$ and $P':[m]\to\Sigma$ be arbitrary, and let
$\widetilde T\sim N_\sigma(T)$. Then, with probability at least
\[
1-n\exp(-\rho m/8),
\]
for every offset $s\in\{0,\ldots,n-m\}$,
\[
\left|\{j\in[m]:P'(j)\neq \widetilde T(s+j)\}\right|
\geq \frac{\rho m}{2}.
\]
\end{fact}

\begin{proof}
Fix an offset $s\in\{0,\ldots,n-m\}$, and let
\[
X_s
:=
\left|\{j\in[m]:P'(j)\neq \widetilde T(s+j)\}\right|.
\]
By Lemma~\ref{lem:coordinate}, for every $j\in[m]$,
\[
\Pr[P'(j)\neq \widetilde T(s+j)]\geq \rho.
\]
Since the perturbations are independent across positions,
$\mathbb E[X_s]\geq \rho m$. By the Chernoff bound (Lemma~\ref{lem:chernoff2}),
\[
\Pr\left[X_s<\frac{\rho m}{2}\right]
\leq
\exp(-\rho m/8).
\]
Taking a union bound over all at most $n$ possible offsets proves the claim.
\end{proof}

\begin{proofof}{Theorem~\ref{thm:smoothedpattern}}
The proof follows immediately from the lemma and fact above. 
Corollary~\ref{cor:bothinstancessmooth} implies that in both instances of the problem,
$\widetilde{P},\widetilde{T}$ have injectivity length at most $\nu=O((\log n)/\sigma)$,
with probability $\geq 1-1/n$ over the smoothing randomness, so let us assume this is the case.

In case one, by definition $\widetilde{P}$ is a substring of $\widetilde{T}$. Moreover, since
$\nu(\widetilde{T})\leq \nu$, for every offset at which $\widetilde{P}$ does not match
$\widetilde{T}$, every length-$\nu$ block contains a disagreement. Hence,
$\widetilde{P}$ and the corresponding substring of $\widetilde{T}$ disagree on at least an
$\Omega(1/\nu)$-fraction of positions.

In case two, conditioning on any fixed $\widetilde{P}$, by the fact above,
$\widetilde{P}$ and the corresponding substring of $\widetilde{T}$ disagree on at least a $\rho/2$-fraction of positions for every offset, with probability at least $1-ne^{-\rho m/8}.$
If $\rho m=\Omega(\log n)$, this probability is also at least $1-1/n$.
By a union bound, we assume both these events hold. Since $q\geq 2$, we have
$\rho=\Omega(\sigma)$, and hence we have $\gamma=\Omega\left(\min\{1/\nu,\rho\}\right)
=\widetilde{\Omega}(\sigma)$.
Plugging this and $\nu=O((\log n)/\sigma)$ into
Theorem~\ref{thm:montanarothm2}, we get that using
\[
\widetilde{O}\Big(
\sqrt{n/m}\cdot 1/\sigma\cdot
2^{O(\sqrt{\log m})}
\Big)
\]
queries, one can distinguish whether $\widetilde{P},\widetilde{T}$ belong to case one or case two,
and in case one output a matching offset $s^*$.
\end{proofof}

\begin{remark}
Note that Montanaro~\cite{montanaro2017quantum} showed an $\widetilde{\Omega}(n/m+\sqrt n)$ classical query lower bound for average-case pattern matching. On the other hand, in the regime $m=\Theta(n)$, our smoothed quantum algorithm uses
$n^{o(1)}/\sigma$ 
queries. Hence, for any constant $\varepsilon>0$, if $\sigma\geq n^{-1/2+\varepsilon}$, our query complexity is $n^{1/2-\varepsilon+o(1)}$, giving a polynomial quantum advantage over the $\widetilde{\Omega}(\sqrt n)$ classical average-case bound.
Moreover, if $\sigma=n^{-o(1)}$, our quantum query complexity is $n^{o(1)}$, yielding a superpolynomial separation from this classical average-case complexity.
\end{remark}

\subsection{Smoothed pattern matching when pattern is known}

\Cref{thm:smoothedpattern} shows a trade-off for smoothed complexity for all $\sigma$. However, when $\sigma$ is small (i.e. $\Theta(\log n/m)$ and $m$ is large (i.e. $\Theta(n)$), the complexity of this algorithm is $n^{1\pm o(1)}$. However, when $\sigma = 0$ (i.e. the worst-case), we know that the complexity of pattern matching is $\tilde{O}(\sqrt{n})$ \cite{ramesh2003string}. This suggests that the dependency of $\sigma$ in this algorithm may not be tight. In this subsection, we consider a variant of pattern matching where the pattern is explicitly given. This model is motivated to be a generalization of Grover's search problem where the searched-for object (i.e. a length $1$ pattern) is known explicitly. Interestingly, in this model, we can tightly characterize the smoothed complexity of pattern matching upto polylogarithmic factors for a large range of $\sigma$. Recall that $\rho = \sigma(1-1/q)$. We will consider the range of $\sigma$ for which $\rho m \geq c \log n$ for some large enough constant $c$. When $\rho m << \log n$, the perturbed text will likely contain length-$m$ consecutive blocks which we do not go through any rerandomization so the smoothed pattern matching in this regime will be quite similar to worst-case pattern matching, which is very well-understood \cite{yao1979complexity, ramesh2003string}. In fact, when $\rho m << \log n$, the \textsc{Yes} and \textsc{No} instances in  \Cref{defn:smoothedpatmatchwithpatternknown} will likely have a non-trivial intersection, which will make the distinguishing task unclear.   
    
\begin{definition}[Smoothed pattern matching for an explicit pattern]
\label{defn:smoothedpatmatchwithpatternknown}
    Let $T:[n]\rightarrow \Sigma$, $P:[m]\rightarrow \Sigma$. Define $\widetilde{T}$:
    \begin{enumerate}
        \item If $\operatorname{occ}(T,P)\neq\emptyset$, fix any $s^*\in\operatorname{occ}(T,P)$ and let $\widetilde{T}[1,s^{*}] \sim N_\sigma(T[1,s^{*}])$, $\widetilde{T}[s^{*}+m+1,n] \sim N_\sigma(T[s^{*}+m+1,n])$ and  $\widetilde{T}(x+s^*)={P}(x)$ for all $x\in [m]$.
        \item If $\operatorname{occ}(T,P)=\emptyset$, let $\widetilde{T}\sim N_\sigma(T)$. 
    \end{enumerate}
    The goal is to distinguish which is the case with high probability (where the probability is taken over $\widetilde{T}$) given $P$ for free and query access to $\widetilde{T}$. In case $(1)$, one needs to output $s^*$.
\end{definition}

Furthermore, for any $P \in \Sigma^m$ and $T \in \Sigma^n$, let
    $\cR_\sigma(\mathsf{match}_P,T)$ and $\cQ_\sigma(\mathsf{match}_P,T)$
    denote the randomized and quantum query complexities, respectively, of deciding whether
    $\mathsf{match}(\widetilde T,P)=1$, where $\widetilde T \sim N_\sigma(T)$.

    For any $P \in \Sigma^m$, let
    \[
    \cR_\sigma(\mathsf{match}_P)
    =
    \max_{T \in \Sigma^n} \cR_\sigma(\mathsf{match}_P,T),
    \qquad
    \cQ_\sigma(\mathsf{match}_P)
    =
    \max_{T \in \Sigma^n} \cQ_\sigma(\mathsf{match}_P,T).
    \]
    Finally, let
    \[
    \cR_\sigma(\mathsf{match})
    =
    \max_{P \in \Sigma^m} \cR_\sigma(\mathsf{match}_P),
    \qquad
    \cQ_\sigma(\mathsf{match})
    =
    \max_{P \in \Sigma^m} \cQ_\sigma(\mathsf{match}_P).
    \]

   The three quantities defined above form a natural hierarchy: we first fix both $P$ and $T$, then maximize over the perturbation center $T$ for a fixed pattern $P$, and finally maximize over the pattern $P$. The resulting quantities $\cR_\sigma(\mathsf{match})$ and $\cQ_\sigma(\mathsf{match})$ are the smoothed query complexities of interest in this subsection.

\subsubsection{Randomized query complexity}

\begin{lemma}[Upper bound]
If $\rho m \geq \Omega(\log n)$, then, $\cR_\sigma(\mathsf{match}) = \tilde{O}\left({n}/{\rho m}\right)$.
\end{lemma}

\begin{proof}
    Let $\alpha$ be $\min\left\{\frac{2\ln n}{m\rho},1\right\}$. If $\alpha=1$, then we query the entire text, and the claimed
$\widetilde O(n/(\rho m))$ bound follows immediately since
$\rho m=O(\log n)$. Hence, in the following we assume $\alpha<1$, so that
$\alpha\rho m=2\ln n.$ Independently include every text position $i\in[n]$ in a sample set $S$ with probability $\alpha$, and query all $\widetilde T_i$ for $i\in S$. Call $s \in [n-m+1]$ \emph{compatible} if $\widetilde T_i=P_{i-s}$ for all $i\in S\cap\{s,\ldots,s+m-1\}$. Since we know the pattern, we can determine the set $\mathcal C$ of all compatible $s \in [n-m+1]$. 
    If $\mathcal C$ is empty, output $\textsc{No}$. Otherwise, query all the positions $i \in [n] \setminus S$. This is sufficient to determine if there is an $s \in \mathcal C$ such that $\widetilde T_{s+j}=P_j$ for every $j\in[m]$. Output $\textsc{Yes}$ and this $s$ if there is such an $s$, and $\textsc{No}$ otherwise.
    It remains to bound the query cost of this algorithm. Fix an $s \in [n-m+1]$. Any $j\in[m]$ witnesses that $s$ is not compatible if $s+j\in S$ and $P_j\neq \widetilde T_{s+j}$. Sampling is independent of perturbation so Lemma~\ref{lem:coordinate} gives
    \begin{align*}
        \Prb[\text{$j$ witnesses that $s$ is not compatible}]\ge \alpha\rho.
    \end{align*}
    These events are independent for different $j \in [m]$ so
    \begin{align*}
        \Prb[s\text{ is compatible}] \leq(1-\alpha\rho)^m \leq e^{-\alpha\rho m} = 1/n^2.
    \end{align*}
    Union bounding over all $s \in [n-m+1]$ gives
    \begin{align*}
        \Prb[\mathcal C \neq \emptyset] \leq (n-m) e^{-\alpha\rho m} \leq 1/n.
    \end{align*}
    The expected number of initial sample queries is $\alpha n$, and if $\mathcal C \neq \emptyset$, we make at most $n$ queries. Therefore, 
    \begin{align*}
        R_\sigma(\mathsf{match}) \leq \alpha n + 1/n \cdot n = O(\alpha n) = \tilde{O}\left(\frac{n}{\rho m}\right). 
    \end{align*}
\end{proof}

\begin{lemma}[Lower bound]
If $\rho m \geq  \Omega(\log n)$, then, $\cR_\sigma(\mathsf{match}) = \Omega\left({n}/{\rho m}\right)$.
\end{lemma}

\begin{proof}
Fix the pattern to be \(P=0^m\) and the text to be \(T=0^n\). Suppose that $\widetilde T \sim N_\sigma(T)$ (i.e. the \textsc{No} instance). For each \(i\in[n]\), define
\[
Y_i=\mathbf 1\{\widetilde T_i\neq 0\}.
\]
Then \(Y_1,\ldots,Y_n\) are independent \(\operatorname{Bernoulli}(\rho)\) random variables.
Suppose, without loss of generality, that $m$ divides $n$. The whole argument can be easily extend to the case when $m$ does not divide $n$. Partition the text positions $[n]$ into disjoint consecutive blocks \(B_1,\ldots,B_{n/m}\), each of length \(m\). Let \(\mathcal N\) denote the event that \(\widetilde T\) contains no occurrence of \(0^m\). For every $s \in \{0,\ldots,n-m\}$, the length-\(m\) substring
$\widetilde T_{s+1},\dots,\widetilde T_{s+m}$
equals \(P\) with probability \((1-\rho)^m\). Therefore, by the union bound,
\[
    \Pr[\mathcal N]
    \ge 1 - (n-m+1)(1-\rho)^m
    \ge 1-n e^{-\rho m}
    \ge \frac34.
\]

Let \(\mathcal A\) be an arbitrary randomized algorithm that is correct with probability at least \(2/3\) on every input. Fix $Z$ satisfying the event $\mathcal N$ and fix a block \(B_j\). Let \(Z^{[j]}\) be the text obtained from \(Z\) by replacing every symbol in \(B_j\) by \(0\). Note that \(Z\) is a no instance. 
Consider the executions of \(\mathcal A\) on \(Z\) and \(Z^{[j]}\). Let \(\mathcal H_j\) be the event that the execution on \(Z\) queries a position \(i\in B_j\) satisfying \(Z_i\neq 0\). Until \(\mathcal H_j\) occurs, every answer received by the two executions is identical. Consequently,
\[
\begin{aligned}
    \Pr_Z[\mathcal H_j] \ge \Pr_Z[\mathcal A\text{ outputs \textsc{No}}] -\Pr_{Z^{[j]}}[\mathcal A\text{ outputs \textsc{No}}] \ge \frac23-\frac13=\frac13.
\end{aligned}
\]
Therefore, $\Pr_{\mathcal A}[\mathcal H_j] \ge \frac13\Pr[\mathcal N]\ge \frac14$.

We assume without loss of generality that no position is queried more than once. Let \(q_j\) be the number of positions queried in \(B_j\), and let \(k_j\) be the number of those positions whose symbols are different from \(0\). 
For every \(q\), conditional on the transcript before the \(q\)th query to \(B_j\), the value at that query is an independent \(\operatorname{Bernoulli}(\rho)\) random variable. Therefore
$    \Pr[\text{\(q\)th query to \(B_j\) is not $0$}]
    =\rho\,\Pr[q_j\ge q]$ and hence it follows that
\[
\begin{aligned}
    \mathbb E[k_j]
    = \sum_{q\ge1}
    \Pr[\text{$q$th query to \(B_j\) is not $0$}]
    =\rho\sum_{q\ge1}\Pr[q_j\ge q]
    =\rho\,\mathbb E[q_j].
\end{aligned}
\]
From the definition $\mathcal{H}_j$, we get \(\mathbf 1_{\mathcal H_j}\le k_j\). Hence $\rho\,\mathbb E[q_j] = \mathbb E[k_j] \ge \Pr[\mathcal H_j] \ge \frac14$, and consequently $\mathbb E[q_j]\ge\frac{1}{4\rho}$.
Since the blocks \(B_1,\ldots,B_{n/m}\) are disjoint, the total number \(Q\) of queries satisfies
\[
    \mathbb E[Q]
    \ge\sum_{j=1}^{n/m}\mathbb E[q_j]
    =\Omega\!\left(\frac{n}{\rho m}\right).
\]
Our desired result follows.
\end{proof}

\subsubsection{Quantum query complexity}

\begin{lemma} \label{lem:pattern_matching_subroutine}
Let $P \in \Sigma^m$ and $T \in \Sigma^n$ be any pattern and text, respectively. Suppose that $\rho m \geq c \log n$ for some large enough constant $c$. Also, suppose that $m$ is even and $m/2$ divides $n-m+1$. The positions in $\{0,\ldots,n-m\}$ can be partitioned into $k=\frac{2(n-m+1)}{m}$ consecutive blocks \(B_1,\ldots,B_k\) such that, for every \(j \in [k]\) and $\widetilde T \sim N_\sigma(T)$, there is a quantum algorithm $\mathcal{A}$ with acceptance probability \(q_j(\widetilde T)\) satisfying the following.
\begin{enumerate}
    \item If \(B_j\) contains the starting position of an occurrence of \(P\), then \(q_j(\widetilde T)=1\),
    \item $\mathbb E_{\widetilde T\sim N_\sigma(T)}[q_j(\widetilde T)]
        \le n^{-3}$, and
    \item $\mathcal A$ makes $O\!\left(\frac{\log n}{\sqrt\rho}\right)$ queries.
\end{enumerate}
\end{lemma}

\begin{proof}
Let $\widetilde T \sim N_\sigma(T)$. Let $a_j$ be the starting position of block $B_j$. Then, any length-$m$ substring of $\widetilde T$ beginning at some position in $B_j$ will contain the positions in
$C_j=\{a_j+m/2,\ldots,a_j+m-1\}.$

For \(i\in B_j\), let \(D_{i,j}\) be the substring of \(P\) aligned with \(C_j\) when \(P\) starts at position \(i\) in $\widetilde T$, and define
$\mathcal D_j=\{D_{i,j}:i\in B_j\}.$
Thus, \(|\mathcal D_j|\le m/2\), and an occurrence of $P$ beginning in \(B_j\) implies $\widetilde T_{C_j}\in\mathcal D_j.$

Now, fix a $j \in [k]$ and some large enough constant $c'$. Note that
$|B_j| = m/2 \geq \lceil c'\log n\rceil \cdot \lceil 1/\rho\rceil = \ell.$
Partition the first \(\ell\) positions of \(C_j\) into \(\lceil\log n\rceil\) disjoint chunks of length \(\lceil 1/\rho\rceil\).
The algorithm $\mathcal{A}$ will maintain a set $\mathcal D_j'\subseteq\mathcal D_j$, initially \(\mathcal D_j'=\mathcal D_j\). Let $Z$ be the majority string for the set of strings in $\mathcal D_j'$. That is, for each position \(i\), let \(Z_i\) be the symbol having maximum frequency among
$\{S_i:S\in\mathcal D_j'\}.$

Use Grover's search algorithm~\cite{Gro96} to find a position \(i\) for which $\widetilde T_i\neq Z_i$
using \(O(\sqrt{1/\rho})\) queries. If such an $i$ is found, update $\mathcal D_j'\leftarrow
\{S\in\mathcal D_j':S_i=\widetilde T_i\}.$
Reject if \(\mathcal D_j'\) becomes empty, and accept if it remains nonempty after all \(\lceil c'\log n\rceil\) chunks.
Every successful update reduces \(|\mathcal D_j'|\) by at least a factor of two. Indeed, if the frequency of \(Z_i\) is greater than \(|\mathcal D_j'|/2\), then every other symbol (including $\widetilde T_i$) has frequency less than \(|\mathcal D_j'|/2\). Therefore
$r =\left\lceil\log_2|\mathcal D_j|\right\rceil+1
\le \left\lceil\log_2 n\right\rceil+1$
successful updates force \(\mathcal D_j'\) to become empty.
If \(\widetilde T_{C_j}=S^\star\) for some \(S^\star\in\mathcal D_j\), every update retains \(S^\star\). Hence, $\mathcal{A}$ accepts with certainty whenever \(B_j\) contains the starting position of an occurrence.

Next, we bound the acceptance probability of $\mathcal{A}$ under smoothing.
Condition on the history before a chunk while
\(\mathcal D_j'\neq\textsf{Var}nothing\). For every position \(i\) and fixed symbol \(Z_i\), we have $\Pr[\widetilde T_i\neq Z_i]\ge\rho$ by Lemma~\ref{lem:coordinate}.
The perturbations in each chunk are independent, and therefore the probability that a fresh chunk contains no disagreement with $Z$ is $\leq (1-\rho)^{\lceil 1/\rho\rceil} \le e^{-1}.$
Consequently, the probability $\alpha$ of a successful update, which is distributed according to \(\operatorname{Bernoulli}(\alpha)\), is at least $\frac23(1-e^{-1})>0.$
Hence, for \(c\) sufficiently large, the Chernoff bound (Lemma~\ref{lem:chernoff2}) gives
\[
    \Pr[\mathcal A \text{ accepts}]
    \le
    \Pr[\operatorname{Bin}(\lceil c' \log n\rceil,\alpha)<r]
    \le n^{-3}.
\]
Finally, the cost of the algorithm $\mathcal A$ is $O\!\left(\frac{\log n}{\sqrt\rho}\right)$ queries.
\end{proof}

\begin{lemma}[Upper bound]
Suppose $\rho m \geq \Omega(\log n)$. Then, $\cQ_\sigma(\mathsf{match}) = \tilde{O}\left(\sqrt{{n}/{\rho m}}\right)$.
\end{lemma}

\begin{proof}
We will suppose without loss of generality that $m$ is even and $m/2$ divides $n-m+1$. Let \(B_1,\ldots,B_k\) and \(q_1(\widetilde T),\ldots,q_k(\widetilde T)\) be given by \Cref{lem:pattern_matching_subroutine}. Prepare a uniform superposition over \(j\in[k]\) and run the algorithm $\mathcal A$ in \Cref{lem:pattern_matching_subroutine}. For a perturbed text \(\widetilde T\), its acceptance probability~is
\[
    p(\widetilde T)=\frac1k\sum_{j=1}^{k}q_j(\widetilde T).
\]
If \(\widetilde T\) contains an occurrence whose starting position is in block $B_j$, then $\mathcal{A}$ in \Cref{lem:pattern_matching_subroutine} accepts with certainty, and hence $ p(\widetilde T)\ge\frac1k$. Apply amplitude amplification (Theorem~\ref{thm:Fixed_AA}) using \(O(\sqrt k)\) calls to $\mathcal{A}$ and its inverse. When $p(\widetilde T)\ge\frac1k$, this amplification accepts with probability at least $2/3$ and for every perturbed text $\widetilde T$, it only accepts with probability $O(k \cdot p(\widetilde T))$.  
If it accepts, query the entire text $\widetilde T$ and decide pattern matching
exactly.  This is also sufficient to ouput the starting position of a match. Otherwise, output \textsc{No}. If $P$ occurs in $\widetilde T$, we will succeed with probability at least \(2/3\); otherwise, we will return \textsc{No} certainly. We can now calculate the query complexity of our algorithm. Using \Cref{lem:pattern_matching_subroutine}, the probability of our procedure accepting and thus needing to query the entire $\widetilde T$ is
\begin{align*}
    \Pr[\text{accept}] = O\left(k \cdot\,\mathbb E_{\widetilde T \sim N_\sigma(T)} [p(\widetilde T)]\right) = O\left(\sum_{j=1}^{k} E_{\widetilde T \sim N_\sigma(T)}[q_j(\widetilde T)]\right) = O\left(k n^{-3}\right)  = O\left(n^{-2}\right)
\end{align*}
It follows that 
\begin{align*}
    \cQ_\sigma(\mathsf{match}) = O\left(\frac{\sqrt k \log n}{\sqrt \rho} + n \cdot n^{-2}\right) = \tilde{O}\left(\sqrt{\frac{n}{\rho m}}\right).
\end{align*}
\end{proof}

To prove our quantum lower bound, we will use the following adversary lower bound.

\begin{theorem}[Distributional adversary bound~\cite{BBHKLS17}]\label{thm:distributionAdvMatrix}
Let $\mathcal A$ be a quantum algorithm making $T$ queries to an input $x=(x_1,\ldots,x_n)$, and let $\mathcal P,\mathcal Q$ be two probability distributions on the inputs. Suppose that the acceptance probabilities of $\mathcal A$ on inputs drawn from $\mathcal P$ and $\mathcal Q$ are $s_{\mathcal P}$ and $s_{\mathcal Q}$, respectively. Let $\delta_{\mathcal P}[x]=\sqrt{\mathcal P(x)}$ and $\delta_{\mathcal Q}[y]=\sqrt{\mathcal Q(y)}$, define $\tau(s_{\mathcal P},s_{\mathcal Q})=\sqrt{s_{\mathcal P}s_{\mathcal Q}}+\sqrt{(1-s_{\mathcal P})(1-s_{\mathcal Q})}$, and let $\Delta_j[x,y]=\one\{x_j\neq y_j\}$. Then, for any adversary matrix $\Gamma$ whose rows and columns are indexed by the supports of $\mathcal P$ and $\mathcal Q$, respectively,
\[
T=\Omega\left(\min_{j\in[n]}\frac{\delta_{\mathcal P}^{T}\Gamma\delta_{\mathcal Q}-\tau(s_{\mathcal P},s_{\mathcal Q})\|\Gamma\|}{\|\Gamma\circ\Delta_j\|}\right).
\]
\end{theorem}

\begin{lemma}[Lower bound]
Suppose $\rho m \geq \Omega(\log n)$. Then, $\cQ_\sigma(\mathsf{match}) = \Omega\left(\sqrt{{n}/{\rho m}}\right)$.
\end{lemma}

\begin{proof}
Fix the pattern to be \(P=0^m\) and the text to be \(T=0^n\). Suppose, without loss of generality, that $m$ divides $n$. The whole argument can be easily extend to the case when $m$ does not divide $n$. Partition the text positions $[n]$ into disjoint consecutive blocks \(B_1,\ldots,B_{n/m}\), each of length \(m\). Divide each block $B_j$ into two consecutive cells $C_{j,1}$ and $C_{j,2}$ of lengths \(\lfloor m/2\rfloor\) and \(\lceil m/2\rceil\) respectively. Suppose that $\widetilde T \sim N_\sigma(T)$. Let \(\mathcal N\) denote the event that each of these $2n/m$ cells contains a non-zero symbol. A fixed cell is equal to \(0^m\) with probability at most $(1-\rho)^{\lfloor m/2\rfloor} \le e^{-\rho\lfloor m/2\rfloor}$.
Hence, by the union bound,
\[
    \Pr[\mathcal N]
    \ge
    1 - 2(n/m) e^{-\rho\lfloor m/2\rfloor}
    \ge 1-\frac14 = \frac34.
\]
Notice that any $y$ satisfying \(\mathcal N\) is a no instance. Indeed, any run of $0$s must have length at most $m-2$ since each cell will have at least 1 non-zero symbol. 

Fix a $j \in [n/m]$. Let \(\nu\) be the distribution on $B_j$ conditioned on each of $C_{j,1}$ and $C_{j,2}$ containing a non-zero symbol, and let \(\mathcal Z\) be its support. Conditioned on \(\mathcal N\), the \(n/m\) blocks are independent with each being drawn from the distribution \(\nu^{\otimes n/m}\). 
Let $\mathcal P=\nu^{\otimes n/m}$ be the distribution of $\widetilde T$ conditioned on $\mathcal N$. Let $\mathcal Q$ be the distribution obtained by choosing $j\in[n/m]$ uniformly, replacing block $j$ by $0^m$, and drawing all remaining blocks independently from $\nu$.
Let $C$ be any cell. If an index \(i \in [n]\) belongs to \(C\), then,
\[
    \Pr_{z\sim\nu}[z_i\neq 0]=\frac{\rho}{1-(1-\rho)^{|C|}}\le 2\rho
\]
where the inequality follows since $\rho m \geq c \log n$.

We apply Theorem~\ref{thm:distributionAdvMatrix} with the following adversary matrix:
We design an adversary matrix whose rows are indexed by $z=(z_1,\ldots,z_{n/m})\in\mathcal Z^{n/m}$, which are no instances, and whose columns are indexed by \((j,w)\) with \(j\in[n/m]\) and \(w\in\mathcal Z^{n/m-1}\), which corresponds to the yes instance whose block \(j\) is \(0^m\) and the remaining blocks are given by \(w\). Define
\[
    \Gamma_{z,(j,w)}
    =
    \sqrt{\nu(z_j)}\,\mathbf 1\{z_{-j}=w\}.
\]
Let \(u_z=\sqrt{\nu(z)}\) for \(z\in\mathcal Z\). Then \(\|u\|=1\), and
\[
    \Gamma\Gamma^{\mathsf T}
    =
    \sum_{j=1}^{n/m}
    I^{\otimes(j-1)}
    \otimes |u\rangle\langle u|
    \otimes I^{\otimes(n/m-j)}.
\]
The terms in this summation are commuting projections, so this sum has norm \({n/m}\). Thus, $\|\Gamma\|=\sqrt{n/m}$.
Let $\delta_{\mathcal P}[z]=\sqrt{\mathcal P(z)}$ and
$\delta_{\mathcal Q}[(j,w)]=\sqrt{\mathcal Q(j,w)}$. By the definitions of
$\mathcal P,\mathcal Q$ and $\Gamma$, we have $\delta_{\mathcal P}^{T}\Gamma\delta_{\mathcal Q} =\sqrt{n/m} =\|\Gamma\|.$

Fix a position $i \in [n]$ and let $B_j$ be the block which contains $i$. Only the columns indexed by $(j, w)$ can differ from their corresponding rows at position \(i\). In these columns, entrywise multiplication by \(\Delta_i\) replaces \(u\) by the vector $(u_i)_z =\sqrt{\nu(z)}\,\mathbf 1\{z_i\neq 0\}$, whose squared norm is at most $\Pr_{z\sim\nu}[z_i\neq 0]$. Therefore,$\|\Gamma\circ\Delta_i\|\le\sqrt{\Pr_{z\sim\nu}[z_i\neq 0]} \le \sqrt{2 \rho}$ for every \(i\). 
Let $L$ denote the expected number of queries of an arbitrary bounded-error
algorithm under $\widetilde T\sim N_\sigma(T)$. Since
$\Pr[\mathcal N]\geq 3/4$, its expected number of queries under
$\mathcal P$ is at most $4L/3$. Truncate the algorithm after $12L$ queries,
and output \textsc{Yes} upon truncation. By Markov's inequality, under
$\mathcal P$ the truncation probability is at most $1/9$. Hence, viewing
\textsc{No} as acceptance, its acceptance probabilities satisfy
$s_{\mathcal P}\geq \frac23-\frac19=\frac59$ and $s_{\mathcal Q}\leq\frac13.$
Therefore $\tau(s_{\mathcal P},s_{\mathcal Q})\leq
\tau(5/9,1/3)<1$. Applying Theorem~\ref{thm:distributionAdvMatrix},
\[
12L
=\Omega\left(
\frac{(1-\tau(5/9,1/3))\|\Gamma\|}
{\max_i\|\Gamma\circ\Delta_i\|}
\right)
=\Omega\left(\sqrt{\frac{n}{\rho m}}\right).
\]
Thus,
$\cQ_\sigma(\mathsf{match})
=\Omega\left(\sqrt{\frac{n}{\rho m}}\right).$
\end{proof}

\section{Edit distance in the smoothed model}\label{sec:smoothedEditD}

For two strings $X,Y\in\Sigma^*$ over an arbitrary alphabet $\Sigma$, their edit distance $\textsf{ed}(X,Y)$ is the minimum number of elementary edits needed to transform $X$ into $Y$, where an edit consists of inserting a symbol, deleting a symbol, or replacing one symbol by another. Computing edit distance is a fundamental primitive in applications ranging from computational biology to text processing. Exact computation, however, appears inherently expensive: the standard algorithm runs in $O(n^2)$ time, and a truly subquadratic algorithm would refute the Strong Exponential Time Hypothesis~\cite{IP01,IPZ01}. Since the exact value is often unnecessary in applications, a large body of work has instead studied approximation.

\begin{definition}[Worst-case approximating edit distance]
Let $X,Y\in\Sigma^*$ and $\alpha\geq 1$. Given (quantum) query access to $X$ and $Y$, the goal is to output $\widetilde e$ such that
\(        \frac{1}{\alpha}\,\textsf{ed}(X,Y)
        \leq \widetilde e
        \leq \alpha\,\textsf{ed}(X,Y).
   \)
\end{definition}

There has been substantial progress on approximating edit distance in the worst case. Classically, a constant-factor approximation can be obtained in $n^{1+\varepsilon}$ time for every constant $\varepsilon>0$~\cite{AN20EditApproxLinear}. At the other extreme, obtaining an approximation arbitrarily close to $1$ appears considerably harder: the recent work of Mao and Rubinstein~\cite{MaoRubinstein26} gives a randomized $(1+\varepsilon)$-approximation in time~$
    n^2/2^{\log^{\Omega(1)}n}.
$

For very long strings, even near-linear running times can be prohibitive, motivating the search for faster algorithms under additional assumptions on the input. In this direction, Andoni and Krauthgamer~\cite{andoni2012smoothed} studied edit distance beyond the worst case through smoothed analysis. Analogously to the smoothed pattern-matching model considered in Section~\ref{sec:smoothedpattern}, their model begins with arbitrary strings $X^*,Y^*$ together with a fixed longest common subsequence. Each character is then independently perturbed with probability $\sigma$ by replacing it with a uniformly random character from the alphabet, except that perturbations of characters matched by the fixed longest common subsequence are coupled across the two strings. We call the resulting pair $(X,Y)$ a \emph{$\sigma$-smoothed~instance}.

\begin{definition}[Smoothed instance approximating edit distance]
Let $\sigma>0$ and $\alpha\geq 1$. Fix arbitrary strings $X^*,Y^*\in\Sigma^*$ together with a longest common subsequence between them, and let $(X,Y)$ be the resulting $\sigma$-smoothed instance. Given (quantum) query access to $X$ and $Y$, the goal is to output $\widetilde e$ such that
\(        \frac{1}{\alpha}\,\textsf{ed}(X,Y)
        \leq \widetilde e
        \leq \alpha\,\textsf{ed}(X,Y).
   \)
\end{definition}

The key structural insight of Andoni and Krauthgamer~\cite{andoni2012smoothed} is that smoothed edit distance can be reduced to \emph{Ulam distance}, namely edit distance between non-repetitive strings. Roughly speaking, after perturbation, sufficiently long substrings are unlikely to repeat, and can therefore be treated as distinct symbols. This allows a smoothed instance of edit distance to be represented, up to a controlled approximation loss, by an instance of Ulam distance. 
For the classical algorithm in~\cite{andoni2012smoothed}, queries to the resulting non-repetitive strings are answered using $O((\log n)/\sigma)$ queries to the original strings.  The answer returned at a newly queried position is determined using tables containing the positions queried earlier.  A separate reversible implementation is therefore needed for a quantum algorithm; we give such an implementation for all the subset databases used in our algorithm in Appendix~\ref{app:ak-quantum-access}.

Motivated by this connection, Andoni and Nguyen~\cite{AN10UlamTester} developed sublinear algorithms for approximating Ulam distance. If $R=\textsf{ed}(A,B)$ for two non-repetitive strings of length $n$, their algorithm obtains a constant-factor approximation in
time $$
    \widetilde O\!\left(\frac{n}{R}+\sqrt n\right)
$$
 Their approach has two main ingredients. First, the strings are decomposed into shorter pairs of substrings whose edit distances add up to the original distance. Second, rather than estimating the distance of every pair exactly, they use gap testers at multiple distance scales and aggregate the resulting contributions.

Since our focus is on query complexity, we note that the $\widetilde O(n/R+\sqrt n)$-time algorithm of Andoni and Nguyen also uses at most $\widetilde O(n/R+\sqrt n)$ queries to the input strings. Our goal is to quantumize this framework and further reduce the number of queries.
This is not obtained by simply replacing individual classical routines by Grover search: the reduction involves many substring pairs at different distance scales, with highly nonuniform evaluation costs. We develop new quantum collision-estimation primitives, use them to obtain a faster quantum gap tester for Ulam distance, and combine the resulting tests using a variable-time quantum product estimator. Together with a quantum implementation of the partial-alignment step, these ingredients give the following quantum analogue of the Andoni--Nguyen Ulam estimator.

\begin{restatable}{theorem}{thmrestateUlam}\label{thm:ulam}
Let $A,B\in\Sigma^n$ be non-repetitive strings, and let
$
    R=\textsf{ed}(A,B).
$
Given quantum query access to $A$ and $B$, there is a bounded-error quantum algorithm that outputs a constant-factor approximation to $R$ in query complexity
$$
    \widetilde O\!\left(
        \frac{n}{R}
        +
        \frac{n^{2/3}}{R^{1/3}}
    \right)
$$
\end{restatable}

In particular, when the edit distance between the two strings is sufficiently large, they can even obtain an algorithm that runs in sublinear time in the string length $n$ under the smoothed model. By considering the smoothed model for edit distance, Andoni and Krauthgamer explained the empirical success of practical heuristic techniques based on finding approximately matching local substrings. They observed that, after perturbation, the strings behave more like non-repetitive strings when sufficiently long substrings are treated as new characters.  This phenomenon suggests that one can use an algorithm for computing Ulam distance, namely, the edit distance between two non-repetitive strings in the worst case, to solve edit distance problems on general strings in the smoothed model.

The two terms have different origins. The first comes from constructing the partial alignment of the two strings, while the second is the cost of estimating the total distance of the resulting substring pairs using our quantum gap and product testers. In particular, when $R=\Omega(n)$, the query complexity becomes $\widetilde O(n^{1/3})$, improving on the $\widetilde O(\sqrt n)$ scale of the classical Ulam estimator.

Finally, we combine this theorem with the reduction of Andoni and Krauthgamer from smoothed edit distance to Ulam distance. Their construction loses an $\widetilde{O}(1/\sigma)$ factor in the approximation and allows each query to the resulting Ulam instance to be simulated using $\widetilde O(1/\sigma)$ queries to the original strings. This immediately yields our smoothed edit distance algorithm.

\begin{restatable}{corollary}{SmoothedEditDistanceCor}
\label{cor:smoothed-edit}
Let $\sigma>0$, and let $X,Y\in\Sigma^n$ be a $\sigma$-smoothed instance with
$
    R=\textsf{ed}(X,Y).
$
With overwhelming probability over the smoothing, there is a bounded-error quantum algorithm that outputs an $O(1/\sigma)$-approximation to $R$ using
$$
    \widetilde O\!\left(
        \frac{1}{\sigma}
        \left(
            \frac{n}{R}
            +
            \frac{n^{2/3}}{R^{1/3}}
        \right)
    \right)
$$
queries to $X,Y$.
\end{restatable}

The rest of the section develops the three quantum ingredients underlying the theorem: collision estimation, a quantum \textsc{GapUlamTest}, and a variable-time \textsc{UlamProductTest}. We then show how these primitives combine with the (quantum version) partial-alignment procedure to obtain the stated Ulam estimator and, through the reduction above, the smoothed edit distance result.

\subsection{New quantum collision subroutines}\label{sec:Qcollision}
In this subsection we will introduce some collision subroutines that we will use for Ulam distance estimation (and hence smoothed edit distance). The first subroutine we would like to have is to estimate the approximate collisions between two substrings, assuming we have a lower bound on the number of collisions.  To prove the theorem, we would first like to prove the following gap tester for collisions between non-repetitive strings. We first prove a gap-sensitive collision tester, which distinguishes between at most $\ell_1$ and at least $\ell_2$ collisions. This is the main counting primitive of the subsection. We then bootstrap it to obtain a multiplicative estimator for the total number of collisions, and subsequently an additive estimator with a prescribed additive error. These estimators will be used repeatedly when we only need approximate overlap information between substrings. Finally, when the locations of the collisions themselves are required, we adapt the multiple-collision algorithm of~\cite{Bonnetain2025multiplecollision} to enumerate all $m$ collisions in $\widetilde O(n^{2/3}(m+1)^{1/3})$ queries. A self-contained proof of this adaptation is given in Appendix~\ref{app:fixed-matching-collisions}. Thus, the first three results form a progression from gap testing to increasingly flexible forms of collision estimation, while the last provides the complementary enumeration primitive. Together they supply the collision subroutines used in the algorithms for smoothed edit distance.
For convenience, we also include two additional quantum primitives used later: a known quantum pair-finding procedure and our variable-time approximate-counting procedure for combining quantum tests with nonuniform query costs.

\begin{theorem} \label{thm:gap_dist_collisions}
    Let $\tau>0$, $\ell_1, \ell_2 \in [n]$ such that $\ell_1 < \ell_2$. Let $A,B \in \Sigma^n$ be non-repetitive strings and let $m$ denote the number of collisions between them. Suppose that either $m \leq \ell_1$ or $m \geq \ell_2$ and suppose we can make quantum queries to entries of $A$ and $B$. Then, there is an algorithm that decides if $m \leq \ell_1$ or $m \geq \ell_2$ with success probability $\geq 1-\tau$, using $\widetilde{O}\left(\left(\frac{n \sqrt{\ell_2}}{\ell_2-\ell_1}\right)^{2/3}\right)$ queries.
\end{theorem}

\begin{proof}
    Let $r \leq n/2$ be a parameter to be specified later. Let $Z$ be the concatenation of strings $A$ and $B$ of length $2n$. Since $A$ and $B$ are non-repetitive, any collision in $Z$ will have a part in $A$ and a part in $B$. For a size-$r$ subset $S \subseteq [2n]$, define $N_m(S)$ to be the number of collision pairs fully contained in $S$ given that there are $m$ collisions in the input, and $N_m = \mathbb{E}_S[N_m(S)]$. Let $\rho = \frac{r(r-1)}{2n(2n-1)}$ be the probability that a fixed collision pair is fully contained in a uniformly random size-$r$ subset of $[2n]$. Let $\mu_2 = \ell_2 \rho$ be the lower bound on the average number of collisions in a uniformly random size-$r$ subset of $[2n]$ if $m \geq \ell_2$. Let $t = \ceil{\mu_2}$.

    Consider the Johnson graph $\mathcal{J}([2n], r)$ on size-$r$ subsets of $[2n]$, which by the discussion in Section~\ref{sec:q_subroutines} its spectral gap is
$\delta=\Theta(1/r)$. 
   Let $p(m) = \Pr_S [N_m(S) \geq t]$.    Also, define $p_1 = p(\ell_1) = \Pr_S [N_{\ell_1}(S) \geq t]$ and $p_2 = p(\ell_2) = \Pr_S [N_{\ell_2}(S) \geq t]$. Note that $p_1$ and $p_2$ can be calculated without any queries since $\ell_1$, $\ell_2$ and $t$ are known, and also note that the function $p(m)$ is monotonically increasing in $m$ so $p(m) \leq p_1$ if $m \leq \ell_1$ and $p(m) \geq p_2$ if $m \geq \ell_2$. 

    We will do our quantum walk on a modification of $\mathcal{J}([2n], r)$ to ensure that the probability of a vertex being marked is $\Omega(p_2)$. In particular, our graph will have vertices from $\binom{[2n]}{r} \times \{0,1\} \times [\alpha]$ where $\alpha = \floor{1/p_2}$ and a vertex $(S, b, j)$ is connected to a vertex $(S', b', j')$ iff $S$ is connected to $S'$ in $\mathcal{J}([2n], r)$ (i.e. $|S \cap S'| = r-1$). In the database of our walk, along with maintaining a vertex label $(S, b, j)$, we will also maintain the input queries associated with the indices in $S$. We say that a vertex $(S, 1, j)$ is marked if $j=1$, and a vertex $(S, 0, j)$ is marked if $N_m(S) \geq t$. Note that a random vertex $(S, b, j)$ in this graph is marked with probability $q(m) = \frac{1}{2\alpha} + \frac{p(m)}{2}$. We will estimate the marked fraction $q(m)$ in our graph using Theorem \ref{thm:approx_marked_fraction}. Note that for any $m$, $q(m) \geq \frac{1}{2\alpha}$ so $\lambda = \Omega(p_2)$. Since the spectral gap of $\mathcal{J}([2n], r)$ is $\Theta(1/r)$, the spectral gap $\delta$ of our graph will also be $\Theta(1/r)$. The setup cost $S$ of preparing a uniform superposition of vertices in our graph (and the querying the associated input indices) is $O(R)$ while the update cost $U$ of moving between neighboring vertices is $O(1)$. For any vertex $(S, b, j)$, it is sufficient to check markedness using the information contained in the walk database so $C = 0$. It remains to compute $\eta$. Let $q_1 = q(\ell_1)$ and $q_2 = q(\ell_2)$. The gap between the probability of a random vertex being marked in the case when $m \leq \ell_1$ and when $m \geq \ell_2$ is at least $q_2 - q_1 = \frac{p_2-p_1}{2}$. We must choose $\eta$ to be more than $\frac{q_2-q_1}{2 \max\{q_1,q_2\}}$ so that $(1+\eta)q_1 < \frac{q_1+q_2}{2}$ and $(1+\eta)q_2 > \frac{q_1+q_2}{2}$. Thus, it suffices to choose $\eta$ to be $\frac{q_2-q_1}{4 \max\{q_1,q_2\}} = \frac{p_2-p_1}{8p_2}$.     We will proceed by lower bounding the probability gap $p_2-p_1$ in terms of $\ell_2-\ell_1$. Notice that
    \begin{align*}
        p_2-p_1 =  p(\ell_2)-p(\ell_1) = \sum_{\lambda = \ell_1}^{\ell_2-1} p(\lambda+1) - p(\lambda)
    \end{align*}

    For any $\lambda \in [\ell_1,\ell_2-1]$, consider an instance with $\lambda$ collisions and an instance with $\lambda+1$ collisions, where the latter consists of the same $\lambda$ collisions together with one additional collision at indices $i,j\in[2n]$. In particular, $i$ and $j$ do not belong to any of the first $\lambda$ collision pairs. Then, for every size-$r$ subset $S\subseteq[2n]$,
    \[
        N_{\lambda+1}(S)
        =
        N_\lambda(S)+\mathbf{1}[i,j\in S].
    \]
    Therefore,
    \[
        p(\lambda+1)-p(\lambda)
        =
        \Pr_S[N_\lambda(S)=t-1,\ i,j\in S]
        =
        \rho\Pr_S[N_\lambda(S)=t-1\mid i,j\in S].
    \]

\begin{claim} \label{lem:concentration_collision}
Let $n/2\geq s\geq 2$, let $\ell\in[n]$ be a positive integer, and let $A,B\in\Sigma^n$ be non-repetitive strings. Let
\[
    \rho=\frac{s(s-1)}{2n(2n-1)},
    \qquad
    \mu=\ell\rho,
    \qquad
    t=\ceil{\mu}\geq1.
\]
Let $\lambda$ denote the number of collisions between $A$ and $B$, let $\mu_\lambda=\lambda\rho$, and, for every size-$s$ subset $S\subseteq[2n]$, let $N_\lambda(S)$ denote the number of collision pairs fully contained in $S$. Let $I\subseteq[2n]$ be either empty or consist of two indices, one corresponding to an entry of $A$ and one corresponding to an entry of $B$, neither of which belongs to any collision pair. If
\[
    |\mu_\lambda-\mu|
    \leq
    \frac{\sqrt{t}}{10},
\]
then
\[
    \Pr_S[N_\lambda(S)=t-1\mid I\subseteq S]
    =
    \Omega\left(\frac{1}{\sqrt{t}}\right).
\]
\end{claim}

The proof of this claim is a standard anti-concentration estimate for the number
of sampled collision pairs and defer the proof to Appendix~\ref{app:claim5.6}.    It follows from Claim~\ref{lem:concentration_collision}, applied with $I=\{i,j\}$, that for every $\lambda\in[\ell_1,\ell_2-1]$ satisfying
    \[
        |\lambda-\ell_2|
        \leq
        \frac{\sqrt{t}}{10\rho},
    \]
    we have
    \[
        \Pr_S[N_\lambda(S)=t-1\mid i,j\in S]
        =
        \Omega\left(\frac{1}{\sqrt{t}}\right).
    \]
    Therefore,
    \[
        p(\lambda+1)-p(\lambda)
        =
        \Omega\left(\frac{\rho}{\sqrt{t}}\right).
    \]
    Since $\rho<1/16$, we have $\sqrt{t}/(10\rho)>1$. Consequently, the number of such $\lambda\in[\ell_1,\ell_2-1]$ is
    \[
        \min\left\{
            \ell_2-\ell_1,
            \floor{\frac{\sqrt{t}}{10\rho}}
        \right\}
        =
        \Omega\left(
            \min\left\{
                \ell_2-\ell_1,
                \frac{\sqrt{t}}{\rho}
            \right\}
        \right).
    \]
    Therefore,
    \begin{align*}
        \sum_{\lambda = \ell_1}^{\ell_2-1} p(\lambda+1) - p(\lambda) = \Omega(\min\{\ell_2-\ell_1, \sqrt{t}/\rho\} \cdot \rho/\sqrt{t}) = \Omega\left(\min\left\{1, \frac{(\ell_2-\ell_1) \rho}{\sqrt{t}}\right\}\right).
    \end{align*}

    Notice that $p_2 \leq \min\{1, \mu_2\}$. Indeed, if $\mu_2 \geq 1$, this is immediate; if $\mu_2 \leq 1$, then using Markov's inequality we get 
    \begin{align*}
        p_2 = \Pr[N_{\ell_2} \geq 1] \leq \mathbb{E}[N_{\ell_2}] = \mu_2. 
    \end{align*}

    We are ready to bound $\frac{1}{\eta \sqrt{\lambda}}$. We have
    \begin{align*}
        \frac{1}{\eta \sqrt{\lambda}} = \frac{\sqrt{p_2}}{p_2-p_1} = O\left(\frac{\sqrt{\min\{1, \mu_2\}}}{\min\left\{1, \frac{(\ell_2-\ell_1) \rho}{\sqrt{t}}\right\}}\right) 
        = O\left(\max\left\{1, \frac{\sqrt{\mu_2}}{(\ell_2-\ell_1) \rho}\right\}\right) = O\left(1 + \frac{\sqrt{\ell_2}}{(\ell_2-\ell_1) \sqrt{\rho}}\right)
    \end{align*}

    Applying Theorem \ref{thm:approx_marked_fraction}, we get that the complexity of our quantum walk is 
    \begin{align*}
        \widetilde{O}\left(S + \frac{1}{\eta} \frac{1}{\sqrt{\lambda}}\left(\frac{1}{\sqrt{\delta}} U + C\right) \right) = \widetilde{O}\left(r + \left(1 + \frac{\sqrt{\ell_2}}{(\ell_2-\ell_1) \sqrt{\rho}}\right) \cdot \sqrt{r}\right) =  \widetilde{O}\left(r+\frac{\sqrt{\ell_2}}{\ell_2-\ell_1}\cdot {\frac{n}{\sqrt{r}}}\right),
    \end{align*}
    where the last equation follows because $\rho=\frac{r(r-1)}{2n(2n-1)}$. By choosing $r=\Theta\left(\left(\frac{n\sqrt{\ell_2}}{\ell_2-\ell_1}\right)^{2/3}\right)$ to balance the above two terms, we therefore get the desired $\widetilde{O}\left(\left(\frac{n \sqrt{\ell_2}}{\ell_2-\ell_1}\right)^{2/3}\right)$ bound.
\end{proof}

Once we have the gap tester, one can easily convert it into an $\eps$-approximation collision estimator by binary search, given a known lower bound on the number of collisions.

\begin{theorem}[Approximating number of collisions]
\label{cor:legall}
Let $A,B\in\Sigma^n$ be non-repetitive strings, and let $m$ be the number of
collisions between $A,B$. Suppose that a known lower bound
$m_0\in[n]$ satisfying $m_0\leq m$ is given. For every
$\varepsilon,\tau\in(0,1)$, there is a quantum algorithm, given quantum
query access to $A,B$, that outputs, with success probability at least
$1-\tau$, 
\begin{enumerate}
    \item an estimate $\widetilde m$ such that
$
|\widetilde m-m|<\varepsilon m
$
using
$
\widetilde O\left(
\left(\frac{n}{\varepsilon\sqrt m}\right)^{2/3}
+
\left(\frac{n}{\sqrt{m_0}}\right)^{2/3}
\right)
$ queries.
\item an estimate $\widetilde{m}$ such that 
$
|\widetilde{m}-m | < t
$
using
$
\widetilde{O}\left(\left( \frac{n\sqrt{m+t}}{t} \right)^{2/3}\right)
$
queries.
\end{enumerate}
\end{theorem}

\begin{proof}
     We first prove item 1 and then use it to prove item 2.

        \begin{algorithm}[!ht]
        \caption{$\varepsilon$-approximate collision counting between non-repetitive strings $x, y \in \Sigma^n$}
        \label{alg:eps_approx_collisions}
        \begin{algorithmic}[1] 
            \State Let $\gamma = 1 + \frac{\varepsilon}{8}$ and $i = 0$.
            \State For $h=0,1,\ldots$, let $L_h=(3/2)^h m_0$ and run Theorem~\ref{thm:gap_dist_collisions} with $(\ell_1,\ell_2)=(L_h,L_{h+1})$ until it first outputs $m\leq L_h$. Let $\widetilde m:=L_h$.                         \State Let $a_0 = 2\widetilde{m}/3$ and $b_0 = 2\widetilde{m}$.
            \While{$b_i/a_i > \gamma^3$}
                \State Let $\mu_i = \sqrt{a_ib_i}$.
                \State Let $\alpha_i =\mu_i/\gamma$ and $\beta_i = \mu_i \gamma$.
                \State Run Theorem \ref{thm:gap_dist_collisions} to decide if $m \leq \alpha_i$ or $m \geq \beta_i$.
                \If{the output is $m \leq \alpha_i$}
                    \State Let $a_{i+1} = a_i$ and $b_{i+1} = \beta_i$.
                \Else
                    \State Let $a_{i+1} = \alpha_i$ and $b_{i+1} = b_i$.
                \EndIf
            \State Let $i = i+1$.
            \EndWhile
            \State Output $\widehat{m} = \sqrt{a_i b_i}$.
        \end{algorithmic}
    \end{algorithm}

   \textbf{Item 1}: We use the algorithm in Theorem \ref{thm:gap_dist_collisions} to give the desired algorithm. Let us argue correctness for Algorithm \ref{alg:eps_approx_collisions}. 
For Line~2 of Algorithm~1, condition on every invocation of Theorem~\ref{thm:gap_dist_collisions} being correct whenever its input satisfies one of the two promised cases. Let $h$ be the first iteration in which the tester outputs $m\leq L_h$, and let $m':=L_h$. This output implies $m<L_{h+1}=3m'/2$. If $h>0$, the output in the preceding iteration was $m\geq L_h$, which implies $m>L_{h-1}=2m'/3$. If $h=0$, we instead use $m\geq m_0=m'$. Therefore, $\frac{2m}{3}<m'<\frac{3m}{2}.$
Notice that this argument remains valid when $m$ lies in the promise gap of an invocation, since the output may then be arbitrary.
   After that, we have an estimator $\widetilde m$ s.t. $2m/3 \leq \widetilde m \leq 3m/2 $ to start the while loop from line~$4$ to $12$. Then we argue that $m \in [a_i,b_i]$ for all iterations $i$ of the while loop, and when the algorithm terminates, the output $\hat{m}$ is an $\eps$-approximation of $m$. 
    
    It is easy to see that $m \in [a_0, b_0]$. Suppose that $m \in [a_i, b_i]$, and we want to argue that $m \in [a_{i+1}, b_{i+1}]$. 
    Since $b_i/a_i>\gamma^3$, we have $a_i<\alpha_i<\beta_i<b_i.$ If the tester outputs $m\leq\alpha_i$, correctness on the upper promised case rules out $m\geq\beta_i$, and hence $m\in[a_i,\beta_i]=[a_{i+1},b_{i+1}]$. Similarly, if it outputs $m\geq\beta_i$, correctness on the lower promised case rules out $m\leq\alpha_i$, and hence $m\in[\alpha_i,b_i]=[a_{i+1},b_{i+1}]$. This also covers the case $\alpha_i<m<\beta_i$, since either updated interval contains $m$. Therefore, $m\in[a_i,b_i]$ throughout the algorithm. 
    
    Now, we argue that $\hat{m}$ is an $\varepsilon$-approximation of $m$. Note that when the algorithm terminates with $i = i^\ast$, $b_{i^\ast}/a_{i^\ast} \leq \gamma^3$. Since $m \in [a_{i^\ast},b_{i^\ast}]$, we have $\widehat{m}/m \leq \sqrt{b_{i^\ast}/a_{i^\ast}} \leq \gamma^{3/2}$ and $m/\widehat{m} \leq \sqrt{b_{i^\ast}/a_{i^\ast}} \leq \gamma^{3/2}$. Hence the multiplicative error of $\hat{m}$ relative to $m$ is at most $\gamma^{3/2}$, so 
    \begin{align*}
        \frac{|\widehat{m}-m|}{m} \leq \gamma^{3/2}-1 = \left(1+\frac{\varepsilon}{8}\right)^{3/2}-1 < \varepsilon.
    \end{align*}

    Now, we compute the query complexity of Algorithm \ref{alg:eps_approx_collisions}. We do that by bounding the number of iterations of the while loop and the query complexity of each iteration. 
    
    Let $\nu_i:=b_i/a_i$. Regardless of the outcome in the $i$th iteration,  $\nu_{i+1} = \gamma\sqrt{\nu_i}.$ Since $\nu_0=3$, solving this recurrence gives $\nu_i = \gamma^{2(1-2^{-i})}3^{2^{-i}}.$ Since $\log\gamma=\Theta(\varepsilon)$, after $O(\log(1/\varepsilon))$ iterations we have $\nu_i\leq\gamma^3$, and the algorithm terminates. 
    
    We use Theorem \ref{thm:gap_dist_collisions} to bound the query complexity of each iteration and the query complexity for the line~$2$ of the algorithm. Note that $a_0, b_0 = \Theta(m)$. Since, from an earlier argument, $m \in [a_{i+1}, b_{i+1}] \subseteq [a_i, b_i]$, we have $a_i, b_i = \Theta(m)$ for all iterations $i$. It follows that $\mu_i, \alpha_i, \beta_i = \Theta(m)$. Now, $\beta_i - \alpha_i = \mu_i(\gamma - 1/\gamma) = \Theta(\varepsilon m)$. Therefore, the cost of evoking Theorem \ref{thm:gap_dist_collisions} in the $i$th iteration (and the total cost of the $i$th iteration) is $\widetilde{O}\left(\left(\frac{n \sqrt{\beta_i}}{\beta_i - \alpha_i}\right)^{2/3}\right) = \widetilde{O}\left(\left(\frac{n}{\varepsilon \sqrt{m}}\right)^{2/3}\right)$. Also, for the line~$2$, we know the total query complexity is $\tilde O((\frac{n}{\sqrt{m_0}})^{2/3}\cdot \log (m/m_0))$. Therefore combining the earlier observation about the number of iterations, the total query complexity of our algorithm is
    \begin{align*}
       \widetilde{O}\left(\left(\frac{n}{\varepsilon \sqrt{m}}\right)^{2/3}+\left(\frac{n}{ \sqrt{m_0}}\right)^{2/3}\right).
    \end{align*}
    As for the success probability of Algorithm~\ref{alg:eps_approx_collisions}, we let $M$ be the number of total invocation of Theorem~\ref{thm:gap_dist_collisions} we used in Algorithm~\ref{alg:eps_approx_collisions}, and we can simply set the failure probability $\tau$ in Theorem~\ref{thm:gap_dist_collisions} to be $1/(10M)$, which only incurs $\polylog(m,m_0,1/\eps,n)$ extra overhead here, and hence the success probability of the algorithm will be $\geq 0.9$ by the union bound.
    To amplify the success probability from constant to $1-\tau$ only incurs extra $\polylog (1/\tau)$ overhead, and therefore we finish the proof for the first part.

    \textbf{Item 2}:
    To prove item two of the theorem,  we reduce it to item one with a known collision lower bound by padding dummy symbols to both $A$ and $B$. Precisely, let $ A^+ = A\circ(\bot_1,\ldots,\bot_t),
    \text{ and }
    B^+ = B\circ(\bot_1,\ldots,\bot_t),$ where the symbols $\bot_1,\ldots,\bot_t$ are distinct and do not appear in either $A$ or $B$. One can easily see now new non-repetitive strings $A^+$ and $B^+$ have length $n+t$ and the number $m^+$ of collisions between $A^+$ and $B^+$ is exactly $m^+=m+t\geq t$.

    By applying item one with $\eps_1=1/10$ on $A^+$ and $B^+$, with success probability $\geq 1-\tau/2$, we can obtain $m_1^+$ s.t. $|m_1^+ -m^+|\leq m^+/10$. Now we apply item one again with $\eps_2 := \frac{t}{2m_1^+},$ and then with success probability $\geq 1-\tau/2$, we can obtain $m_2^+$ s.t. $|m_2^+-m^+|\leq \eps_2 \cdot m^+$, which implies $|(m_2^+-t)-m|\leq \frac{t}{2m_1^+}m^+\leq t$ as desired.      
    
    As for the total number of queries, for the part where we apply item one with $\varepsilon_1=1/10$ on $A^+$ and $B^+$, it uses 
\[
    \widetilde O\left(
        \left(\frac{n+t}{\sqrt{m+t}}\right)^{2/3}
        +
        \left(\frac{n+t}{\sqrt t}\right)^{2/3}
    \right)=
    \widetilde O\left(
        \left(\frac{n+t}{\sqrt t}\right)^{2/3}
    \right)
\]
many queries, and for the part applying item one with $\eps_2= \frac{t}{2m_1^+}=\Theta(\frac{t}{m+t})$ it uses
\[
    \widetilde O\left(
        \left(\frac{n+t}{\eps_2\sqrt{m+t}}\right)^{2/3}
        +
        \left(\frac{n+t}{\sqrt t}\right)^{2/3}
    \right)=
    \widetilde O\left(
        \left(\frac{(n+t)\sqrt{m+t}}{t}\right)^{2/3}
        +
        \left(\frac{n+t}{\sqrt t}\right)^{2/3}
    \right)
\]
many queries. Moreover, observe that
$
    \left(\frac{n+t}{\sqrt t}\right)^{2/3}
    \le
    \left(\frac{(n+t)\sqrt{m+t}}{t}\right)^{2/3}. 
$
Hence the total number of queries is $$\tilde{O}(\left(\frac{(n+t)\sqrt{m+t}}{t}\right)^{2/3})=\tilde{O}(\left(\frac{n\sqrt{m+t}}{t}\right)^{2/3}).$$
\end{proof}

The quantum procedures developed from the above collision subroutines will later be invoked with different parameters and hence can have highly nonuniform query costs. To combine such procedures efficiently, we use the following variable-time approximate-counting primitive. We refer to~\cite{LeGall14Triangle,Jeffery22Composition,Ambainis10VariableTime} for a more detailed discussion of the variable-time model.

\begin{theorem}[Variable-time approximate counting tester]
\label{thm:VTcounting}
Let $\eps,\tau\in(0,1)$, $p_1,\ldots,p_k\in[0,1]$, $B=\sum_{s=1}^k p_s$, $t_1,\ldots,t_k>0$, and $T=\sum_{s=1}^k t_s^2$. Let threshold $W>0$. Suppose that, for each $s\in[k]$, there exists a quantum algorithm $\mathcal A_s$ using at most $t_s$ queries to a common input oracle, and let
\[
p_s:=\left\|\Pi_{\mathrm{acc}}\mathcal A_s\ket{0}\right\|^2
\]
denote its acceptance probability, where $\Pi_{\mathrm{acc}}$ projects onto its designated accepting subspace. There is a quantum algorithm that, with success probability $\geq1-\tau$, decides whether $B\geq(1+\eps)W$ or $B\leq(1-\eps)W$ in query complexity
\[
\widetilde O\left(\frac{\sqrt T}{\eps\sqrt W}\right).
\]
\end{theorem}

\begin{proof}
Without loss of generality, we assume all $t_s$ are at least $1$. 
The idea is to divide the $p_s$ into groups based on their costs $t_s$.
Define $G_j\equiv\{s\in [k]: 2^{j-1}< t_s \leq 2^j\}$, 
$B_j:=\sum_{s\in G_j}p_s,$
$T_j=\sum\limits_{s\in G_j} t_s^2$, and $W_j=W\cdot T_j/T$. Note that there are at most $1+\log_2 T$ many nonempty $G_j$ because $\max_s t_s\leq \sqrt{T}$. Now we would like to estimate each $B_j$ with additive error $\eps(B_j+W_j)/10$ because that would imply an estimator $\hat B=\sum_j\hat B_j$ for $B$ with total additive error $|\hat B-B|\leq\sum_j\eps(B_j+W_j)/10=\eps(B+W)/10$. We output the \textsf{YES} case if $\hat B>W$ and the \textsf{NO} case otherwise. If $B\geq(1+\eps)W$, then $\hat B\geq B-\eps(B+W)/10>W$, while if $B\leq(1-\eps)W$, then $\hat B\leq B+\eps(B+W)/10<W$. In particular, if $B=0$, then $\hat B\leq\eps W/10<W$, and hence the algorithm outputs the \textsf{NO} case.

For every $j$ such that $\frac{100}{\eps}\sqrt{|G_j|/W_j}<1$, we simply set $\hat B_j=0$. In this case $|G_j|<\eps^2W_j/10000$, and hence $|\hat B_j-B_j|=B_j\leq |G_j|\leq\eps(B_j+W_j)/10$.  
For every other $j$, prepare a uniform superposition $\frac{1}{\sqrt{|G_j|}}\sum_{s\in G_j}\ket{s}\ket{0}$ over $s\in G_j$ and apply the corresponding algorithm $\mathcal A_s$ controlled on the first register. The acceptance probability of the resulting procedure is
\[
p_j=\frac{1}{|G_j|}\sum_{s\in G_j}p_s=\frac{B_j}{|G_j|}.
\]
The procedure can be implemented using
$\widetilde O(\max_{s\in G_j}t_s)=\widetilde O(2^j)$ queries.

Now if we apply amplitude estimation (Theorem~\ref{thm:amplitude_estimation}) with failure probability at most $\tau/(1+\lceil\log_2T\rceil)$ and with $M_j=\left\lceil\frac{100}{\eps}\sqrt{\frac{|G_j|}{W_j}}\right\rceil$ for each $j$, then we can get $\hat p_j$ such that $|\hat p_j-p_j|\leq 2\pi\sqrt{p_j}/M_j+\pi^2/M_j^2$. Multiplying both sides by $|G_j|$ and substituting $M_j$, we get
$$
|\hat B_j-B_j|\leq \frac{2\pi\eps}{100}\sqrt{B_jW_j}+\frac{\pi^2\eps^2}{10000}W_j\leq \frac{\eps}{10}(B_j+W_j),
$$
where we let $\hat B_j\equiv |G_j|\hat p_j$ and the second inequality uses the fact $2\sqrt{B_jW_j}\leq B_j+W_j$. By a union bound over all $G_j$, all these estimates succeed simultaneously with probability at least $1-\tau$.

Now we calculate our total cost. For every $j$, we use $2^j$ cost to implement the superposition query and Theorem~\ref{thm:amplitude_estimation} uses $\widetilde O(M_j)$ queries, so the total cost for estimating each nontrivial $B_j$ is
$$
2^jM_j=\widetilde O\left(\frac{2^j}{\eps}\sqrt{\frac{|G_j|}{W_j}}\right)\leq \widetilde O\left(\frac{2}{\eps}\sqrt{\frac{T_j}{W_j}}\right)=\widetilde O\left(\frac{1}{\eps}\sqrt{\frac{T}{W}}\right),
$$
where the inequality follows because $T_j=\sum\limits_{s\in G_j}t_s^2\geq |G_j|(2^{j-1})^2$, and the last equality follows from $W_j=WT_j/T$. Summing over the at most $1+\log_2T$ groups and absorbing the logarithmic factors into the $\widetilde O$ notation, the total number of queries is $\widetilde O(\sqrt{T}/(\eps\sqrt W))$.
\end{proof}

Finally, besides estimating the number of collisions, we will sometimes need to locate the collisions themselves. We therefore include the following two additional primitives for collision problems.
The first one finds all collisions between two non-repetitive strings. A self-contained proof, adapting the algorithm of~\cite{Bonnetain2025multiplecollision} to a fixed matching, is given in Appendix~\ref{app:fixed-matching-collisions}.

\begin{theorem}
\label{thm:findingcollisions}
Let $A,B\in\Sigma^n$ be non-repetitive strings, and let $m$ be the
number of collisions between $A$ and $B$. Given quantum query access
to $A$ and $B$, there is a quantum algorithm that, with probability
at least $1-\tau$, finds all $m$ collisions using
\[
  \widetilde O\!\left(
    n^{2/3}m^{1/3}\log(1/\tau)
  \right)
\]
expected queries.
\end{theorem}

\begin{proof}
See Appendix~\ref{app:fixed-matching-collisions}.
\end{proof}

We will also need to find a single collision, sampled uniformly among the existing collisions. For this, we use the following quantum pair-finding result.

\begin{theorem}[Quantum pair finding]
\label{thm:quantumpairfinding}
Let $X,Y$ be two sets of sizes $N\leq M$, and let
\[
\mathcal C:=\{(x,y)\in X\times Y:f(x)=g(y)\}
\]
consist of $L$ pairwise disjoint pairs. Suppose that a lower bound $K\leq L$ is known and that $N\leq M\leq KN^2$. There is a quantum algorithm that, with success probability at least $2/3$, returns a pair in $\mathcal C$ using
\[
\widetilde O\left(\left(\frac{NM}{K}\right)^{1/3}\right)
\]
queries. Moreover, conditioned on returning a pair, every pair in $\mathcal C$ is returned with probability $1/L$.
\end{theorem}

\begin{proof}
The query bound follows implicitly from Theorem~2.8 of~\cite{AHJKS22}. To obtain the uniform-output guarantee, independently apply uniformly random permutations to the elements of $X$ and $Y$, run the product-Johnson-graph implementation of the pair-finding algorithm on the relabelled instance, and undo the permutations on the returned pair. The initial state, walk operations, and mark test of this implementation are invariant under such relabellings. Therefore, for any two pairs $e,e'\in\mathcal C$, the probabilities of returning $e$ and $e'$ are equal. If $\mathsf{Succ}$ denotes the event that a pair is returned, then 
$\Pr[\text{output }e\mid\mathsf{Succ}] = \frac{\Pr[\text{output }e]} {\sum_{e'\in\mathcal C}\Pr[\text{output }e']} = \frac1L.$
The random relabelling uses no input queries and hence does not change the query complexity.
\end{proof}

\subsection{Quantum \textsc{GapUlamTest}}
\label{sec:QGapUlamTester}

We first give a quantum analogue of the \textsc{GapUlamTest} of Andoni and Nguyen~\cite{AN10UlamTester}. Given two non-repetitive strings $A,B\in\Sigma^n$ and thresholds $a<b$, the goal is to distinguish
$$
    \ed(A,B)\leq a
    \qquad\text{from}\qquad
    \ed(A,B)\geq b.
$$
We retain the multiscale framework of~\cite{AN10UlamTester}, replacing its local estimation procedures by the quantum collision primitives developed in the previous subsection. 
We briefly recall the structure of their tester. Partition $A$ into blocks of length $a$, and associate to each block a local window of $B$ of length $3a$:
\begin{align}
\label{eq:GapStringsPartition}
A_k=A[ka+1,\ldots,ka+a],
\qquad
B_k=B[ka-a+1,\ldots,ka+2a].
\end{align}
In~\cite{andoni2012smoothed}, they characterize large Ulam distance through two types of contributions. The first is
$$
    X:=\sum_{k=0}^{n/a-1}|A_k\setminus B_k|,
$$
which counts symbols of $A_k$ that have no local match in $B_k$. The second, denoted $Y_\delta$, counts locally matching pairs $A[u]=B[v]$ that are significantly misaligned: namely, pairs with $A[u]\in A_k\cap B_k$ for some $k$ and
$$
    |A[u-t,\ldots,u-1]\triangle B[v-t,\ldots,v-1]|
    >2\delta t
$$
for some $t\leq4a$. Their analysis shows that estimating these two quantities at geometrically increasing scales suffices to distinguish the two cases above. The $X$-contribution can be handled using our additive collision estimator, since
$$
    |A_k\setminus B_k|
    =
    a-|A_k\cap B_k|.
$$
For the $Y$-contribution, we first use Theorem~\ref{thm:findingcollisions} to enumerate suitably sampled local collisions, and then determine which of them are significantly misaligned. The latter requires one additional primitive.

\begin{theorem}
\label{cor:ytest}
Let $A,B\in\Sigma^n$ be non-repetitive strings, let $\tau\in(0,1)$, and let $a\in\mathbb N$. Given positions $u,v$, there is a quantum algorithm that 
makes $\widetilde O(a^{1/3})$ queries and with probability  $\geq 1-\tau$, distinguishes
$$
    \exists\,t\in[4a]:
    |A[u-t,\ldots,u-1]\triangle B[v-t,\ldots,v-1]|>t $$
from
$$\forall\,t\in[4a]:
        |A[u-t,\ldots,u-1]\triangle B[v-t,\ldots,v-1]|
    \leq0.8t.
$$

\end{theorem}

\begin{proof}
Let
$
    \mathcal Z
    :=
    \{\lfloor1.01^j\rfloor:
      0\leq j\leq\lceil\log_{1.01}(4a)\rceil\}
    \cap[4a].
$
For $z\in\mathcal Z$, define
$$
    m_z
    :=
    |A[u-z,\ldots,u-1]\cap B[v-z,\ldots,v-1]|.
$$
Since $A$ and $B$ are non-repetitive,
$$
    |A[u-z,\ldots,u-1]\triangle B[v-z,\ldots,v-1]|
    =
    2z-2m_z.
$$
Hence a symmetric difference larger than $0.98z$ implies $m_z<0.51z$, whereas a symmetric difference at most $0.8z$ implies $m_z\geq0.6z$. By Theorem~\ref{cor:legall}, we can estimate $m_z$ to additive error $0.02z$ using $\widetilde O(z^{1/3})$ queries. We therefore accept at scale $z$ whenever the estimate is smaller than $0.55z$. Run this test for every $z\in\mathcal Z$. Since $|\mathcal Z|=O(\log a)$, 
$$
    \sum_{z\in\mathcal Z}\widetilde O(z^{1/3})
    =
    \widetilde O(a^{1/3}).
$$
If the YES case holds for some $t$, choose $z\in\mathcal Z$ such that $z\leq t<1.02z$. Removing the final $t-z$ symbols from each substring changes their symmetric difference by at most $2(t-z)$, and therefore
$$
    |A[u-z,\ldots,u-1]\triangle B[v-z,\ldots,v-1]|
    >
    t-2(t-z)
    =
    2z-t
    >
    0.98z.
$$
Thus the test accepts at scale $z$. In the NO case, every $z\in\mathcal Z$ has symmetric difference at most $0.8z$, so every test rejects. Amplifying each call and taking a union bound completes the proof.
\end{proof}

We now plug these quantum primitives into the multiscale sampling scheme of~\cite{AN10UlamTester}. We give the algorithm explicitly for completeness. All quantum subroutines below are amplified so that their total failure probability is at most $\tau/2$.

\begin{algorithm}[H]
\caption{\textsc{Q-GapUlamTest}$(A,B,a,b,\tau)$}
\label{alg:gapulamtest}
\setlength{\baselineskip}{1.25\baselineskip}
\begin{algorithmic}[1]
\Require Non-repetitive strings $A,B$; thresholds $a<b/512$; failure probability $\tau\in(0,1)$.
\Ensure Distinguish $\ed(A,B)\leq a$ ($\textsf{CLOSE}$) vs.\ $\ed(A,B)\geq b$ ($\textsf{FAR}$).

\State Define $A_k,B_k$ as in Eq~(\ref{eq:GapStringsPartition}), and set
$L_a\gets\lfloor\log_2a\rfloor+1$.

\For{$i\gets0$ \textbf{to} $L_a-1$}
    \State Set $\widehat X_i\gets0$, $\widehat Y_i\gets0$,
    $q_i\gets\min\left\{\frac{2000L_a2^i\ln(4L_a/\tau)}{b},1\right\}$, and
    $r_i\gets\min\left\{\sqrt{\frac{10000\ln(8nL_a/(a\tau))}{2^i}},1\right\}$.
    \State Let $S\subseteq[n/a]$ include each block independently with probability $q_i$.

    \For{\textbf{each} sampled block $k\in S$}
        \State \textit{// $X$-contribution estimation:} \\ Let $\widetilde c_k$ be an
        $(0.1\cdot2^i)$-additive approximation of
        $c_k:=|A_k\cap B_k|$ by Theorem~\ref{cor:legall}, with failure
        probability at most $\tau a/(2nL_a)$.
        \State \textbf{If} $a-\widetilde c_k\geq0.9\cdot2^i$, then
        $\widehat X_i\gets\widehat X_i+1$.

        \State \textit{// $Y$-contribution estimation:} \\ Retain each position in
        $A_k,B_k$ independently with probability $r_i$, keeping its original
        index. Call the resulting strings $A'_k,B'_k$, and find all collisions
        between them using Theorem~\ref{thm:findingcollisions}.
        \State For each sampled collision $(u,v)$, apply
        Corollary~\ref{cor:ytest} to $(A,B,u,v,a)$ with failure probability at
        most $\tau/(4nL_a)$. Let $D$ be the number of pairs for which it returns
        \textsf{YES}.
        \State \textbf{If} $D/r_i^2\geq0.9\cdot2^i$, then
        $\widehat Y_i\gets\widehat Y_i+1$.
    \EndFor
\EndFor

\State Set
$
\widehat X\gets
\sum_{i=0}^{L_a-1}\frac{2^i}{q_i}\widehat X_i,
\qquad
\widehat Y\gets
\sum_{i=0}^{L_a-1}\frac{2^i}{q_i}\widehat Y_i.
$

\State \textbf{If} $\widehat X+\widehat Y>b/10$ \textbf{then return} \textsf{FAR}; \textbf{else return} \textsf{CLOSE}.

\end{algorithmic}
\end{algorithm}

\begin{theorem}
\label{thm:gapulamcomplexity}
Let $A,B\in\Sigma^n$ be non-repetitive strings, and let $a,b$ satisfy $a<b/512$. Given quantum query access to $A$ and $B$, there is a bounded-error quantum algorithm that distinguishes
$
    \ed(A,B)\leq a
 \text{ from }    \ed(A,B)\geq b
$
using expected
$
    \widetilde O\left({na^{1/3}}/{b}\right)
$
queries.
\end{theorem}

\begin{proof}
We analyze Algorithm~\ref{alg:gapulamtest}. Its correctness follows from the same multiscale sampling argument as the \textsc{GapUlamTest} of~\cite{AN10UlamTester}: the algorithm uses the same block decomposition, distance scales, and sampling probabilities, while our quantum routines approximate the corresponding $X$- and $Y$-contributions within the slack allowed by their analysis. We give the adapted correctness argument in Appendix~\ref{sec:correctGapUlamTester}. It remains to bound the expected query complexity. Ignoring logarithmic factors,
$$
    q_i=\widetilde O(2^i/b),
    \qquad
    r_i=\widetilde O(2^{-i/2}).
$$
We separately bound the costs of estimating the $X$-contribution, enumerating collisions, and testing whether the surviving collisions contribute to $Y$.

\textbf{$X$-contribution.}
For a sampled block $k$, Theorem~\ref{cor:legall} estimates 
$c_k=|A_k\cap B_k|$ to additive error $\Theta(2^i)$~with
$$
    \widetilde O\left(
        \left(
        \frac{a\sqrt{c_k+2^i}}{2^i}
        \right)^{2/3}
    \right)
    \leq
    \widetilde O\left(
        \frac{a}{2^{2i/3}}
    \right)
$$
queries, where we use $c_k\leq a$ and $2^i\leq a$. The expected number of sampled blocks at level $i$ is
$
    \frac{n}{a}q_i
    =
    \widetilde O\left(
        \frac{2^in}{ab}
    \right).
$
Thus the expected cost at level $i$ is
$
    \widetilde O\left(
        \frac{2^{i/3}n}{b}
    \right),
$
and summing over $i=0,\ldots,L_a-1$ gives
$
    \widetilde O\left(
        \frac{na^{1/3}}{b}
    \right).
$

\textbf{Enumerating collisions.}
For a sampled block $k$, let $N'_{k,i}$ denote the length of the subsampled strings and let $c'_{k,i}$ denote their number of collisions. Then
$$
    \Exp[N'_{k,i}]=O(ar_i),
    \qquad
    \Exp[c'_{k,i}]=c_kr_i^2.
$$
By Theorem~\ref{thm:findingcollisions}, the inequality $(x+1)^{1/3}\leq x^{1/3}+1$, H\"older's inequality, and Jensen's inequality, the expected cost of enumerating these collisions is
\[
    \widetilde O\left(
        \Exp\left[
            (N'_{k,i})^{2/3}(c'_{k,i}+1)^{1/3}
        \right]
    \right)
    \leq
    \widetilde O\left(
        (ar_i)^{2/3}(c_kr_i^2)^{1/3}
        +(ar_i)^{2/3}
    \right)
    \leq
    \widetilde O\left(
        \frac{a}{2^{2i/3}}
    \right),
\]
where we use $c_k\leq a$, $r_i=\widetilde O(2^{-i/2})$, and $2^i\leq a$ for the second term.
Multiplying by the expected number of sampled blocks gives
$
    \widetilde O\left(
        \frac{2^{i/3}n}{b}
    \right)
$
queries at level $i$, and hence
$
    \widetilde O\left(
        \frac{na^{1/3}}{b}
    \right)
$
queries over all levels.

\textbf{Testing collisions.}
Since the strings are non-repetitive and the blocks $A_k$ are disjoint,
$
    \sum_k c_k\leq n.
$
At level $i$, each collision is considered with probability $q_ir_i^2$. Therefore, the expected number of collisions tested at this level is at most $
    q_ir_i^2\sum_kc_k
    \leq
    q_ir_i^2n
    =
    \widetilde O\left(\frac{n}{b}\right).
$
By Theorem~\ref{cor:ytest}, testing each collision costs
$\widetilde O(a^{1/3})$ queries. Thus the expected cost at every level is
$
    \widetilde O\left(
        \frac{na^{1/3}}{b}
    \right).
$
There are only $O(\log a)$ levels, which is absorbed into the $\widetilde O(\cdot)$.
 Combining the three costs proves the claimed
$
    \widetilde O\left(
        \frac{na^{1/3}}{b}
    \right)
$
expected query complexity.
\end{proof}

\subsection{Quantum Ulam product tester}
\label{sec:QUlamProcut}

We next give a quantum analogue of the \textsc{UlamProductTest} of~\cite{AN10UlamTester}. We are given $k$ pairs of non-repetitive strings
$$
    (A_1,B_1),\ldots,(A_k,B_k),
$$
each of length at most $\beta R$, and wish to distinguish whether their total Ulam distance
$$
    D:=\sum_{s=1}^k\ed(A_s,B_s)
$$
is large or small. 
One could apply the \textsc{Q-GapUlamTest} from the previous subsection separately to every pair. At distance scale $R$, each such test costs
$
    \widetilde O(\beta R^{1/3}),
$
leading to a total cost $\widetilde O(k\beta R^{1/3})$. To improve the dependence on $k$, we follow the multiscale viewpoint of~\cite{AN10UlamTester}: rather than estimating every
$    D_s:=\ed(A_s,B_s)$  
 For each scale $2^j$, ideally we would like to know
$$
    x_{s,j}^\star:=\mathbf 1[D_s\geq2^j],
    \qquad
    B_j^\star:=\sum_{s=1}^k x_{s,j}^\star.
$$
The weighted counts $B_j^\star$ recover the total distance up to constant factors. Indeed, if we ignore distances below some threshold $b_{\min}=2^{j_0}$ and write
$
    E_{\min}:=\sum_{s:D_s<b_{\min}}D_s,
$
then
\begin{align}
\label{eq:dyadic-distance}
\frac12(D-E_{\min})
\leq
\sum_{j\geq j_0}2^j B_j^\star
\leq
2D.
\end{align}
This follows simply by summing the powers of two below each $D_s$. 
The advantage of this formulation is that different scales have different evaluation costs: \textsc{Q-GapUlamTest} becomes cheaper as the distance threshold grows. We therefore combine the scale-dependent tests using the variable-time approximate-counting primitive from Theorem~\ref{thm:VTcounting}, rather than paying the cost of the most expensive test on every~pair.

\begin{theorem}
\label{thm:QUlamProduct}
Let $R\in\mathbb N$, $\beta>0$, and let $(A_s,B_s)_{s=1}^k$ be non-repetitive string pairs satisfying
$
    |A_s|,|B_s|\leq\beta R.
$
Let
$
    D=\sum_{s=1}^k\ed(A_s,B_s).
$
Given quantum query access to all the strings, there is a bounded-error quantum algorithm that distinguishes
$
    D>R
    \text{ from }
    D<\frac{R}{32768}
$
using
$$
    \widetilde O\left(
        \beta\min\{\sqrt{kR},\,k^{2/3}R^{1/3}\}
    \right)
$$
queries.
\end{theorem}

\begin{proof}
Write
$
    D_s:=\ed(A_s,B_s).
$
Set
$$
    j_0
    :=
    \max\left\{
        0,
        \left\lceil\log_2\frac{R}{200k}\right\rceil
    \right\},
    \qquad
    j_{\max}
    :=
    \left\lceil\log_2(\beta R)\right\rceil,
$$
and let
$
    b_{\min}:=2^{j_0}.
$
The contribution
$
    E_{\min}
    :=
    \sum_{s:D_s<b_{\min}}D_s
$
of scales below $b_{\min}$ is less than $R/100$: if $j_0=0$, then $E_{\min}=0$ since the $D_s$ are integers, while if $j_0>0$, then
$
    b_{\min}<\frac{R}{100k}.
$
Let
$
    \eta
    :=
    \min\left\{
        \frac{1}{100},
        \frac{1}{400k\beta}
    \right\}.
$
For every $s\in[k]$ and $j\in\{j_0,\ldots,j_{\max}\}$, consider
$$
    \textsc{Q-GapUlamTest}
    \left(
        A_s,B_s,\frac{2^j}{1024},2^j
    \right),
$$
where output \textsf{FAR} corresponds to acceptance.\footnote{Note that we can pad the shorter string with fresh symbols to make them have the same length and still have length $O(\beta R)$.} Since Theorem~\ref{thm:gapulamcomplexity} gives an expected query bound, we truncate each invocation after a sufficiently large constant multiple of its expectation and amplify it so that its failure probability in either promised case is at most $\eta$. This only changes its query complexity by logarithmic factors. Let $p_{s,j}$ denote the acceptance probability of the resulting procedure. Then
$$
    (1-\eta)\mathbf 1[D_s\geq2^j]
    \leq
    p_{s,j}
    \leq
    \mathbf 1[D_s>2^j/1024]+\eta.
$$
These inequalities also hold when $2^j/1024<D_s<2^j$, since in this case they reduce to the trivial bounds $0\leq p_{s,j}\leq1$.
Consequently, for
$
    S
    :=
    \sum_{j=j_0}^{j_{\max}}
    2^j\sum_{s=1}^k p_{s,j},
$
the dyadic estimate~\eqref{eq:dyadic-distance} gives
\begin{align}
\label{eq:S-vs-D}
\frac{1-\eta}{2}(D-E_{\min})
\leq
S
\leq
2048D+\eta k\sum_{j=j_0}^{j_{\max}}2^j
\leq
2048D+\frac{R}{100},
\end{align}
where the last inequality uses
$
    \sum_{j=j_0}^{j_{\max}}2^j\leq4\beta R
$
and the choice of $\eta$.

It remains to test whether $S$ is large or small. To express the weights $2^j$ as an ordinary sum, for every $r\in[2^j]$ let $\mathcal A_{s,j,r}$ be a copy of the above procedure and let
$
    p_{s,j,r}:=p_{s,j}
$
denote its acceptance probability. Then
$$
    S
    =
    \sum_{j=j_0}^{j_{\max}}
    \sum_{s=1}^k
    \sum_{r=1}^{2^j}
    p_{s,j,r}.
$$
By Theorem~\ref{thm:gapulamcomplexity}, running $\mathcal A_{s,j,r}$ costs
$
    t_{s,j}
    =
    \widetilde O\left(
        \frac{\beta R}{2^{2j/3}}
    \right)
$
queries, where the truncation and amplification overheads are absorbed into the $\widetilde O$ notation. Thus the squared-cost parameter in Theorem~\ref{thm:VTcounting} satisfies
$$
\begin{aligned}
    T
    &:=
    \sum_{j=j_0}^{j_{\max}}
    \sum_{s=1}^k
    \sum_{r=1}^{2^j}
    t_{s,j}^2=
    \widetilde O\left(
        k\beta^2R^2
        \sum_{j=j_0}^{j_{\max}}2^{-j/3}
    \right)=
    \widetilde O\left(
        k\beta^2R^2\,2^{-j_0/3}
    \right).
\end{aligned}
$$
Applying Theorem~\ref{thm:VTcounting} with $\eps=1/3$, $W=R/4$, and constant failure probability therefore uses
$$
    \widetilde O\left(\frac{\sqrt T}{\sqrt R}\right)
    =
    \widetilde O\left(
        \beta\sqrt{kR}\,2^{-j_0/6}
    \right)=
    \widetilde O\left(
        \beta\min\{\sqrt{kR},\,k^{2/3}R^{1/3}\}
    \right)
$$
queries. 
Finally, if $D>R$ and $E_{\min}<R/100$, then, since $\eta\leq1/100$,
$$
    S>
    \frac{1-\eta}{2}
    \left(R-\frac{R}{100}\right)
    >
    \frac R3
    =
    (1+\eps)W.
$$

On the other hand, if
$
    D<\frac{R}{32768},
$ then
$$
    S
    \leq2048D+\frac{R}{100}
    <
    \frac{R}{16}+\frac{R}{100}
    <
    \frac R6
    =
    (1-\eps)W.
$$
Theorem~\ref{thm:VTcounting} therefore distinguishes the two cases, completing the proof.
\end{proof}

\subsection{Quantum algorithms for Ulam distance and smoothed edit distance}
\label{sec:QPartialAlign}

Our final ingredient is a quantum analogue of the partial-alignment procedure of~\cite{AN10UlamTester}.
Once this is established, we combine it with the quantum Ulam product tester from the previous subsection to obtain our quantum Ulam distance estimator, and then apply the reduction of Andoni and Krauthgamer~\cite{andoni2012smoothed} to obtain the smoothed edit distance result.
The purpose of the partial-alignment procedure is to find matching positions $(a_1,b_1),\ldots,(a_q,b_q)$ belonging to a common longest common subsequence of $A$ and $B$. Cutting the strings at these positions produces substring pairs $(A_s,B_s)$ whose edit distances add to $\ed(A,B)$ and whose lengths are $O(\beta R)$, as required for the product tester of Theorem~\ref{thm:QUlamProduct}.

We follow the partial-alignment framework of~\cite{AN10UlamTester}. Starting from a previously identified LCS pair $(a_{i-1},b_{i-1})$, their procedure examines progressively larger local windows around the corresponding alignment and searches for a matching pair $(a_i,b_i)$. A random local collision need not itself belong to the fixed LCS, but the analysis of Andoni and Nguyen shows that, when the local windows contain sufficiently many collisions, a uniformly sampled collision is an LCS pair with sufficiently high probability. Quantumly, we implement this local search using two primitives. We first use Theorem~\ref{cor:legall} to verify that the local windows contain sufficiently many collisions.  The preliminary collision-counting step is a minor modification of the procedure in~\cite{AN10UlamTester}. It does not change the asymptotic query complexity, but makes explicit the promise under which the quantum pair-finding routine is invoked.

\begin{algorithm}[H]
\caption{\textsc{QuantumPartialAlign}$(A,B,R,\beta)$}
\label{alg:quantumpartialalign}
\begin{algorithmic}[1]
\Require Non-repetitive strings $A,B$; distance threshold $R$; block parameter $\beta$.
\Ensure Output matching positions $(a_i,b_i)_i$ forming a subsequence of an LCS of $A,B$, or return \textsf{FAIL}.

\State Set $m_0\gets0$.
\For{$i=1,\ldots,\lceil n/(\beta R)\rceil$}
\State $\mathsf{found}\gets\mathsf{false}$.
\For{$j=4,\ldots,\lceil\log_2(4R)\rceil$}
\State Set $z\gets2^j$ and sample $p\sim\mathrm{Unif}([(i-1)\beta R+4R,;i\beta R-4R])$.
\State Set $I\gets[p,p+z]$ and $J\gets[p+m_{i-1}-z,;p+m_{i-1}+2z]$.
\State Estimate the number $M$ of collisions between $A[I]$ and $B[J]$ to additive error $0.02z$ using Theorem~\ref{cor:legall}, obtaining $\widetilde M$.
\State If $\widetilde M<0.85z$, \textbf{continue}.
\State Use Theorem~\ref{thm:quantumpairfinding} to find a collision $(u,v)$ between $A[I]$ and $B[J]$.
\State Set $a_i\gets u$, $b_i\gets v$, $m_i\gets v-u$, set $\mathsf{found}\gets\mathsf{true}$, and \textbf{break}.
\EndFor
\State If $\mathsf{found}=\mathsf{false}$, \textbf{return} \textsf{FAIL}.
\EndFor
\State \Return $(a_1,b_1),\ldots,(a_q,b_q)$, where $q=\lceil n/(\beta R)\rceil$.
\end{algorithmic}
\end{algorithm}

\begin{lemma}
\label{lem:quantumpartialalignment}
Let $A,B\in\Sigma^n$ be non-repetitive strings with $\ed(A,B)\leq R$, and let $\beta=\Omega(\log^4 n)$. Given quantum query access to $A$ and $B$, Algorithm~\ref{alg:quantumpartialalign} succeeds with probability at least $2/3$ and outputs matching positions $(a_i,b_i)$ belonging to a common LCS such that $|a_{i+1}-a_i|,\ |b_{i+1}-b_i|\leq2\beta R$ for every $i$. Its query complexity is

$$
    \widetilde O\left(
        \frac{n}{\beta R}
        +
        \frac{n^{2/3}}{\beta^{2/3}R^{1/3}}
    \right).
$$

\end{lemma}

\begin{proof}
The correctness argument follows the partial-alignment analysis of~\cite{AN10UlamTester}. Our only modification is the preliminary collision-counting test before invoking quantum pair finding. Conditioned on this test accepting, there are at least $0.8z$ local collisions, and Theorem~\ref{thm:quantumpairfinding} returns one uniformly among them. Thus the argument bounding the probability of selecting a collision outside the fixed LCS applies as in~\cite{AN10UlamTester}. We give the adapted argument, together with the guarantee $|a_{i+1}-a_i|,\ |b_{i+1}-b_i|\leq2\beta R$, in Appendix~\ref{sec:correctPartialAlign}. It remains to bound the query complexity. Fix a longest common subsequence of $A$ and $B$, and let $S_A,S_B$ denote its position sets. Conditioned on $(a_{i-1},b_{i-1})$ lying on this LCS, define

$$
    f_i
    :=
    |[a_{i-1},i\beta R]\setminus S_A|
    +
    |[b_{i-1},m_{i-1}+i\beta R]\setminus S_B|,
$$

where $m_{i-1}=b_{i-1}-a_{i-1}$. The analysis of~\cite{AN10UlamTester}, recalled in Appendix~\ref{sec:correctPartialAlign}, gives $\sum_i f_i\leq4R$ and shows that, with high probability, the inner loop for block $i$ terminates by a scale $z=O(\max{1,f_i})$.\footnote{We truncate Algorithm~\ref{alg:quantumpartialalign} after a sufficiently large constant times the stated query bound and return \textsc{Fail} upon timeout. On the good event analyzed below, the $f_i$-based charging argument stays within this budget; otherwise, the timeout is absorbed into the failure probability.}

We therefore only need to bound the cost of one scale $z$. By Theorem~\ref{cor:legall}, estimating the number of collisions between $A[I]$ and $B[J]$ to additive error $0.02z$ requires $\widetilde O(z^{1/3})$ queries. If the estimate satisfies $\widetilde M\geq0.85z$, then the true number of collisions is at least $0.83z>0.8z$. Since the two local windows have lengths $\Theta(z)$, Theorem~\ref{thm:quantumpairfinding} finds a collision using
$
    \widetilde O\left(
        \left(
            \frac{z^2}{z}
        \right)^{1/3}
    \right)
    =
    \widetilde O(z^{1/3})
$
queries. Thus each scale costs $\widetilde O(z^{1/3})$. 
Because the scales grow geometrically, the total cost for block $i$ is
$$
    \widetilde O\left(
        1+
        \sum_{j:\,2^j\leq O(f_i)}2^{j/3}
    \right)
    =
    \widetilde O(1+f_i^{1/3}).
$$
There are $q=\lceil n/(\beta R)\rceil$ blocks. Hence, by Holder's inequality and $\sum_i f_i\leq4R$,
$$
\begin{aligned}
    \widetilde O\left(
        q+\sum_{i=1}^q f_i^{1/3}
    \right)
    \leq
    \widetilde O\left(
        q+
        q^{2/3}
        \left(\sum_i f_i\right)^{1/3}
    \right)=
    \widetilde O\left(
        \frac{n}{\beta R}
        +
        \frac{n^{2/3}}
             {\beta^{2/3}R^{1/3}}
    \right),
\end{aligned}
$$
which proves the claimed query bound.
\end{proof}

\paragraph{Putting everything together.} 
As we claimed at the beginning of this section, we will give a quantum algorithm for estimating edit distance (with $\sim1/\sigma$ approximation factor) in the $\sigma$-smoothed model, using $\widetilde O\left(\frac{1}{\sigma}\left(\frac{n}{R}+\frac{n^{2/3}}{R^{1/3}}\right)\right)$ quantum queries, where $R$ denotes the edit distance between the two input strings. In particular, the query complexity becomes $\widetilde O( n^{1/3})$ when $R=\Omega(n)$ and $\sigma$ is constant. We start with putting previous subsections together to have the following quantum algorithm for estimating Ulam distance

\thmrestateUlam*

\begin{proof}
Let $\beta=100\log^{10} n$.  We first show how to construct a tester for a guess $r$ under the promise that $\textsf{ed}(A,B)\leq 2r$. Apply Lemma~\ref{lem:quantumpartialalignment} with parameter $2r$. It decomposes $A,B$ into $k=O(n/(\beta r))$ pairs (in terms of positions) of non-repetitive substrings $(A_s,B_s)$ such that
    $|A_s|,|B_s|=O(\beta r) $ and $\sum_s \textsf{ed}(A_s,B_s)=\textsf{ed}(A,B).$

Storing those positions in quantum-readable, classical-writable classical memory to allow quantum queries, we can then apply Theorem~\ref{thm:QUlamProduct} to distinguish the case
    $\textsf{ed}(A,B)>r$ form the case $\textsf{ed}(A,B)<\frac{r}{32768}$.
The query complexity of the quantum product tester is
$\widetilde O\left(\beta\min\left\{\sqrt{kr},k^{2/3}r^{1/3}\right\}\right)=\widetilde O\left(    \min\left\{\sqrt n,\frac{n^{2/3}}{r^{1/3}}\right\}\right)$.
Together with the quantum partial-alignment algorithm, the total query complexity
of the tester is therefore $\widetilde O\left(\frac{n}{r}+\frac{n^{2/3}}{r^{1/3}}\right)$. Now we explain how to use the above test to estimate $R=\textsf{ed}(A,B)$: We run the above tester for $r=n/2,n/4,n/8,\ldots$ and stop at the first value of $r$ for which it returns \textsc{Far}. By the guarantees above, the returned $r$ is within a constant factor of $R$. Amplifying the success probability for each guess only incurs polylogarithmic overhead. Moreover, since the query complexity increases geometrically as $r$ decreases, the total query complexity is dominated by the last guess which is 
\[
    \widetilde O\left(
        \frac{n}{R}+\frac{n^{2/3}}{R^{1/3}}
    \right).
\]
\end{proof}

Finally, we combine Theorem~\ref{thm:ulam} with the reduction of Andoni and Krauthgamer from smoothed edit distance to Ulam distance. There is one additional point that must be addressed for quantum queries.  The procedure in~\cite{andoni2012smoothed} for answering a query to the implicit Ulam instance maintains tables of the positions queried so far, and hence does not directly define a fixed unitary of the form
\[
    |i,z\rangle
    \longmapsto
    |i,z\oplus P[i]\rangle .
\]

In Appendix~\ref{app:ak-quantum-access}, we instead give a reversible implementation tailored to the subset databases appearing in the proof of Theorem~\ref{thm:ulam}.  Conditioned on the high-probability event of Lemma~4.6 of~\cite{andoni2012smoothed}, each stored position is represented by $O(L)$ entries of the original strings, where $ L=\Theta((\log n)/\sigma).$ These local records determine all collisions among the positions currently stored in the database.  The setup cost is multiplied by $O(L)$, replacing one stored position costs $O(L)$ queries, and checking for collisions requires no further queries to the original strings. This gives the following corollary.

\SmoothedEditDistanceCor*

\begin{proof}
Apply Lemma~4.6 of~\cite{andoni2012smoothed} to $X,Y$, and let $P,Q$ denote
the resulting non-repetitive strings.  Write $R_{\mathrm U}:=\ed(P,Q).$
Conditioned on the high-probability event in that lemma,
\[
    \Omega(1)\,\ed(X,Y)
    \leq
    R_{\mathrm U}
    \leq
    \widetilde O(1/\sigma)\,\ed(X,Y).
\]

Let $L=\Theta((\log n)/\sigma).$ By Corollary~\ref{cor:ak-quantum-access}, the proof of Theorem~\ref{thm:ulam} can be executed on the implicit strings $P,Q$ using quantum query access only to $X,Y$, with every setup and update query cost multiplied by $O(L)$.  Its total query complexity
is therefore
    $\widetilde O\left(L\left(\frac{n}{R_{\mathrm U}}+\frac{n^{2/3}}{R_{\mathrm U}^{1/3}}\right)\right).$
Since $R_{\mathrm U}=\Omega(\ed(X,Y))$, this is
\[
    \widetilde O\left(
        \frac{1}{\sigma}
        \left(
            \frac{n}{\ed(X,Y)}
            +
            \frac{n^{2/3}}{\ed(X,Y)^{1/3}}
        \right)
    \right).
\]
The comparison between $\ed(P,Q)$ and $\ed(X,Y)$ gives the claimed
approximation guarantee.
\end{proof}

\bibliographystyle{alpha}
\bibliography{refs}

\appendix
\section{Missing proofs}
\subsection{Proof of the correctness of Algorithm~\ref{alg:gapulamtest}}\label{sec:correctGapUlamTester}

\begin{theorem}[\cite{AN10UlamTester}, Lemma~4.3]
\label{thm:an10}
Let $A,B\in \Sigma^n$ be non-repetitive strings, let $a\mid n$, and let
$\delta\leq 1/2$. Let
$
A_k=A[ka+1,\ldots,ka+a],\hspace{2mm}
B_k=B[ka-a+1,\ldots,ka+2a],\hspace{2mm}
I_k=A_k\backslash B_k,
$
and $X=\sum_{k=0}^{n/a-1}|I_k|$. Define $Y_\delta$ as the number of pairs
$u,v$ such that $A[u]=B[v]$ and
$A[u]\in A_k\cap B_k$ for some $k\in\{0,\ldots,n/a-1\}$, along with
$
|A[u-t,\ldots,u-1]\Delta B[v-t,\ldots,v-1]|>2\delta t
$
for some $t\leq 4a$. Then,
\begin{enumerate}
    \item If $\ed(A,B)\leq a$, then $X\leq a$ and
    $Y_\delta\leq 4a/\delta$.
    \item If $\ed(A,B)\geq b$, then
    $X+Y_\delta\geq b(1-\delta)/2$.
\end{enumerate}
\end{theorem}

\begin{proofof}{the correctness of Algorithm~\ref{alg:gapulamtest}}
     We first show the correctness of Algorithm~\ref{alg:gapulamtest}. Since $a<b/512$, if $\ed(A,B)\geq b$, then applying Theorem~\ref{thm:an10} with $\delta=0.5$ gives
$X+Y_{0.5}\geq b/4$. If $\ed(A,B)\leq a$, then applying Theorem~\ref{thm:an10} with $\delta=0.4$ gives
$X\leq a$ and $Y_{0.4}\leq 10a$. Therefore, it suffices to show that, with success probability $\geq 1-\tau$, $
0.9X-0.01 b\leq \widehat X\leq 3X+0.01 b$
and
$
0.9Y_{0.5}-0.01 b\leq \widehat Y\leq 2.5Y_{0.4}+0.01 b$.
Indeed, in the close case,
\[
\widehat X+\widehat Y
\leq 3a+2.5\cdot 10a+0.02 b
=28a+0.02 b
<\frac{b}{10},
\]
whereas in the far case,
\[
\widehat X+\widehat Y
\geq 0.9(X+Y_{0.5})-0.02 b
\geq 0.9\cdot\frac{b}{4}-0.02 b
>\frac{b}{10}.
\]
We start with showing
$
0.9X-0.01b\leq \widehat X\leq 3X+0.01b
$
with probability at least $1-\tau$. For every block $k$, let
$
x_k:=|A_k\setminus B_k|=a-|A_k\cap B_k|,
$
so that $X=\sum_k x_k$. For every inner loop $i$, define
$
H_i:=|\{k:x_k\geq 2^i\}| \text{ and }
K_i:=|\{k:x_k>0.8\cdot 2^i\}|.
$
According to Algorithm~\ref{alg:gapulamtest}, we call Theorem~\ref{cor:legall}
$(n/a)L_a$ times, each with failure
probability at most $\tau a/(2nL_a)$. By union bound, all inner loops succeed simultaneously with probability at least $1-\tau/2$. Moreover, by Theorem~\ref{cor:legall}, we know $|\widetilde c_k-c_k|\leq 0.1\cdot 2^i$ $\forall k\in[n/a]$. 
Thus, if $x_k\geq 2^i$, then
\[
a-\widetilde c_k
\geq a-c_k-0.1\cdot 2^i
=x_k-0.1\cdot 2^i
\geq 0.9\cdot 2^i,
\]
and hence the block is counted by $\widehat X_i$. On the other hand, if a block
is counted by $\widehat X_i$, then
\[
x_k=a-c_k
\geq a-\widetilde c_k-0.1\cdot 2^i
>0.8\cdot 2^i.
\]

Let $Z_i^-$ and $Z_i^+$ denote the numbers of sampled blocks among the
$H_i$ and $K_i$ blocks, respectively, so that
$
Z_i^-\leq \widehat X_i\leq Z_i^+,
Z_i^-\sim\operatorname{Bin}(H_i,q_i),\text{ and }
Z_i^+\sim\operatorname{Bin}(K_i,q_i).
$
By Bernstein's inequality (Lemma~\ref{thm: Bernstein}) for each $i$, with success probability $\geq 1-\tau/(4L_a)$,
$
\frac{Z_i^-}{q_i}\geq
0.9H_i-\frac{0.01b}{L_a2^i} \text{ and }
\frac{Z_i^+}{q_i}\leq
1.1K_i+\frac{0.01b}{L_a2^i}.
$
Therefore, by the union bound we know that, with probability at least $1-\tau/2$, they hold
simultaneously for every $i$.
 Combining the above two, the following upper and lower bounds hold with success probability $\geq 1-\tau$: For the lower bound,
\[
\begin{aligned}
\widehat X=\sum_{i=0}^{L_a-1}\frac{2^i}{q_i}\widehat X_i
&\geq
0.9\sum_{i=0}^{L_a-1}2^iH_i-0.01b=
0.9\sum_k\sum_{i:2^i\leq x_k}2^i-0.01b
\geq
0.9\sum_kx_k-0.01b=
0.9X-0.01b.
\end{aligned}
\]
Similarly,
\[
\begin{aligned}
\widehat X
\leq
1.1\sum_{i=0}^{L_a-1}2^iK_i+0.01b=
1.1\sum_k\sum_{i:0.8\cdot2^i<x_k}2^i+0.01b
<
1.1\cdot 2.5\sum_kx_k+0.01b
\leq
3X+0.01b,
\end{aligned}
\]
which proves the $\widehat X$ part.

We next show
$0.9Y_{0.5}-0.01b
\leq
\widehat Y
\leq
2.5Y_{0.4}+0.01b.$
For every block $k$ and $\delta\in\{0.4,0.5\}$, let
$P_{k,\delta}$ be the number of pairs $(u,v)$ contributing to
$Y_\delta$, and therefore,
$
Y_\delta=\sum_k P_{k,\delta}.
$
Condition on the event that all calls to algorithms in Theorem~\ref{thm:findingcollisions} and Corollary~\ref{cor:ytest} succeed. For a fixed block $k$ and level $i$, let $D_{k,i}^-$ be the number of sampled pairs contributing to $Y_{0.5}$, and let $D_{k,i}^+$ be the number of sampled pairs contributing to $Y_{0.4}$. By Corollary~\ref{cor:ytest} and the geometric-scale argument we know that
$D_{k,i}^-\leq D\leq D_{k,i}^+.$
Moreover, since the strings are non-repetitive, distinct collision pairs
use distinct positions, and hence
$D_{k,i}^-\sim\operatorname{Bin}(P_{k,0.5},r_i^2)$ and $D_{k,i}^+\sim\operatorname{Bin}(P_{k,0.4},r_i^2)$.

By Bernstein's inequality and the choice of $r_i$, with high probability
simultaneously for every block and every level,
$P_{k,0.5}\geq2^i $ hence $
\frac{D}{r_i^2}>0.9\cdot2^i,$
whereas
$\frac{D}{r_i^2}>0.9\cdot2^i$ hence $P_{k,0.4}>0.88\cdot2^i.$
For every level $i$, define
$P_i:=|\{k:P_{k,0.5}\geq2^i\}|,
\qquad
Q_i:=|\{k:P_{k,0.4}>0.88\cdot2^i\}|.$
Therefore, the number of blocks passing the level-$i$ test is between
$P_i$ and $Q_i$. Applying the same Bernstein and block-sampling argument
as for $\widehat X$, with high probability simultaneously for every $i$,
$0.9P_i-\frac{0.01b}{L_a2^i}
\leq
\frac{\widehat Y_i}{q_i}
\leq
1.1Q_i+\frac{0.01b}{L_a2^i}.$
Consequently,
$$\widehat Y \geq 0.9\sum_i2^iP_i-0.01b \geq 0.9Y_{0.5}-0.01b,$$ and similarly we have $$\widehat Y \leq 1.1\sum_i2^iQ_i+0.01b< 1.1\cdot\frac{2}{0.88}Y_{0.4}+0.01b =2.5Y_{0.4}+0.01b,$$ which proves the $\widehat Y$ part.

\end{proofof}

\subsection{Proof of the correctness of Algorithm~\ref{alg:quantumpartialalign}}
\label{sec:correctPartialAlign}

Fix an LCS of $A,B$, and let $S_A,S_B$ be the sets of its positions
in $A,B$, respectively. Conditioned on the previous anchor
$(a_{i-1},b_{i-1})$ belonging to this LCS, define
$$
f_i:=|[a_{i-1},i\beta R]\setminus S_A|
+|[b_{i-1},m_{i-1}+i\beta R]\setminus S_B|,
$$
where $m_{i-1}=b_{i-1}-a_{i-1}$. Thus, $f_i$ is the number of
positions outside the fixed LCS from the previous anchor to the end
of the $i$-th block in $A$ and its corresponding region in $B$.
Moreover, each position outside the fixed LCS is contained in at most
two of the intervals defining the $f_i$'s. Since
$|[n]\setminus S_A|\leq R$ and $|[n]\setminus S_B|\leq R$, we have
$$
\sum_i f_i
\leq
2|[n]\setminus S_A|+2|[n]\setminus S_B|
\leq4R.
$$

We recall the following two claims from the proof of Lemma~3.1
in~\cite{AN10UlamTester}.

\begin{claim}[\cite{AN10UlamTester}, Claim~3.1]
\label{clm:an10-bad-anchor}
For any fixed $j\leq\log(4R)$, in the original sampling procedure of
\cite{AN10UlamTester}, the probability that a position outside $S_A$
appears among the sampled collisions is at~most
$$
\frac{f_i\gamma^2}{\beta R}.
$$
\end{claim}

\begin{claim}[\cite{AN10UlamTester}, Claim~3.2]
\label{clm:an10-stopping}
For $f_i>0$, the $j$-loop of the original procedure stops for some
$j$ satisfying $2^j\leq2f_i$ with probability at least
$$
1-\frac{21f_i}{(\beta-8)R}.
$$
If $f_i=0$, the $j$-loop stops at $j=4$.
\end{claim}

\begin{proof}
We follow the proof of Lemma~3.1 in~\cite{AN10UlamTester}. Throughout
the proof, we ignore the amplification needed to make all quantum
subroutines succeed simultaneously, as it only introduces
polylogarithmic factors in the running time.

We prove by induction on $i$ that every returned anchor $(a_i,b_i)$
belongs to the fixed LCS. Suppose $(a_{i-1},b_{i-1})$ belongs to the
fixed LCS. Fix a scale $z=2^j$, and let $M_{i,j}$ be the number of
collisions between the two corresponding full windows and
$M_{i,j}^{\rm bad}$ the number of such collisions not belonging to the
fixed LCS. The random-window calculation in the proof of
Claim~\ref{clm:an10-bad-anchor} gives
$$
\mathbb E_p[M_{i,j}^{\rm bad}]
\leq
\frac{2zf_i}{(\beta-8)R}.
$$
Whenever our density test passes, the $(0.02z)$-additive guarantee
implies $M_{i,j}\geq0.83z$. By the uniform-output guarantee of Theorem~\ref{thm:quantumpairfinding}, every collision is returned with the same probability, and hence, the probability of returning a bad anchor at this scale is at most
$$
\mathbb E_p\left[
\frac{M_{i,j}^{\rm bad}}{M_{i,j}}
\mathbf 1[M_{i,j}\geq0.83z]
\right]
\leq
\frac{2.5f_i}{(\beta-8)R}.
$$
Taking a union bound over all $O(\log R)$ scales, the probability that
the $i$-th returned anchor is bad is
$\widetilde O(f_i/(\beta R))$.

It remains to show that the $j$-loop stops. In the proof of
Claim~\ref{clm:an10-stopping}, for the smallest $j$ satisfying
$z=2^j\geq f_i$, with probability at least
$1-20f_i/((\beta-8)R)$, the $A$-window contains at least $0.9z$
positions from the fixed LCS and all their matching positions belong
to the corresponding $B$-window. Hence $M_{i,j}\geq0.9z$, so its
$(0.02z)$-additive estimate is at least $0.88z>0.85z$. Therefore,
the density test passes and the quantum pair-finding procedure returns
a collision. Thus, the $j$-loop stops by
$z=O(\max\{1,f_i\})$.

Therefore, conditioned on all previous anchors being correct, the
probability that the $i$-th step fails or returns a bad anchor is
$\widetilde O(f_i/(\beta R))$. Since $\sum_i f_i\leq4R$, a union bound
over all $i$ gives an overall failure probability
$\widetilde O(1/\beta)$, which is small for sufficiently large
$\beta=\polylog(n)$. Hence, with high probability, every returned
anchor belongs to the fixed LCS. Since $A,B$ are non-repetitive,
$a_i\in S_A$ and $A[a_i]=B[b_i]$ imply $b_i\in S_B$. Thus the
returned matching positions form a subsequence of the fixed LCS.

Finally, our procedure uses the same windows as
\textsc{PartialAlign} in~\cite{AN10UlamTester}. Therefore, the same
interval argument as in Lemma~3.1 gives
$$
|a_{i+1}-a_i|,\ |b_{i+1}-b_i|\leq2\beta R.
$$
Hence the returned anchors partition $A$ and $B$ into corresponding
substrings of length at most $2\beta R$, which completes the proof.
\end{proof}
\subsection{Proof of Claim~\ref{lem:concentration_collision}}
\label{app:claim5.6}
    We first remove the conditioning on $I$. Since $t=\ceil{\mu}$, if $t\geq2$, then
    \[
        \lambda\rho
        \geq
        \mu-\frac{\sqrt{t}}{10}
        >
        t-1-\frac{\sqrt{t}}{10}
        =
        \Theta(t).
    \]
    Since $\rho<1/16$, this also implies that $t-1\leq\lambda$.

    If $I=\emptyset$, no reduction is necessary. Otherwise, conditioned on $I\subseteq S$, the set $S\setminus I$ is a uniformly random size-$(s-2)$ subset of the concatenation of the two strings obtained by deleting the two entries indexed by $I$. These strings are non-repetitive, have length $n-1$, and have the same $\lambda$ collision pairs. If $s\leq3$, then $\mu<1$, so $t=1$, and the conclusion is immediate because $S\setminus I$ contains at most one index. We may therefore assume that $s\geq4$. Let
    \[
        \rho'
        =
        \frac{(s-2)(s-3)}{(2n-2)(2n-3)}.
    \]
    Since $\lambda\leq n$ and $s\leq n/2$, we have
    \[
        \rho'=\Theta(\rho)
        \qquad\text{and}\qquad
        \lambda|\rho'-\rho|=O(1).
    \]
    Consequently,
    \[
        |\lambda\rho'-\mu|
        =
        O(\sqrt{t}),
    \]
    and, when $t\geq2$, we also have $\lambda\rho'=\Theta(t)$ and $t-1\leq\lambda$. Thus, the desired conditional probability is precisely the corresponding unconditional probability after replacing $n,s,\rho,\mu_\lambda$ by $n-1,s-2,\rho',\lambda\rho'$, respectively. For notational simplicity, we continue to use $n,s,\rho,\mu_\lambda$ for these quantities.
    
    Let $\theta:=s/(2n)$. Consider an experiment in which we choose a subset $R$ of $[2n]$ by including every $i\in[2n]$ independently with probability $\theta$.     Let $N_\lambda(R)$ be the number of collision pairs between $A$ and $B$ fully contained in $R$. Then,
    \begin{align*}
        \Pr_S[N_\lambda(S) = t-1] = \frac{\Pr_R[|R| = s, N_\lambda(R) = t-1]}{\Pr_R[|R| = s]}
    \end{align*}

    We will separately prove that $\Pr_R[|R| = s, N_\lambda(R) = t-1] = \Omega(1/\sqrt{st})$ and $\Pr_R[|R| = s] = O(1/\sqrt{s})$. For the latter, notice that $|R|$ is distributed according to the binomial distribution on $2n$ elements and probability $s/2n$. 
    Thus, $\E[|R|]=s$ and $\text{Var}(|R|)    = s(1-\theta)= \Theta(s).$
By Lemma~\ref{lem:local-limit},
\[
    \Pr_R[|R|=s]
    =
    \Theta\left(\frac{1}{\sqrt s}\right)
\]
when $s$ is sufficiently large. For bounded $s$, the required
$O(1/\sqrt s)$ upper bound is immediate after adjusting the implicit constant. 
    
    We can write $\Pr_R[|R| = s, N_\lambda(R) = t-1]$ as $\Pr_R[N_\lambda(R) = t-1] \cdot \Pr_R[|R| = s \mid N_\lambda(R) = t-1]$.
    For the former, notice that $N_\lambda(R)$ is distributed according to the binomial distribution on $\lambda$ elements and probability $(s/2n)^2$. 
    Thus, $\nu:=\E[N_\lambda(R)] =    \lambda\left(\frac{s}{2n}\right)^2$ and  $\text{Var}(N_\lambda(R)) =    \lambda\left(\frac{s}{2n}\right)^2    \left(1-\left(\frac{s}{2n}\right)^2\right) = \Theta(\nu).$  

    Observe that
    \begin{align*}
        |t-1-\nu|
        &\leq
        |t-1-\mu|
        +
        |\mu-\mu_\lambda|
        +
        |\mu_\lambda-\nu|
        =
        O(\sqrt{t}),
    \end{align*}
    where $|\mu_\lambda-\nu|=O(1)$ since $\lambda\leq n$ and $s\leq n/2$.

    First suppose that $t=1$. Then $\nu=O(1)$, and hence
    \begin{align*}
        \Pr_R[N_\lambda(R)=t-1]
        &=
        \Pr_R[N_\lambda(R)=0]\\
        &=
        \left(1-\left(\frac{s}{2n}\right)^2\right)^\lambda
        =
        \Omega(1)
        =
        \Omega\left(\frac{1}{\sqrt{t}}\right).
    \end{align*}

    Now suppose that $t\geq2$. By the reduction above, $t-1\leq\lambda$ and $\mu_\lambda=\Theta(t)$. If $\nu=O(1)$, then $t=O(1)$ and $\nu=\Theta(1)$. Since $t-1\leq\lambda$, the binomial probability mass function gives
    \[
        \Pr_R[N_\lambda(R)=t-1]
        =
        \Omega(1)
        =
        \Omega\left(\frac{1}{\sqrt{t}}\right).
    \]
    Otherwise, $\nu=\Theta(t)$ and $|t-1-\nu|=O(\sqrt{\nu})$. Therefore, by Lemma~\ref{lem:local-limit},
    \[
        \Pr_R[N_\lambda(R)=t-1]
        =
        \Omega\left(\frac{1}{\sqrt{\nu}}\right)
        =
        \Omega\left(\frac{1}{\sqrt{t}}\right).
    \]

   It remains to show that   $\Pr_R[|R|=s\mid N_\lambda(R)=t-1] =    \Omega\left(\frac{1}{\sqrt{s}}\right).$ For each of the remaining $\lambda-t+1$ collisions, conditioned on not being fully sampled, the probability that it contributes exactly one element to $R$ is   $\phi :=    \frac{2\theta(1-\theta)}{1-\theta^2}    =    \frac{2\theta}{1+\theta}.$ Thus, conditioned on $N_\lambda(R)=t-1$,$    |R|=2(t-1)+U+V,$ where $U\sim\operatorname{Bin}(\lambda-t+1,\phi)$ and $    V\sim\operatorname{Bin}(2n-2\lambda,\theta)$ are independent. Let $W:=U+V$. We have
\[
    \Pr_R[|R|=s\mid N_\lambda(R)=t-1]
    =
    \Pr_R[W=s-2(t-1)].
\]
Notice that
\[
\begin{aligned}
    \E[W]-(s-2(t-1))
    &=
    (\lambda-t+1)\phi
    +(2n-2\lambda)\theta
    -(s-2(t-1))\\
    &=
    \frac{2}{1+\theta}(t-1-\nu)
    =
    O(\sqrt{s}),
\end{aligned}
\]
where the last equality follows from $|t-1-\nu|=O(\sqrt{t})$ and $t=O(s)$. Furthermore,
\[
\begin{aligned}
    \text{Var}(W)
    &=
    \text{Var}(U)+\text{Var}(V)\\
    &=
    (\lambda-t+1)\phi(1-\phi)
    +(2n-2\lambda)\theta(1-\theta)\\
    &=
    \Theta(s).
\end{aligned}
\]
Indeed, $\theta,\phi=\Theta(s/n)$, and either $\lambda-t+1=\Omega(n)$ or $2n-2\lambda=\Omega(n)$. 
If $s$ is bounded, then $t$ is also bounded, and the same lower bound follows directly from the displayed binomial distributions of $U$ and $V$, after adjusting the implicit constant. We may therefore assume that $s$ is sufficiently large.
Since $W$ is a sum of independent Bernoulli random variables, Lemma~\ref{lem:local-limit} gives
\[
    \Pr_R[W=s-2(t-1)]
    =
    \Omega\left(\frac{1}{\sqrt{s}}\right).
\]                                                                                         Hence,
\[
\begin{aligned}
    \Pr_R[|R|=s,N_\lambda(R)=t-1]=
    \Pr_R[N_\lambda(R)=t-1]\,
    \Pr_R[|R|=s\mid N_\lambda(R)=t-1]=
    \Omega\left(\frac{1}{\sqrt{st}}\right).
\end{aligned}
\]
Since $R$ conditioned on $|R|=s$ is a uniformly random size-$s$ subset of $[2n]$, and $\Pr_R[|R|=s]
    =
    O\left(\frac{1}{\sqrt{s}}\right),$
we conclude that
\[
\begin{aligned}
    \Pr_S[N_\lambda(S)=t-1]
    =
    \frac{\Pr_R[|R|=s,N_\lambda(R)=t-1]}
         {\Pr_R[|R|=s]}=
    \Omega\left(\frac{1}{\sqrt{t}}\right),
\end{aligned}
\]
which proves the claim.     
\subsection{Multiple-collision search for a fixed matching}
\label{app:fixed-matching-collisions}

In this subsection we prove the collision-enumeration result used in
Theorem~\ref{thm:findingcollisions}.  The multiple-collision algorithm
of~\cite{Bonnetain2025multiplecollision} is stated for a random function.
Here the input is fixed.  The only property that we use is that the collision
pairs are disjoint.  We give the full argument, including the quantum states,
the operations applied to them, and the error accumulated over repeated
measurements.

Let $U$ be a set of size $N$, and let $F:U\to\Sigma$ be given by the query
unitary
\[
  O_F\ket{x}\ket{w}=\ket{x}\ket{w\oplus F(x)}.
\]
Suppose that exactly $q$ unordered pairs $\{x,y\}\subseteq U$ satisfy
$F(x)=F(y)$ and that these pairs are disjoint.  Let $\mathcal M$ denote this
matching.  For $S\subseteq U$, define
\[
  Z_{\mathcal M}(S)
  :=\bigl|\{e\in\mathcal M:e\subseteq S\}\bigr|.
\]
For a set $I$ of integers, let
\[
  \mathcal V_I(N,q,r)
  :=\{S\subseteq U:|S|=r,\ Z_{\mathcal M}(S)\in I\},
\]
and let
\[
  \ket{\mathcal V_I(N,q,r)}
  :=\frac{1}{\sqrt{|\mathcal V_I(N,q,r)|}}
    \sum_{S\in\mathcal V_I(N,q,r)}\ket{S,D(S)},
\]
where
\[
  D(S):=((x,F(x)):x\in S),
\]
and the pairs in $D(S)$ are stored in a fixed order.  We omit the parameters
when they are clear.  We also write
\[
  \ket{\mathcal V}
  :=\frac{1}{\sqrt{\binom Nr}}
    \sum_{|S|=r}\ket{S,D(S)}.
\]
The initial state $\ket{\mathcal V}$ costs $r$ queries: first prepare a
uniform $r$-subset and then query the values at its positions.

The sizes of all the sets above depend only on $N,q,r$, and not on the
locations of the pairs.  More precisely, if $v_z$ is the number of
$r$-subsets containing exactly $z$ pairs, then
\begin{align}
  v_z
  =\binom qz
    \sum_{h=0}^{q-z}
      \binom{q-z}{h}2^h
      \binom{N-2q}{r-2z-h},
  \label{eq:fixed-matching-count}
\end{align}
where a binomial coefficient is zero when its lower argument is outside its
usual range.  Indeed, after choosing the $z$ complete pairs, we choose $h$
other pairs that contribute one position, choose one of the two positions
from each such pair, and then choose the remaining positions outside the
matching.  Thus $|\mathcal V_I|=\sum_{z\in I}v_z$ can be computed without
queries once $N,q,r$ are known.

We first prove the probability estimate used throughout the algorithm.

\begin{lemma}[Collision statistics for a fixed matching]
\label{lem:fixed-matching-statistics}
Let $S$ be a uniformly random $r$-subset of $U$, where $2\leq r\leq N/4$, and
let
\[
  \mu:=\mathbb E[Z_{\mathcal M}(S)]
      =q\frac{(r)_2}{(N)_2}.
\]
Then
\[
  \operatorname{Var}(Z_{\mathcal M}(S))\leq\mu.
\]
Moreover, for every constant $C>0$, if $\mu$ is sufficiently large and
$|z-\mu|\leq C\sqrt\mu$, then
\[
  \Pr[Z_{\mathcal M}(S)=z]
  =\Theta_C(1/\sqrt\mu).
\]
Consequently, every interval containing $\Theta(\sqrt\mu)$ consecutive
integers, all within distance $O(\sqrt\mu)$ of $\mu$, has probability
bounded above and below by positive constants.  The same conclusions hold
after any collection of already reported pairs and their positions is
deleted, with $N,q,r$ replaced by their new values.
\end{lemma}

\begin{proof}
For $e\in\mathcal M$, let $I_e$ be the indicator of $e\subseteq S$.  Then
\[
  Z_{\mathcal M}(S)=\sum_{e\in\mathcal M}I_e,
  \qquad
  \mathbb E[I_e]=\frac{(r)_2}{(N)_2}.
\]
If $e,e'\in\mathcal M$ are distinct, their four positions are distinct and
\[
  \mathbb E[I_eI_{e'}]
  =\frac{(r)_4}{(N)_4}
  \leq
  \left(\frac{(r)_2}{(N)_2}\right)^2.
\]
Thus the indicators are pairwise negatively correlated, and
\[
  \operatorname{Var}(Z_{\mathcal M}(S))
  \leq\sum_{e\in\mathcal M}\operatorname{Var}(I_e)
  \leq\mu.
\]

We next prove the point-probability estimate.  Set $p=r/N$, and form a random
set $R\subseteq U$ by including every position independently with
probability $p$.  Conditioned on $|R|=r$, the set $R$ is a uniformly random
$r$-subset.  Let
\[
  X:=Z_{\mathcal M}(R),
  \qquad
  \nu:=\mathbb E[X]=qp^2.
\]
Since the pairs are disjoint, $X\sim\operatorname{Bin}(q,p^2)$.  Also,
\[
  |\nu-\mu|
  =q\frac{r(N-r)}{N^2(N-1)}=O(1).
\]
The variance of $X$ is $\Theta(\mu)$.  Lemma~\ref{lem:local-limit} therefore
gives, uniformly for $|z-\mu|\leq C\sqrt\mu$,
\begin{align}
  \Pr[X=z]=\Theta_C(1/\sqrt\mu).
  \label{eq:matching-binomial-local}
\end{align}

Condition on $X=z$.  The $z$ complete pairs contribute $2z$ positions to
$R$.  Each of the other $q-z$ pairs contributes one position with
conditional probability
\[
  \phi=\frac{2p(1-p)}{1-p^2}=\frac{2p}{1+p},
\]
and the $N-2q$ positions outside the matching are still selected
independently with probability $p$.  Hence
\[
  |R|\mid(X=z)\ \stackrel{d}{=}\ 2z+Y_1+Y_2,
\]
where
\[
  Y_1\sim\operatorname{Bin}(q-z,\phi),
  \qquad
  Y_2\sim\operatorname{Bin}(N-2q,p)
\]
are independent.  A direct calculation gives
\[
  \mathbb E[|R|\mid X=z]-r
  =\frac{2(z-qp^2)}{1+p}=O(\sqrt\mu)=O(\sqrt r).
\]
Furthermore,
\[
  \operatorname{Var}(|R|\mid X=z)
  =(q-z)\phi(1-\phi)+(N-2q)p(1-p)
  =\Theta(r).
\]
For the last equality, if $q\leq N/4$, the second term is $\Theta(r)$.  If
$q>N/4$, then $z=O(\mu)=O(qr^2/N^2)$, and hence $q-z=\Theta(N)$, while
$\phi=\Theta(r/N)$; the first term is therefore $\Theta(r)$.  Applying
Lemma~\ref{lem:local-limit} once more gives
\begin{align}
  \Pr[|R|=r\mid X=z]=\Theta(1/\sqrt r).
  \label{eq:matching-conditional-size}
\end{align}
Also, $|R|\sim\operatorname{Bin}(N,p)$ with mean $r$, and hence
\begin{align}
  \Pr[|R|=r]=\Theta(1/\sqrt r).
  \label{eq:matching-unconditional-size}
\end{align}
Bayes' rule and
\eqref{eq:matching-binomial-local}--\eqref{eq:matching-unconditional-size}
now give
\[
  \Pr[Z_{\mathcal M}(S)=z]
  =\Pr[X=z\mid |R|=r]
  =\Theta_C(1/\sqrt\mu).
\]
Summing over the integers in the stated interval proves the interval bound.
Deleting reported pairs and their positions leaves another matching, so the
same argument applies to every later instance.
\end{proof}

We next give the quantum operations used to change one uniform state into
another.  This is the part that must be checked independently of the random
function setting of~\cite{Bonnetain2025multiplecollision}.

\begin{lemma}[Reflection about the uniform state]
\label{lem:fixed-matching-reflection}
For every $\delta>0$, there is a unitary $\widetilde R_{\mathcal V}$ with a
work register satisfying the following property.  For every normalized
state $\ket{\psi}$ in the span of the states $\ket{S,D(S)}$,
\[
  \left\|
  \widetilde R_{\mathcal V}(\ket{\psi}\ket{0})
  -(R_{\mathcal V}\ket{\psi})\ket{0}
  \right\|
  \leq\delta,
\]
where
\[
  R_{\mathcal V}:=2\ket{\mathcal V}\!\bra{\mathcal V}-I.
\]
It uses
\[
  O(\sqrt r\log(1/\delta))
\]
queries to $F$.  The same guarantee holds for the controlled unitary.
\end{lemma}

\begin{proof}
Consider the random walk on the Johnson graph $J(U,r)$.  From a subset $S$,
it chooses $x\in S$ and $y\in U\setminus S$ uniformly and moves to
$S-x+y$.  Define
\[
  \ket{p_S}
  :=\frac{1}{\sqrt{r(N-r)}}
    \sum_{x\in S}\sum_{y\notin S}\ket{x,y}.
\]
The following two operations define the quantum walk.  First, without any
query to $F$, prepare the second register by
\[
  A\ket{S,D(S)}\ket0
  =\ket{S,D(S)}\ket{p_S}.
\]
Second, update the subset and its database by
\[
  U_{\rm upd}\ket{S,D(S)}\ket{x,y}
  =\ket{S-x+y,D(S-x+y)}\ket{y,x}.
\]
To implement this map, use one query at $x$ to change the stored value
$F(x)$ to zero, replace $x$ by $y$ in the ordered database, and use one
query at $y$ to write $F(y)$.  The last two registers are then exchanged.
The inverse performs these steps in reverse order.  Thus the update and its
inverse use two queries each.  Both maps are extended to unitaries on the
full space.

Let $\mathcal A$ be the span of the states
$\ket{S,D(S)}\ket{p_S}$, let $\mathcal B=U_{\rm upd}\mathcal A$, and let
\[
  W=(2\Pi_{\mathcal B}-I)(2\Pi_{\mathcal A}-I).
\]
The reflection about $\mathcal A$ is implemented by $A$, a phase change on
the zero register, and $A^\dagger$.  The reflection about $\mathcal B$ is
obtained by conjugating this operation with $U_{\rm upd}$.  Thus one use of
$W$ costs $O(1)$ queries.

The uniform distribution is the unique stationary distribution of the
Johnson walk.  The corresponding state is
\[
  \ket{\overline{\mathcal V}}
  :=\frac{1}{\sqrt{\binom Nr}}
    \sum_{|S|=r}\ket{S,D(S)}\ket{p_S}.
\]
This is the unique zero-phase state of $W$ in the subspace reached from
$\mathcal A$.  The Johnson walk has spectral gap $\Theta(1/r)$.  The
two-reflection construction maps a transition-matrix eigenvalue $\lambda$
to eigenphases whose distance from zero is
$\Theta(\sqrt{1-\lambda})$.  Hence every nonzero eigenphase reached from
$\mathcal A$ has absolute value $\Omega(1/\sqrt r)$.

Perform phase estimation for $W$ with precision $\Theta(1/\sqrt r)$.
Repeat the test $O(\log(1/\delta))$ times, apply phase $+1$ when all tests
return zero and phase $-1$ otherwise, and uncompute every test.  The
zero-phase state is unchanged exactly.  On every other eigenvector, the
probability that all tests return zero is $O(\delta^2)$, so the resulting
state differs from its negative by norm at most $\delta$.  The same bound
holds for every superposition of eigenvectors.  Conjugating by $A$ gives
the claimed operation on the subset register, including the stated bound
on the work register.  Controlling each operation gives the controlled
version.  This is the standard Johnson-graph implementation of the MNRS
reflection~\cite{MNRS11,BCSS23}.
\end{proof}

For an integer set $I$, let $\Pi_I$ denote the projector onto the span of
$\{\ket{S,D(S)}:Z_{\mathcal M}(S)\in I\}$.  Since all the queried values are
stored in $D(S)$, a reversible computation can find
$Z_{\mathcal M}(S)$ and implement
\[
  R_I:=I-2\Pi_I
\]
without an additional query to $F$.

\begin{lemma}[Changing the interval]
\label{lem:fixed-matching-state-conversion}
Let $I,J$ be two integer sets with
\[
  p_I:=\frac{|\mathcal V_I|}{\binom Nr}>0,
  \qquad
  p_J:=\frac{|\mathcal V_J|}{\binom Nr}>0.
\]
Given $\ket{\mathcal V_I}$, one can prepare a state within distance
$\epsilon$ of $\ket{\mathcal V_J}$ using
\[
  \widetilde O\left(
    \sqrt r\left(\frac{1}{\sqrt{p_I}}+
                  \frac{1}{\sqrt{p_J}}\right)
    \log^2(1/\epsilon)
  \right)
\]
queries.  The values $p_I,p_J$ required by the procedure are computable from
\eqref{eq:fixed-matching-count}.
\end{lemma}

\begin{proof}
The normalized projection of $\ket{\mathcal V}$ onto the range of $\Pi_I$
is exactly $\ket{\mathcal V_I}$.  Use the coherent form of fixed-point
amplitude amplification from Theorem~\ref{thm:Fixed_AA}.  In that circuit,
each reflection about the starting state is
$U R_{\ket0}U^\dagger=R_{\mathcal V}$, where
$U\ket0=\ket{\mathcal V}$.  We implement it directly by
Lemma~\ref{lem:fixed-matching-reflection}.  For completeness, the phase
change
\[
  S_{\mathcal V}(\theta)
  :=I+(e^{i\theta}-1)\ket{\mathcal V}\!\bra{\mathcal V}
\]
required by fixed-point amplitude amplification is obtained as follows.
Initialize one qubit to $\ket0$, apply a Hadamard gate, apply
$R_{\mathcal V}$ controlled by that qubit, and apply another Hadamard gate.
The qubit is now $\ket0$ on $\ket{\mathcal V}$ and $\ket1$ on its
orthogonal complement.  Apply phase $e^{i\theta}$ when the qubit is zero and
reverse the preceding operations.  This uses two controlled calls to
$R_{\mathcal V}$ and returns the qubit to zero.  The same construction with
$R_I$, applying the phase when the qubit is one, gives the phase changes
about the accepting subspace.  Thus the amplification circuit defines a
unitary $P_I$ such that
\[
  \|P_I(\ket{\mathcal V}\ket0)-\ket{\mathcal V_I}\ket0\|
  \leq\epsilon/10.
\]
The circuit uses $O(p_I^{-1/2}\log(1/\epsilon))$ such operations.  We
implement every
$R_{\mathcal V}$ by Lemma~\ref{lem:fixed-matching-reflection}, with error at
most $\epsilon$ divided by ten times the number of uses.  The sum of these
errors is at most $\epsilon/10$.  Therefore the stated query bound prepares
$\ket{\mathcal V_I}\ket0$ from $\ket{\mathcal V}\ket0$.  The same
construction gives $P_J$.

Starting from $\ket{\mathcal V_I}\ket0$, apply $P_I^\dagger$ and then
$P_J$.  Since $P_I$ is unitary, the first operation returns a state within
distance $O(\epsilon)$ of $\ket{\mathcal V}\ket0$, and the second produces
a state within distance $O(\epsilon)$ of
$\ket{\mathcal V_J}\ket0$.  Rescaling the error used in the two
constructions gives distance at most $\epsilon$.
\end{proof}

We now specify what happens when a collision pair is measured.  For
integers $1\leq a\leq b$, consider a basis state with
$z=Z_{\mathcal M}(S)\in[a,b]$.  From $D(S)$, list the $z$ collision pairs
in a fixed order as $e_1(S),\ldots,e_z(S)$.  For a collision pair
$e=\{x,y\}$, write $F(e)$ for the common value $F(x)=F(y)$.  The following
unitary uses no new input query:
\begin{align}
  \ket{S,D(S)}\ket0
  \longmapsto
  \ket{S,D(S)}\frac{1}{\sqrt b}
  \left(
    \sum_{i=1}^z\ket{e_i(S),F(e_i(S))}+
    \sum_{i=z+1}^b\ket{d_i}
  \right),
  \label{eq:fixed-matching-extraction-unitary}
\end{align}
where $d_1,\ldots,d_b$ are distinct tagged symbols in the same output
register.  It is implemented by a reversible computation of the ordered
list and the common values, followed by preparation of a uniform index in
$[b]$, and then by uncomputing the list.  Its action outside the stated
subspace can be chosen arbitrarily to make it a unitary.

\begin{lemma}[One extraction]
\label{lem:fixed-matching-extraction}
Suppose the current state is $\ket{\mathcal V_{[a,b]}}$, where
$1\leq a\leq b$, and set
\[
  p:=\frac{|\mathcal V_{[a,b]}|}{\binom Nr}.
\]
There is a procedure that reports one collision pair $e$, deletes its two
positions, and leaves the uniform state
\[
  \ket{\mathcal V_{[a-1,b-1]}(N-2,q-1,r-2)}.
\]
The ideal procedure is exact.  Its expected restoration cost is
\[
  \widetilde O\left(
    \sqrt r\,\frac{b-a}{a\sqrt p}
  \right)
\]
queries.
\end{lemma}

\begin{proof}
Apply the unitary in
\eqref{eq:fixed-matching-extraction-unitary} and measure its last register.
The probability of obtaining a collision pair is
\[
  \frac{1}{|\mathcal V_{[a,b]}|}
  \sum_{S\in\mathcal V_{[a,b]}}
  \frac{Z_{\mathcal M}(S)}{b}
  \geq\frac ab.
\]
Suppose the outcome is $(e,F(e))$.  Every subset containing $e$ had the same
amplitude $1/\sqrt b$ on this outcome.  Conditioned on this outcome, the
remaining state is therefore uniform over all $S\in\mathcal V_{[a,b]}$
that contain $e$.  The map
\[
  S\longmapsto S\setminus e
\]
is a bijection from these subsets to
\[
  \mathcal V_{[a-1,b-1]}(N-2,q-1,r-2).
\]
Keep $(e,F(e))$ in the classical output register.  Using this record, remove
the two entries belonging to $e$ from the ordered database by a reversible
computation: the inverse recovers both entries from $e$ and $F(e)$.  Thus no
value is erased and no new query is used.  This proves the claimed state
after a successful measurement.

It remains to handle the other outcomes.  If the outcome is $d_j$, then
$j>Z_{\mathcal M}(S)$, so the conditional state is exactly
\[
  \ket{\mathcal V_{[a,j-1]}}.
\]
Let
\[
  p_j:=\frac{|\mathcal V_{[a,j-1]}|}{\binom Nr}.
\]
This outcome has probability $p_j/(bp)$.  Outcomes for which $p_j=0$ never
occur.  For every other outcome, use
Lemma~\ref{lem:fixed-matching-state-conversion} to return from
$\ket{\mathcal V_{[a,j-1]}}$ to $\ket{\mathcal V_{[a,b]}}$, and repeat the
measurement.  Ignoring logarithmic factors, the expected restoration cost
of one attempt is at most
\begin{align*}
 &\sum_{j=a+1}^b
   \frac{p_j}{bp}\sqrt r
   \left(\frac{1}{\sqrt{p_j}}+
         \frac{1}{\sqrt p}\right)                                      \\
 &\qquad\leq
   2\sqrt r\,\frac{b-a}{b\sqrt p},
\end{align*}
where we used $p_j\leq p$.  Since an attempt succeeds with probability at
least $a/b$, the expected number of attempts is at most $b/a$.  This proves
the stated bound.
\end{proof}

The preceding lemma gives the exact state after every outcome in the ideal
procedure.  We next show that the intervals used in consecutive extractions
continue to have sufficient probability after positions are deleted.

\begin{lemma}[One enumeration round]
\label{lem:fixed-matching-round}
Suppose that $q\geq1$ is known and $\eta\in(0,1/2)$.  There is a quantum
algorithm that reports
\[
  K:=\left\lceil q/32\right\rceil
\]
distinct collision pairs with failure probability at most $\eta$ and with
query complexity
\[
  \widetilde O\left(
    N^{2/3}q^{1/3}\,\operatorname{polylog}(1/\eta)
  \right).
\]
\end{lemma}

\begin{proof}
For $N$ below a fixed constant, query every position.  Otherwise choose
\[
  r:=\min\left\{
       \left\lfloor\frac N4\right\rfloor,
       \left\lceil(N^2q)^{1/3}\right\rceil
     \right\},
  \qquad
  \mu:=q\frac{(r)_2}{(N)_2}.
\]
The rounding of $r$ changes only constant factors.  Notice that $r=\Omega(q)$
because $q\leq N/2$.  Since $K\leq q/32+1$, deleting the positions of $K$
pairs changes $N,r$ only by constant factors, and it changes $q$ only by a
constant factor unless $q=1$.  In the latter case the round ends after the
only collision pair is reported.

We first consider the case in which $\mu$ is at least a sufficiently large
constant.  Let
\[
  E:=\lfloor\mu\rfloor,
  \qquad
  T:=\lfloor c\sqrt\mu\rfloor,
\]
where $c>0$ is a sufficiently small constant.  By
Lemma~\ref{lem:fixed-matching-statistics}, the interval $[E,E+T]$ has
constant probability.  Prepare $\ket{\mathcal V}$ using $r$ queries and
then prepare $\ket{\mathcal V_{[E,E+T]}}$ using
Lemma~\ref{lem:fixed-matching-state-conversion}.

Starting from this state, apply
Lemma~\ref{lem:fixed-matching-extraction} repeatedly.  After $h\leq T$
successful extractions, the parameters are
\[
  N_h=N-2h,
  \qquad q_h=q-h,
  \qquad r_h=r-2h,
\]
and, before any restoration caused by a $d_j$ outcome, the state is exactly
\begin{align}
  \ket{\mathcal V_{[E-h,E+T-h]}(N_h,q_h,r_h)}.
  \label{eq:fixed-matching-state-after-h}
\end{align}
This follows by induction from the bijection in
Lemma~\ref{lem:fixed-matching-extraction}.

Let
\[
  \mu_h:=q_h\frac{(r_h)_2}{(N_h)_2}.
\]
For $h\leq T$, comparison of the factors in this expression gives
\begin{align*}
  |\mu_h-\mu|
  &\leq
  C\mu\left(\frac hq+\frac hr+\frac hN\right)
  =O(h).
\end{align*}
Here we used
$\mu/q=O((r/N)^2)$, $\mu/r=O(qr/N^2)$, and
$\mu/N=O(qr^2/N^3)$.  Hence every integer in
$[E-h,E+T-h]$ is within $O(\sqrt\mu)$ of $\mu_h$.  The interval has
$\Theta(\sqrt\mu)$ integers, so
Lemma~\ref{lem:fixed-matching-statistics} shows that the state in
\eqref{eq:fixed-matching-state-after-h} always corresponds to a constant
fraction of the current $r_h$-subsets.  This proves the required interval
condition after every deletion, rather than only at the beginning of the
round.

Throughout these $T$ extractions we have $a=\Theta(E)$, $b-a=T$, and
$p=\Theta(1)$ in Lemma~\ref{lem:fixed-matching-extraction}.  Their total
expected restoration cost is therefore
\[
  \widetilde O\left(
    T\sqrt r\frac{T}{E}
  \right)
  =\widetilde O(\sqrt r).
\]
After these $T$ pairs have been reported, recompute the current mean and the
corresponding values of $E,T$.  The current interval and the new interval
$[E,E+T]$ both have constant probability, so
Lemma~\ref{lem:fixed-matching-state-conversion} prepares the state for the
next group of extractions using $\widetilde O(\sqrt r)$ queries.  Continue
until $K$ pairs have been reported.  Since the parameters change only by
constant factors during the round, the number of groups is
$O(1+K/\sqrt\mu)$.  The total expected cost in this case is
\begin{align}
  \widetilde O\left(
    r+\frac{q}{\sqrt\mu}\sqrt r
  \right)
  =\widetilde O\left(
    r+\frac{N\sqrt q}{\sqrt r}
  \right).
  \label{eq:fixed-matching-round-cost}
\end{align}

We now consider the case in which $\mu$ is smaller than the constant above.
Fix a constant $L$ larger than four times that threshold, and define
\[
  \mathcal B:=\mathcal V_{[1,L]}.
\]
The second-moment bound gives
\[
  \Pr[Z_{\mathcal M}>0]
  \geq
  \frac{\mathbb E[Z_{\mathcal M}]^2}
       {\mathbb E[Z_{\mathcal M}^2]}
  \geq\frac{\mu}{1+\mu}.
\]
On the other hand, Markov's inequality gives
$\Pr[Z_{\mathcal M}>L]\leq\mu/(L+1)$.  It follows, by the choice of $L$,
that
\begin{align}
  \frac{|\mathcal B|}{\binom Nr}=\Theta(\mu).
  \label{eq:fixed-matching-small-probability}
\end{align}
Prepare $\ket{\mathcal B}$ from $\ket{\mathcal V}$ using
Lemma~\ref{lem:fixed-matching-state-conversion}.  Apply
Lemma~\ref{lem:fixed-matching-extraction} with $a=1$ and $b=L$.  Its expected
cost, including restoration after a $d_j$ outcome, is
\[
  \widetilde O(\sqrt r/\sqrt\mu).
\]
After a collision pair is measured, the state is the uniform state on the
interval $[0,L-1]$ for the new parameters.  The new mean is no larger than
$\mu$, since $q$ decreases and each of the two ratios in
$(r)_2/(N)_2$ decreases when $r,N$ are both reduced by two.  Markov's
inequality therefore shows that $[0,L-1]$ has probability at least a
positive constant.  Unless $K$ pairs have already been reported, convert
this state to the new interval $[1,L]$ and repeat.  During the round the new
values of $N,q,r$ remain within constant factors of their initial values,
so the new mean remains within a constant factor of $\mu$.  Repeating the
second-moment and Markov bounds above shows that the new interval $[1,L]$
has probability $\Theta(\mu)$.  Thus the total expected cost is again
\[
  \widetilde O\left(
    r+q\frac{\sqrt r}{\sqrt\mu}
  \right)
  =\widetilde O\left(
    r+\frac{N\sqrt q}{\sqrt r}
  \right).
\]

It remains to bound the errors and obtain a fixed query bound.  First define
one trial of the round.  In the ideal procedure, allow at most
$c_1\log(N+1)$ attempts for each reported pair.  Since every attempt succeeds
with probability at least a positive constant, the probability that any
extraction exceeds this limit is at most $1/20$ for a sufficiently large
$c_1$.  There are then at most $M=O(N\log(N+1))$ calls to
Lemma~\ref{lem:fixed-matching-state-conversion}.  Implement each call with
error at most $1/(100M)$.  Replacing the ideal calls one at a time and using
the fact that unitary operations and measurements do not increase trace
distance shows that the output distribution of the trial changes by at most
$1/100$.

Let $Q$ be a fixed upper bound on the expected number of queries in the
ideal trial, including the logarithmic factors just introduced.  The
preceding calculations allow us to choose
\[
  Q=\widetilde O(N^{2/3}q^{1/3}).
\]
Abort the trial before an operation that would make its query count exceed
$20Q$.  By Markov's inequality, this happens with probability at most $1/20$
in the ideal trial.  If the trial finishes, verify that the output consists
of $K$ distinct pairs of current positions and query their endpoints to check
that each pair is a collision.  Since $q\leq N$, this verification uses
$O(K)=O(N^{2/3}q^{1/3})$ queries.
Consequently, one trial returns a verified list of $K$ distinct collision
pairs with probability bounded below by a positive constant.

Repeat the trial independently $c_2\log(2/\eta)$ times, always starting from
the instance at the beginning of the round, and accept the first verified
list.  For a sufficiently large $c_2$, the probability that no trial is
accepted is at most $\eta$.  Every trial has a fixed query bound
$\widetilde O(N^{2/3}q^{1/3})$, so the total query complexity is
\[
  \widetilde O\left(
    N^{2/3}q^{1/3}\operatorname{polylog}(1/\eta)
  \right).
\]

Finally, if $r=\lceil(N^2q)^{1/3}\rceil$, the two terms in
\eqref{eq:fixed-matching-round-cost} are
$O(N^{2/3}q^{1/3})$.  If instead $r=\lfloor N/4\rfloor$, then
$q=\Omega(N)$, and the cost in
\eqref{eq:fixed-matching-round-cost} is $O(N)$, which is again
$O(N^{2/3}q^{1/3})$.  This completes the proof.
\end{proof}

\begin{proofof}{Theorem~\ref{thm:findingcollisions}}
Start with the two non-repetitive strings $A,B$.  After reported pairs are
deleted, keep the surviving positions in their original order.  The two
remaining strings are still non-repetitive, and a query to a remaining
position is implemented by a query to its original position.

We first determine $m$ exactly without using a non-integral padding
parameter.  Append one fresh common symbol $\bot$ to both strings and call
the resulting strings $A^+,B^+$.  They are non-repetitive and have
$m^+=m+1$ collision pairs.  Apply Item~1 of
Theorem~\ref{cor:legall} to $A^+,B^+$ with lower bound $m_0=1$ and relative
error $1/4$, obtaining an estimate $g$.  Conditioned on the success of this
call,
\[
  \frac34m^+<g<\frac54m^+.
\]
Let $\bar g:=\max\{1,g\}$.  Apply Item~1 again with lower bound $m_0=1$ and
relative error $\varepsilon=1/(4\bar g)$, obtaining $\widehat m^+$.  Then
\[
  |\widehat m^+-m^+|
  <\frac{m^+}{4\bar g}
  <\frac13.
\]
Thus rounding $\widehat m^+$ to the nearest integer and subtracting one gives
$m$ exactly.  Give each counting call failure probability at most $\tau/20$.
The first call uses $\widetilde O(n^{2/3})$ queries.  Since
$\bar g=\Theta(m^+)$ whenever the first call succeeds, the second uses
\[
  \widetilde O\left(
    n^{2/3}(m+1)^{1/3}\operatorname{polylog}(1/\tau)
  \right)
\]
queries.

Set $q_0=m$.  The number of remaining pairs is then known throughout the
algorithm: after a round reports $K$ pairs, subtract $K$ from the current
value.  If $q_j=0$, stop.  Otherwise, in round $j$ apply
Lemma~\ref{lem:fixed-matching-round} to the concatenation of the two
remaining strings, with failure probability at most
\[
  \frac{\tau}{10(j+1)^2}.
\]
Every collision pair contains one position from each string because both
strings are non-repetitive.  Let $N_j$ be the total number of their remaining
positions, so $N_j\leq 2n$.

Since round $j$ removes $\lceil q_j/32\rceil$ pairs,
\[
  q_{j+1}\leq\frac{31}{32}q_j.
\]
Consequently, when $m\geq1$, there are $J=O(\log(m+1))$ rounds and
\begin{align*}
  \sum_{j=0}^{J-1}
  \widetilde O\left(N_j^{2/3}q_j^{1/3}\right)
  &\leq
  \widetilde O\left(
    n^{2/3}\sum_{j\geq0}
    \bigl(m(31/32)^j\bigr)^{1/3}
  \right)                                                     =\widetilde O\left(n^{2/3}m^{1/3}\right).
\end{align*}
When $m=0$, there are no enumeration rounds.  Including the
initial exact-counting step, the total query complexity is
\[
  \widetilde O\left(
    n^{2/3}(m+1)^{1/3}\operatorname{polylog}(1/\tau)
  \right).
\]
The sum of the assigned failure probabilities is at most $\tau$.  This proves
that all collision pairs are reported with probability at least $1-\tau$.
\end{proofof}

\subsection{Reversible access to the implicit Ulam instance}
\label{app:ak-quantum-access}

In this subsection, we explain how to use the reduction of
Andoni and Krauthgamer~\cite{andoni2012smoothed} inside the quantum procedures of Section~\ref{sec:smoothedEditD}.  The procedure given in~\cite{andoni2012smoothed} for answering queries to the resulting permutations maintains tables containing the positions queried earlier.  We instead associate a uniquely determined database with every subset appearing in our quantum walks.

Let $X,Y\in\Sigma^n$ be the original smoothed strings.  Apply Lemma~4.6 of~\cite{andoni2012smoothed}, with the perturbation parameter $p$ there equal to $\sigma$, and let $P,Q$ be the resulting non-repetitive strings.  We use the choice $ Q[j]=j\text{ for every }j\in[n],$ and write $L=\Theta\left(\frac{\log n}{\sigma}\right).$ We condition throughout this subsection on the high-probability event in that lemma.

Let $c_{\mathrm{sep}}>0$ be an absolute constant such that the second locality property of Lemma~4.6 of~\cite{andoni2012smoothed} gives edit distance at least $c_{\mathrm{sep}}\sigma L$ for every block whose starting position is at distance at least $2L$ from the selected matching block.  We choose the constant $\varepsilon$ used in that lemma sufficiently small so that $\varepsilon<c_{\mathrm{sep}}.$
This only changes the approximation factor by an absolute constant. We also fix the following rule for resolving ties: whenever more than one block minimizes the edit distance, we choose the one with the smallest starting position.

Since $P$ and $Q$ are non-repetitive, the set $M:=\{(i,j)\in[n]\times[n]:P[i]=Q[j]\}$ is a matching.  In particular, every position in either string belongs to at most one pair in $M$.

\begin{lemma}[Local collision test]
\label{lem:ak-local-collision}
There is a fixed deterministic procedure that, given $i,j\in[n]$,
computes
\[
\mathsf{Mate}(i,j)
:=
\mathbf 1[P[i]=Q[j]]
\]
using $O(L)$ queries to $X,Y$.  The corresponding operation
\[
|i,j,b\rangle
\longmapsto
|i,j,b\oplus\mathsf{Mate}(i,j)\rangle
\]
can be implemented reversibly with the same query complexity.
\end{lemma}

\begin{proof}
Define a partial matching map
\[
\mathcal A(i)
=
\begin{cases}
P[i],
&
P[i]\in[n],
\\
\bot,
&
P[i]\notin[n].
\end{cases}
\]
Since $Q[j]=j$, we have
\[
\mathsf{Mate}(i,j)=1
\quad\Longleftrightarrow\quad
\mathcal A(i)=j.
\]

The locality property established in the proof of Theorem~4.1 of~\cite{andoni2012smoothed}, and formalized through Lemma~4.6 there, states that, conditioned on its high-probability event, whether
\[
\mathcal A(i)=j
\]
holds can be determined by querying $O(L)$ entries in an $O(L)$-length neighborhood of the block of $X$ containing $i$ and an $O(L)$-length neighborhood of $Y$ around $j$.  In particular, this answer is determined by $i,j$ and these two neighborhoods, and does not depend on entries outside them. \footnote{For fixed $i,j$, initialize the query procedure of~\cite[Section~4.3]{andoni2012smoothed} with empty tables and query $Q[j]$ followed by $P[i]$.  If $i$ has offset $r$ in $X_k$, equality is possible only through the candidate block beginning at $j-r$. The separation property localizes $s_k=\operatorname{match}(X_k)$ to the queried $O(L)$-neighborhood, after which the fixed tie-breaking rule and all invalidation steps of Lemma~4.6 are applied.  Thus the two answers agree exactly when $P[i]=Q[j]$, and this fixed computation can be evaluated and uncomputed reversibly using $O(L)$ queries.}

For every $i\in[n]$, let $\mathsf{Rec}_X(i)$ contain the entries in
the corresponding neighborhood of $X$, stored in increasing order.
For every $j\in[n]$, define $\mathsf{Rec}_Y(j)$ analogously.  We
store the entire two neighborhoods, so the records are fixed
functions of $i,j,X,Y$, even if the local procedure in
\cite{andoni2012smoothed} chooses which entries to examine
adaptively.  Both records contain $O(L)$ entries.

It follows that there is a fixed deterministic function $g$ such
that
\[
\mathsf{Mate}(i,j)
=
g\bigl(i,j,\mathsf{Rec}_X(i),\mathsf{Rec}_Y(j)\bigr).
\]

To implement the predicate reversibly, first query the two records
into work registers.  Next, reversibly compute $g$, add its output
to the target qubit, and reverse the computation of $g$.  Finally,
apply the input queries in reverse to return all work registers to
zero.  This implements
\[
|i,j,b\rangle
\longmapsto
|i,j,b\oplus\mathsf{Mate}(i,j)\rangle
\]
using $O(L)$ queries to $X,Y$.
\end{proof}

\begin{lemma}[Subset database for the implicit Ulam instance]
\label{lem:ak-subset-database}
Let $S,T\subseteq[n]$.  There is a database $\mathcal D(S,T)$ satisfying the following properties.

\begin{enumerate}
    \item The database can be prepared using
 $   O\bigl((|S|+|T|)L\bigr)
$    queries to $X,Y$.

    \item Given $\mathcal D(S,T)$, the value
    $\mathsf{Mate}(i,j)$ can be computed reversibly, without further
    queries to $X,Y$, for every $i\in S$ and $j\in T$.

    \item If one stored position is removed and one new position is
    inserted, the database can be updated using $O(L)$ queries to
    $X,Y$.  The inverse update has the same query complexity.

    \item The contents and ordering of $\mathcal D(S,T)$ are uniquely
    determined by $S,T,X,Y$.  In particular, they do not depend on
    the order in which the positions were inserted.
\end{enumerate}

Moreover, the database provides labels
\[
\lambda^P_{S,T}(i),\qquad i\in S,
\]
and
\[
\lambda^Q_{S,T}(j),\qquad j\in T,
\]
such that $\lambda^P_{S,T}(i)=\lambda^Q_{S,T}(j)\quad\Longleftrightarrow\quad P[i]=Q[j].$
Reversible access to these labels requires no further queries to
$X,Y$.
\end{lemma}

\begin{proof}
For every $i\in S$, store the pair  $\bigl(i,\mathsf{Rec}_X(i)\bigr),$ and for every $j\in T$, store the pair $\bigl(j,\mathsf{Rec}_Y(j)\bigr).$ The pairs on each side are stored in increasing order of their indices. For every $j\in T$, define
\[
\mu_{S,T}(j)
:=
\begin{cases}
i,
&
\text{if there exists }i\in S
\text{ such that }\mathsf{Mate}(i,j)=1,
\\
\bot,
&
\text{otherwise}.
\end{cases}
\]
The value $i$, when it exists, is unique because $M$ is a matching.
After the local records have been loaded, all values
$\mu_{S,T}(j)$ can be computed reversibly by evaluating the function
from Lemma~\ref{lem:ak-local-collision} on the stored records.  This
may require additional reversible computation, but no further
queries to $X,Y$. The database $\mathcal D(S,T)$ consists of the ordered local records
together with the values $\mu_{S,T}(j)$.  Preparing all local records
uses
\[
O\bigl((|S|+|T|)L\bigr)
\]
queries, which proves the setup bound. We now define the labels.  For $i\in S$, let
\[
\lambda^P_{S,T}(i):=i.
\]
For $j\in T$, let
\[
\lambda^Q_{S,T}(j)
:=
\begin{cases}
\mu_{S,T}(j),
&
\mu_{S,T}(j)\neq\bot,
\\
n+j,
&
\mu_{S,T}(j)=\bot.
\end{cases}
\]
The labels on each side are all distinct.  Indeed, the left-hand
labels are the distinct indices in $S$.  On the right-hand side,
matched positions receive distinct indices in $S$, while unmatched
positions receive distinct values in
$\{n+1,\ldots,2n\}$.

For every $i\in S$ and $j\in T$, we have
\[
\lambda^P_{S,T}(i)=\lambda^Q_{S,T}(j)
\]
if and only if
\[
\mu_{S,T}(j)=i,
\]
which holds if and only if
\[
\mathsf{Mate}(i,j)=1.
\]
Thus
\[
\lambda^P_{S,T}(i)=\lambda^Q_{S,T}(j)
\quad\Longleftrightarrow\quad
P[i]=Q[j].
\]
Hence the labels preserve exactly the collisions between the
positions in $S$ and the positions in $T$, and introduce no
additional collisions.

It remains to describe an update.  Suppose first that
\[
S'=S\setminus\{u\}\cup\{v\}.
\]
Starting from $\mathcal D(S,T)$, reversibly uncompute all values
$\mu_{S,T}(j)$.  Next, apply the queries used to prepare
$\mathsf{Rec}_X(u)$ in reverse, thereby returning that record to
zero.  Replace the index $u$ by $v$, query the $O(L)$ entries needed
to prepare $\mathsf{Rec}_X(v)$, and restore increasing order.
Finally, reversibly compute all values $\mu_{S',T}(j)$ from the new
records.

All temporary registers are returned to zero.  If an auxiliary
register stores the removed and inserted positions, the update acts
as
\[
U_{\mathrm{upd}}
|S,T,\mathcal D(S,T)\rangle|u,v\rangle
=
|S',T,\mathcal D(S',T)\rangle|v,u\rangle.
\]
Applying these operations in reverse implements
$U_{\mathrm{upd}}^\dagger$.  Both directions use $O(L)$ input
queries.  Replacing one position of $T$ is handled in the same way.
If the removed and inserted positions lie on different sides, one
record is removed from one side and one record is inserted on the
other side, and the query complexity remains $O(L)$.

Because the records are stored in increasing order, the
tie-breaking rule is fixed, and all values $\mu_{S,T}(j)$ are
computed by a fixed deterministic procedure, there is exactly one
database basis state associated with each choice of $S,T,X,Y$.
Consequently, any two sequences of updates that reach the same
sets $S,T$ also produce the same database state.  The setup and
update operations therefore extend linearly to superpositions of
the subsets. Finally, access to a stored label can be implemented by a reversible
lookup.  For the left side, the operation is
\[
|S,T,\mathcal D(S,T)\rangle|0,i,z\rangle
\longmapsto
|S,T,\mathcal D(S,T)\rangle
|0,i,z\oplus\lambda^P_{S,T}(i)\rangle.
\]
For the right side, the operation is
\[
|S,T,\mathcal D(S,T)\rangle|1,j,z\rangle
\longmapsto
|S,T,\mathcal D(S,T)\rangle
|1,j,z\oplus\lambda^Q_{S,T}(j)\rangle.
\]
These lookups use only information already contained in the database
and therefore require no further queries to $X,Y$.
\end{proof}

\begin{corollary}[Application to the quantum Ulam distance algorithm]
\label{cor:ak-quantum-access}
Conditioned on the high-probability event of Lemma~4.6
of~\cite{andoni2012smoothed}, the proof of Theorem~6.3 can be executed using
quantum query access only to the original smoothed strings $X,Y$.
Relative to direct quantum query access to the non-repetitive strings
$P,Q$, its input-query complexity increases by a factor of
\[
O(L)
=
O\left(\frac{\log n}{\sigma}\right).
\]
\end{corollary}

\begin{proof}
We do not apply Theorem~\ref{thm:ulam} as a black box with a quantum value oracle for the complete strings $P,Q$.  Instead, we inspect its proof and replace every database of queried symbols by the reversible subset database from Lemma~\ref{lem:ak-subset-database}. We inspect the ways in which the proof of Theorem~6.3 accesses the non-repetitive strings.  In every case, the numerical names of the symbols are irrelevant.  The procedures use the order of the positions and test which pairs of positions contain the same
symbol.

Consider first the Johnson-graph walk in Theorem~6.5.  For a current subset of positions in the concatenation of $P$ and $Q$, let $S$ be the positions coming from $P$ and let $T$ be the positions coming from $Q$.  Replace the database of queried symbols by $\mathcal D(S,T)$.  By Lemma~\ref{lem:ak-subset-database}, the setup cost changes from $O(r)$ to $O(rL)$ and the update cost changes from $O(1)$ to $O(L)$.  The number of collisions contained in the current subset can be computed from the values $\mu_{S,T}(j)$ without any further input queries.  Thus the checking cost remains zero.

The same replacement applies to the product-Johnson-graph procedure
used in Theorem~6.10: one subset supplies $S$, the other supplies
$T$, and the marked condition is exactly the existence of a pair
$(i,j)$ satisfying $\mathsf{Mate}(i,j)=1$.

The multiple-collision procedure in Appendix~A.3 is formulated in
terms of a fixed matching of disjoint collision pairs.  In the
present setting, this matching is $M$.  Its database can therefore
be replaced by the database of
Lemma~\ref{lem:ak-subset-database}.  When a collision is extracted,
we retain its two endpoint positions.  If needed for reversing a
subsequent deletion, we also retain the two corresponding local
records.  These records are already present in the database and can
be copied without an additional query.  Hence setup and update costs
in the multiple-collision procedure also increase by at most a
factor of $O(L)$.

The collision estimators of Theorem~6.7 are obtained from
Theorem~6.5 and inherit the same overhead.  The fresh padding symbols
introduced there are known explicitly and require no query to
$X,Y$.

The procedure in Theorem~6.11 uses the strings only through
collision counts between pairs of intervals.  The Q-GapUlamTest,
the partial-alignment procedure, and the Ulam product tester use the
strings only through the collision estimators, the
multiple-collision procedure, and the pair-finding procedure
discussed above.  Restricting $P,Q$ to substrings simply restricts
the matching $M$ to the corresponding intervals, while preserving
the original position indices.

All the required setup, update, checking, and inverse operations are
therefore available with an $O(L)$ multiplicative overhead.  In the
variable-time application, if a controlled procedure originally
uses $t_s$ queries, the new implementation uses $O(Lt_s)$ queries.
Consequently,
\[
\sqrt{\sum_s O(Lt_s)^2}
=
O(L)\sqrt{\sum_s t_s^2}.
\]
Thus the total input-query complexity in the proof of
Theorem~6.3 increases by a factor of $O(L)$.
\end{proof}

\begin{remark}
The preceding argument concerns input-query complexity.  Computing
all collision relations among the local records stored in a subset
may require additional reversible elementary operations.  As in the
rest of this work, we do not analyze this additional gate
complexity.
\end{remark}
\end{document}